\documentclass[twoside]{article}

\usepackage[accepted]{aistats2021}
\usepackage[round]{natbib}

\usepackage[utf8]{inputenc} 
\usepackage[T1]{fontenc}    
\usepackage{hyperref}       
\usepackage{url}            
\usepackage{booktabs}       
\usepackage{amsfonts}       
\usepackage{nicefrac}       
\usepackage{microtype}      
\usepackage{amsmath}       
\usepackage{amsthm}       
\usepackage{amssymb}
\usepackage{mathtools}
\usepackage{blindtext}
\usepackage{centernot}
\usepackage{algorithm}
\usepackage{algorithmic}
\usepackage{subfigure}
\usepackage{svg}
\usepackage{chngcntr}
\usepackage{apptools}

\usepackage[
  open,
  openlevel=2,
  atend,
  numbered
]{bookmark}

\DeclareMathOperator*{\argmax}{arg\,max}
\DeclareMathOperator*{\argmaxe}{arg\,max-e}

\DeclarePairedDelimiterX{\infdivx}[2]{(}{)}{%
  #1\;\delimsize\|\;#2%
}
\newcommand{\KL}{D_{\mathrm{KL}}\infdivx}

\newtheorem{lemma}{Lemma}
\newtheorem{theorem}{Theorem}
\newtheorem{proposition}{Proposition}
\newtheorem{definition}{Definition}
\newtheorem{assumption}{Assumption}

\AtAppendix{\counterwithin{lemma}{subsection}}
\AtAppendix{\counterwithin{theorem}{subsection}}
\AtAppendix{\counterwithin{proposition}{subsection}}
\AtAppendix{\counterwithin{definition}{subsection}}
\AtAppendix{\counterwithin{assumption}{subsection}}
\AtAppendix{\counterwithin{example}{subsection}}

\newenvironment{deferredproof}[1][\proofname]{%
  \begin{proof}[#1]%
}{%
  \end{proof}%
}
\newenvironment{subproof}[1][\proofname]{%
  \begin{proof}[#1]%
}{%
  \end{proof}%
}

\begin{document}

%
\runningtitle{Approximately Solving Mean Field Games via Entropy-Regularized Deep Reinforcement Learning}

%

\twocolumn[

\aistatstitle{Approximately Solving Mean Field Games via \\Entropy-Regularized Deep Reinforcement Learning}

\aistatsauthor{ Kai~Cui \And Heinz~Koeppl }

\aistatsaddress{ Technische~Universität~Darmstadt\\
  \texttt{kai.cui@bcs.tu-darmstadt.de} \And Technische~Universität~Darmstadt\\
  \texttt{heinz.koeppl@bcs.tu-darmstadt.de} } ]

\begin{abstract}
  The recent mean field game (MFG) formalism facilitates otherwise intractable computation of approximate Nash equilibria in many-agent settings. In this paper, we consider discrete-time finite MFGs subject to finite-horizon objectives. We show that all discrete-time finite MFGs with non-constant fixed point operators fail to be contractive as typically assumed in existing MFG literature, barring convergence via fixed point iteration. Instead, we incorporate entropy-regularization and Boltzmann policies into the fixed point iteration. As a result, we obtain provable convergence to approximate fixed points where existing methods fail, and reach the original goal of approximate Nash equilibria. All proposed methods are evaluated with respect to their exploitability, on both instructive examples with tractable exact solutions and high-dimensional problems where exact methods become intractable. In high-dimensional scenarios, we apply established deep reinforcement learning methods and empirically combine fictitious play with our approximations.
\end{abstract}

\section{Introduction}

The framework of mean field games (MFG) was introduced independently by the seminal works of \citet{huang2006large} and \citet{lasry2007mean} in the fully continuous setting of stochastic differential games. In the meantime, it has sparked great interest and investigation both in the mathematical community, where interests lie in the theoretical properties of MFGs, and in the applied research communities as a framework for solving and analyzing large-scale multi-agent problems. 

At its core lies the idea of reducing the classical, intractable multi-agent solution concept of Nash equilibria to the interaction between a representative agent and the `mass' of infinitely many other agents -- the so-called mean field. The solution to this limiting problem is the so-called mean field equilibrium (MFE), characterized by a forward evolution equation for the agent's state distributions, and a backward optimality equation of representative agent optimality. Importantly, the MFE constitutes an approximate Nash equilibrium in the corresponding finite agent game of sufficiently many agents (\citet{huang2006large}), which would otherwise be intractable to compute (\citet{daskalakis2009complexity}).

Nonetheless, computing an MFE remains difficult in the general case. Standard assumptions in existing literature are MFE uniqueness and operator contractivity (\citet{huang2006large}, \citet{anahtarci2020value}, \citet{guo2019learning}) to obtain convergence via simple fixed point iteration. While these assumptions hold true for some games, we address the case where such restrictive assumptions fail. Applications for such mean field models are manifold and include e.g. finance (\citet{gueant2011mean}), power control (\citet{kizilkale2016collective}), wireless communication (\citet{aziz2016mean}) or public health models (\citet{laguzet2015individual}).

\paragraph{A motivating example.}
Consider the following trivial situation informally: Let a large number of agents choose simultaneously between going left ($L$) or right ($R$). Afterwards, each agent shall be punished proportional to the number of agents that chose the same action. If we had infinitely many independent, identically acting agents, the only stable solution would be to have all agents pick uniformly at random.

The MFG formalism models this problem by picking one representative agent and abstracting all other agents into their state distribution. Unfortunately, analytically obtaining fixed points in general proves difficult and existing computational methods can fail.

\paragraph{Our contribution.}
We begin by formulating the mean field analogue to finite games in game theory. In this setting we give simplified proofs for both existence and the approximate Nash equilibrium property of mean field equilibria. Moreover, we show that in finite MFGs, all non-constant fixed point operators are non-contractive, necessitating a different approach than naive fixed point iteration as in \citet{anahtarci2020value}. 

Consequently, we approximate the fixed point operator by introducing relative entropy regularization and Boltzmann policies. We prove guaranteed convergence for sufficiently high temperatures, while remaining arbitrarily exact for sufficiently low temperatures. Furthermore, repeatedly iterating on the prior policy allows us to perform an iterative descent on exploitability, successively improving the equilibrium approximation.

Finally, our methods are extensively evaluated and compared to other methods such as fictitious play (FP, see \citet{perrin2020fictitious}), which in general fail to converge to a fixed point. We outperform existing state-of-the-art methods in terms of exploitability in our problems, allowing us to find approximate mean field equilibria in the general case and paving the way to practical application of mean field games. In otherwise intractable problems, we apply deep reinforcement learning techniques together with particle-based simulations.

\section{Finite mean field games}

\subsection{Finite agent games}
Consider a discrete-time $N$-agent stochastic game with finite agent state space $\mathcal S$ and finite agent action space $\mathcal A$, equipped with the discrete metric. Let $\mathcal T = \{0, 1, \ldots, T - 1\}$ denote the time index set. Denote by $\mathcal P(\mathcal X)$ the set of all Borel probability measures on a metric space $\mathcal X$. Since we work with finite spaces, we abuse notation and denote both a measure $\nu$ and its probability mass function by $\nu(\cdot)$. For each agent, the dynamical behavior is described by the state transition function $p: \mathcal S \times \mathcal S \times \mathcal A \times \mathcal P (\mathcal S) \to [0, 1]$ and the initial state distribution $\mu_0: \mathcal S \to [0, 1]$. For agents $i = 1, \ldots, N$ at times $t \in \mathcal T$, their states $S_{t}^i$ and actions $A_{t}^i$ are random variables with values in $\mathcal S$ and $\mathcal A$ respectively. Let $\mathbb G^N_s \equiv \frac{1}{N} \sum_{i = 1}^N \delta_{s_i}$ denote the empirical measure of agent states $s = (s_1, \ldots, s_N) \in \mathcal S^N$, where $\delta$ is the Dirac measure. Consider for each agent $i$ a Markov policy $\pi^i = (\pi^i_t)_{t\in\mathcal T} \in \Pi$, where $\pi^i_t: \mathcal A \times \mathcal S \to [0, 1]$ and $\Pi$ is the space of all Markov policies. The state evolution of agent $i$ begins with $S_{0}^i \sim \mu_0$ and subsequently for all applicable times $t$ follows
\begin{align*}
    \mathbb P(A_{t}^i = a \mid S_{t}^i = s_i) &\equiv \pi^i_t(a \mid s_i) \, , \\
    \mathbb P(S_{t+1}^i = s_i' \mid S_{t} = s, A_{t}^i = a) &\equiv p(s_i' \mid s_i, a, \mathbb G^N_s) \, ,
\end{align*}
for arbitrary $s_i, s_i' \in \mathcal S$, $a \in \mathcal A$, $s = (s_1, \ldots, s_N) \in \mathcal S^N$ and $S_{t} = (S_{t}^1, \ldots, S_{t}^N)$. Finally, define agent $i$'s finite horizon objective function 
\begin{align*}
    J_i^N(\pi^1, \ldots, \pi^N) \equiv \mathbb E \left[ \sum_{t=0}^{T-1} r(S_{t}^i, A_{t}^i, \mathbb G^N_{S_t}) \right]
\end{align*} to be maximized, where $r: \mathcal S \times \mathcal A \times \mathcal P (\mathcal S) \to \mathbb R$ is the agent reward function. With this, we can give the notion of optimality used by \citet{saldi2018markov}.

\begin{definition}
A Markov-Nash equilibrium is a $0$-Markov-Nash equilibrium. For $\varepsilon \geq 0$, an $\varepsilon$-Markov-Nash equilibrium (approximate Markov-Nash equilibrium) is defined as a tuple of policies $({\pi^1}, \ldots, {\pi^N}) \in \Pi^N$ such that for any $i = 1, \ldots, N$, we have
\begin{align*}
    &J_i^N({\pi^1}, \ldots, {\pi^N}) \geq \\
    &\quad \max_{\pi \in \Pi} J_i^N({\pi^1}, \ldots, {\pi^{i-1}}, \pi, {\pi^{i+1}}, \ldots, {\pi^N}) - \varepsilon \, .
\end{align*}
\end{definition}

Since analyzing policies acting on joint state information or the state history is difficult, optimality has been restricted to the set of Markov policies $\Pi$ acting on the agent's own state. Although this may seem like a significant restriction, in the $N \to \infty$ limit, the evolution of all other agents -- the mean field -- becomes deterministic and therefore non-informative.

\subsection{Mean field games}
The $N \to \infty$ limit of the $N$-agent game constitutes its corresponding finite mean field game (i.e. with a finite state and action space). It consists of the same elements $\mathcal T, \mathcal S, \mathcal A, p, r, \mu_0$. However, instead of modeling $N$ separate agents, it models a single representative agent and collapses all other agents into their common state distribution, i.e. the mean field $\mu = ( \mu_t )_{t \in \mathcal T} \in \mathcal M$ with $\mu_t: \mathcal S \to [0, 1]$, where $\mathcal M$ is the space of all mean fields and $\mu_0$ is given. The deterministic mean field $\mu$ replaces the empirical measure of the finite game. Consider a Markov policy $\pi \in \Pi$ as before. For some fixed mean field $\mu$, the evolution of random states $S_t$ and actions $A_t$ begins with $S_0 \sim \mu_0$ and subsequently for all applicable times $t$ follows
\begin{align*} 
    \mathbb P(A_{t} = a \mid S_{t} = s) &\equiv \pi_t(a \mid s), \\
    \mathbb P(S_{t+1} = s' \mid S_{t} = s, A_{t} = a) &\equiv p(s' \mid s, a, \mu_t) \, ,
\end{align*}
and the objective analogously becomes
\begin{align*}
    J^\mu(\pi) \equiv \mathbb E \left[ \sum_{t=0}^{T-1} r(S_t, A_t, \mu_t) \right] \, .
\end{align*}
The mean field $\mu$ induced by some fixed policy $\pi$ begins with the given $\mu_0$ and is defined recursively by
\begin{align*}
    \mu_{t+1}(s') \equiv \sum_{s \in \mathcal S} \mu_t(s) \sum_{a \in \mathcal A} \pi_t(a \mid s) p(s' \mid s, a, \mu_t) \, .
\end{align*}

By fixing a mean field $\mu \in \mathcal M$, we obtain an induced Markov Decision Process (MDP) with time-dependent transition function $p(s' \mid s, a, \mu_t)$ and reward function $r(s, a, \mu_t)$. Denote the set-valued map from mean field to optimal policies $\pi$ of the induced MDP as $\hat \Phi: \mathcal M \to 2^{\Pi}$ (i.e. such that $\pi$ is optimal at any time and state). Analogously, define the map from a policy to its induced mean field as $\Psi: \Pi \to \mathcal M$. Finally, we can define the $N \to \infty$ analogue to Markov-Nash equilibria.

\begin{definition}
A mean field equilibrium (MFE) is a pair $(\pi, \mu) \in \Pi \times \mathcal M$ such that $\pi \in \hat \Phi(\mu)$ and $\mu = \Psi(\pi)$ holds.
\end{definition}

By defining any single-valued map $\Phi: \mathcal M \to \Pi$ to an optimal policy, we obtain a composition $\Gamma = \Psi \circ \Phi: \mathcal M \to \mathcal M$, henceforth MFE operator. Shown by \citet{saldi2018markov} for general Polish $\mathcal S$ and $\mathcal A$, the MFE exists and constitutes an approximate Markov-Nash equilibrium for sufficiently many agents under technical conditions. In the Appendix, we give simplified proofs for finite MFGs under the following standard assumption.

\begin{assumption} \label{assumption}
The functions $r(s,a,\mu_t)$ and $p(s' \mid s,a,\mu_t)$ are continuous, therefore bounded. 
\end{assumption}

Note that we metrize probability measure spaces $\mathcal P(\mathcal X)$ with the total variation distance $d_{TV}$. For probability measures $\nu, \nu'$ on finite spaces $\mathcal X$, $d_{TV}$ simplifies to
\begin{align*}
    d_{TV}(\nu, \nu') = \frac{1}{2} \sum_{x \in \mathcal X} | \nu(x) - \nu'(x) | \, .
\end{align*}
Accordingly, we equip $\Pi, \mathcal M$ with sup metrics, i.e. for policies $\pi, \pi' \in \Pi$ and mean fields $\mu, \mu' \in \mathcal M$ we define the metric spaces $(\Pi, d_{\Pi})$ and $(\mathcal M, d_{\mathcal M})$ with
\begin{align*}
    d_{\Pi}(\pi, \pi') &\equiv \max_{t \in \mathcal T} \max_{s \in \mathcal S} d_{TV}(\pi_t(\cdot \mid s), \pi'_t(\cdot \mid s)) \, , \\
    d_{\mathcal M}(\mu, \mu') &\equiv \max_{t \in \mathcal T} d_{TV}(\mu_t, \mu'_t) \, .
\end{align*}

\begin{proposition} \label{prop:kakutani}
Under Assumption~\ref{assumption}, there exists at least one MFE $(\pi^*, \mu^*) \in \Pi \times \mathcal M$.
\end{proposition}
\begin{deferredproof}
See Appendix.
\end{deferredproof}

\begin{theorem} \label{th:epsopt}
Under Assumption~\ref{assumption}, if $(\pi^*, \mu^*)$ is an MFE, then for any $\varepsilon > 0$ there exists $N' \in \mathbb N$ such that for all $N > N'$, the policy $(\pi^*, \ldots, \pi^*)$ is an $\varepsilon$-Markov-Nash equilibrium in the $N$-agent game.
\end{theorem}
\begin{deferredproof}
See Appendix.
\end{deferredproof}

Importantly, finding Nash equilibria in large-$N$ games is hard (\citet{daskalakis2009complexity}), whereas an MFE can be significantly more tractable to compute. Accordingly, solving the limiting MFG approximately solves the finite-$N$ game for large $N$ in a tractable manner.

\section{Exact fixed point iteration} 

Repeated application of the MFE operator constitutes the exact fixed point iteration approach to finding MFE. The standard assumption for convergence in the literature is contractivity and thereby MFE uniqueness (e.g. \citet{caines2019graphon, guo2019learning}).

\begin{proposition} \label{th:banach}
Let $\Phi, \Psi$ be Lipschitz with constants $c_1, c_2$, fulfilling $c_1 c_2 < 1$. Then, the fixed point iteration $\mu^{n+1} = \Psi(\Phi(\mu^{n}))$ converges to the mean field of the unique MFE for any initial $\mu^0 \in \mathcal M$.
\end{proposition}
\begin{proof}
Let $\mu, \mu' \in \mathcal M$ arbitrary, then
\begin{align*}
    d_{\mathcal M}(\Gamma(\mu), \Gamma(\mu')) 
    &= d_{\mathcal M}(\Psi(\Phi(\mu)), \Psi(\Phi(\mu'))) \\
    &\leq c_2 \cdot d_{\Pi}(\Phi(\mu), \Phi(\mu')) \\
    &\leq c_2 \cdot c_1 \cdot d_{\mathcal M}(\mu, \mu') \, .
\end{align*}
Since $\mu, \mu'$ are arbitrary, $\Gamma$ is Lipschitz with constant $c_1 \cdot c_2 < 1$. $(\Pi, d_{\Pi})$ and $(\mathcal M, d_{\mathcal M})$ are complete metric spaces (see Appendix). Therefore, Banach's fixed point theorem implies convergence to the unique fixed point for any starting $\mu^0 \in \mathcal M$.
\end{proof}

Unfortunately, it remains unclear how to proceed if multiple optimal policies of an induced MDP exist, or if contractivity fails, e.g. when multiple MFE exist. In the following, consider again the illuminating example from the introduction.

\subsection{Toy example}
Consider $\mathcal S = \{C, L, R\}$, $\mathcal A = \mathcal S \setminus \{C\}$, $\mu_0(C) = 1$, $r(s,a,\mu_t) = - \mathbf 1_{\{L\}}(s) \cdot \mu_t(L) - \mathbf 1_{\{R\}}(s) \cdot \mu_t(R)$ and $\mathcal T = \{0, 1\}$. The transition function allows picking the next state directly, i.e. for all $s,s' \in \mathcal S, a \in \mathcal A$,
\begin{align*}
    \mathbb P(S_{t+1} = s' \mid S_{t} = s, A_t = a) &= \mathbf 1_{\{s'\}}(a) \, .
\end{align*}
Clearly, any MFE $(\pi^*, \mu^*)$ must fulfill $\pi_0^*(L \mid C) = \pi_0^*(R \mid C) = 1/2$, while $\pi_1^*$ can be arbitrary. Even if the operator $\Phi$ chooses suitable optimal policies, the fixed point operator $\Gamma$ remains non-contractive, as the mean field will necessarily alternate between left and right for any non-uniform starting $\mu^0 \in \mathcal M$. 

We observe that the example has infinitely many MFE, but no deterministic MFE, i.e. an MFE such that for all $t \in \mathcal T, s \in \mathcal S, a \in \mathcal A$ either $\pi_t(a \mid s) = 0$ or $\pi_t(a \mid s) = 1$ holds, similar to the classical game-theoretical insight of mixed Nash equilibrium existence (cf. \citet{fudenberg1991game}). Therefore, choosing optimal, deterministic policies will typically fail.

Most existing work assumes contractivity, which is too restrictive. In many scenarios, agents need to "coordinate" with each other. For example, a herd of hunting animals may collectively choose one of multiple hunting grounds, allowing for multiple MFEs. Hence, it can be difficult to apply existing MFG methodologies in practice, as many problems automatically fail contractivity.

\subsection{General non-contractivity}
From the previous example, we may be led to believe that non-contractivity is a general property of finite MFGs. And indeed, regardless of number of MFEs, it turns out that in any finite MFG with non-constant MFE operator, a policy selection operator $\Phi$ with finite image $\Pi_{\Phi}$ will lead to non-contractivity. Note that this includes both the conventional $\argmax$ and the $\argmaxe$ (cf. \citet{guo2019learning}) choice of actions.

\begin{theorem} \label{th:int}
Let the image of $\Phi$ be a finite set $\Pi_{\Phi} \subseteq \Pi$. Then, either it holds that $\Gamma = \Psi \circ \Phi$ is a constant, or $\Gamma$ is not Lipschitz continuous and thus not a contraction.
\end{theorem}
\begin{deferredproof}
See Appendix.
\end{deferredproof}

Therefore, typical discrete-time finite MFGs have non-contractive fixed point operators and we must change our approach. Note that although non-contractivity does not imply non-convergence, the trivial example from before strongly suggests that non-convergence is the case for many finite MFGs.

\section{Approximate mean field equilibria}

Exact fixed point iteration fails to solve most finite MFGs. Therefore, a different solution approach is necessary. In the following, we present two related approaches that guarantee convergence while plausibly remaining approximate Nash equilibria in the finite-$N$ case. For our results, we require a stronger Lipschitz assumption that implies Assumption~\ref{assumption}.
\begin{assumption} \label{assumption2}
The functions $r(s,a,\mu_t)$ and $p(s' \mid s,a,\mu_t)$ are Lipschitz continuous, therefore bounded.
\end{assumption}

\subsection{Relative entropy mean field games}
A straightforward idea is regularization by replacing the objective by the well-known (see e.g. \citet{abdolmaleki2018maximum}) relative entropy objective  
\begin{align*}
    \tilde J^\mu(\pi) \equiv \mathbb E \left[ \sum_{t=0}^{T-1} r(S_t, A_t, \mu_t) - \eta \log \frac{ \pi_t(A_t \mid S_t) } { q_t(A_t \mid S_t) } \right]
\end{align*} 
with temperature $\eta > 0$ and positive prior policy $q \in \Pi$, i.e. $q_t(a \mid s) > 0$ for all $t \in \mathcal T, s \in \mathcal S, a \in \mathcal A$. Shown in the Appendix, the unique optimal policy $\tilde \pi_t^{\mu,\eta}$ fulfills
\begin{align*} 
    \tilde \pi_t^{\mu,\eta}(a \mid s) 
    &= \frac{ q_t(a \mid s) \exp \left( \frac{\tilde Q_\eta(\mu, t, s, a)}{\eta} \right) }{ \sum_{a' \in \mathcal A} q_t(a' \mid s) \exp \left( \frac{\tilde Q_\eta(\mu, t, s, a')}{\eta} \right) }
\end{align*}
for the MDP induced by fixed $\mu \in \mathcal M$, with the soft action-value function $\tilde Q_\eta(\mu, t, s, a)$ given by the smooth-maximum Bellman recursion
\begin{multline*}
    \tilde Q_\eta(\mu, t, s, a) = r(s, a, \mu_t) + \sum_{s' \in \mathcal S} p(s'\mid s, a, \mu_t) \\
    \cdot \eta \log \left( \sum_{a' \in \mathcal A} q_{t+1}(a' \mid s') \exp \frac{\tilde Q_\eta(\mu, t+1, s', a')}{\eta} \right)
\end{multline*} 
of the MDP induced by fixed $\mu \in \mathcal M$, with terminal condition $\tilde Q_\eta(\mu,T-1,s,a) \equiv r(s, a, \mu_{T-1})$. Note that we recover optimality as $\eta \to 0$, see Theorem \ref{th:epsepsopt}. Define the relative entropy MFE operator $\tilde \Gamma_\eta \equiv \Psi \circ \tilde \Phi_\eta$ with policy selection $\tilde \Phi_\eta(\mu) \equiv \tilde \pi^{\mu,\eta}$ for all $\mu \in \mathcal M$.

\begin{definition}
An $\eta$-relative entropy mean field equilibrium ($\eta$-RelEnt MFE) for some positive prior policy $q \in \Pi$ is a pair $(\pi^E, \mu^E) \in \Pi \times \mathcal M$ such that $\pi^E = \tilde \Phi_\eta(\mu^E)$ and $\mu^E = \Psi(\pi^E)$ hold. An $\eta$-maximum entropy mean field equilibrium ($\eta$-MaxEnt MFE) is an $\eta$-RelEnt MFE with uniform prior policy $q$.
\end{definition}

RelEnt MFE are guaranteed to exist for any $\eta > 0$ by Proposition~\ref{prop:brouwer}. Furthermore, convergence to the regularized solution is guaranteed for large $\eta$ by Theorem~\ref{th:contractive}.

\subsection{Boltzmann iteration}
Since only deterministic policies fail, a derivative approach is to use softmax policies directly with the unregularized action-value function, also called Boltzmann policies. Assume that the action-value function $Q^*$ fulfilling the Bellman equation
\begin{multline*}
    Q^*(\mu, t,s,a) = r(s,a,\mu_t) + \sum_{s' \in \mathcal S} p(s'\mid s, a, \mu_t) \\
    \cdot \max_{a' \in \mathcal A} Q^*(\mu, t+1,s',a') \, .
\end{multline*}
of the MDP induced by fixed $\mu \in \mathcal M$ with terminal condition $Q^*(\mu, T-1, s, a) \equiv r(s,a,\mu_{T-1})$ is known. Define the map $\Phi_\eta(\mu) \equiv \pi^{\mu,\eta}$ for any $\mu \in \mathcal M$, where
\begin{align*}
    \pi_t^{\mu,\eta}(a \mid s) 
    &\equiv \frac{ q_t(a \mid s) \exp \left( \frac{Q^*(\mu, t, s, a)}{\eta} \right) }{ \sum_{a' \in \mathcal A} q_t(a' \mid s) \exp \left( \frac{Q^*(\mu, t, s, a')}{\eta} \right) }
\end{align*}
for all $t \in \mathcal T, s \in \mathcal S, a \in \mathcal A$ and temperature $\eta > 0$.

\begin{definition}
An $\eta$-Boltzmann mean field equilibrium ($\eta$-Boltzmann MFE) for some positive prior policy $q \in \Pi$ is a pair $(\pi^B, \mu^B) \in \Pi \times \mathcal M$ such that $\pi^B = \Phi_\eta(\mu^B)$ and $\mu^B = \Psi(\pi^B)$ hold.
\end{definition}

\subsection{Theoretical properties}

Both $\eta$-RelEnt MFE and $\eta$-Boltzmann MFE are guaranteed to exist for any temperature $\eta > 0$.

\begin{proposition} \label{prop:brouwer}
Under Assumption~\ref{assumption}, $\eta$-Boltzmann and $\eta$-RelEnt MFE exist for any temperature $\eta > 0$.
\end{proposition}
\begin{deferredproof}
See Appendix.
\end{deferredproof}

Contractivity of both $\eta$-Boltzmann MFE operator $\Gamma_\eta \equiv \Psi \circ \Phi_\eta$ and $\eta$-RelEnt MFE operator $\tilde \Gamma_\eta \equiv \Psi \circ \tilde \Phi_\eta$ is guaranteed for sufficiently high temperatures, even if all possible original $\Phi$ are not Lipschitz continuous.

\begin{theorem} \label{th:contractive}
Under Assumption~\ref{assumption2}, $\mu \mapsto Q^*(\mu,t,s,a)$, $\mu \mapsto \tilde Q_\eta(\mu,t,s,a)$ and $\Psi(\pi)$ are Lipschitz continuous with constants $K_{Q^*}$, $K_{\tilde Q}$ and $K_{\Psi}$ for arbitrary $t \in \mathcal T, s \in \mathcal S, a \in \mathcal A, \eta > \eta', \eta' > 0$. Furthermore, $\Gamma_\eta$ and $\tilde \Gamma_\eta$ are a contraction for
\begin{align*}
    \eta > \max \left( \eta', \frac{\left| \mathcal A \right| (\left| \mathcal A \right| - 1) K_{Q} K_{\Psi} q_{\mathrm{max}}^2}{ 2 q_{\mathrm{min}}^2} \right)
\end{align*} 
where $K_Q = K_{Q^*}$ for $\Gamma_\eta$, $K_Q = K_{\tilde Q}$ for $\tilde \Gamma_\eta$, $q_{\mathrm{max}} \equiv \max_{t \in \mathcal T, s \in \mathcal S, a \in \mathcal A} q_t(a \mid s) > 0$ and $q_{\mathrm{min}} \equiv \min_{t \in \mathcal T, s \in \mathcal S, a \in \mathcal A} q_t(a \mid s) > 0$.
\end{theorem}
\begin{deferredproof}
See Appendix.
\end{deferredproof}

Sufficiently large $\eta$ hence implies convergence via fixed point iteration. On the other hand, for sufficiently low temperatures $\eta$, both $\eta$-Boltzmann and $\eta$-RelEnt MFE will also constitute an approximate Markov-Nash equilibrium of the finite-$N$ game.

\begin{theorem} \label{th:epsepsopt}
Under Assumption~\ref{assumption2}, if $(\pi^*_{n}, \mu^*_{n})_{n \in \mathbb N}$ is a sequence of ${\eta_n}$-Boltzmann or ${\eta_n}$-RelEnt MFE with $\eta_n \to 0$, then for any $\varepsilon > 0$ there exist $n', N' \in \mathbb N$ such that for all $n > n', N > N'$, the policy $(\pi^*_{n}, \ldots, \pi^*_{n}) \in \Pi^N$ is an $\varepsilon$-Markov-Nash equilibrium of the $N$-agent game, i.e.
\begin{align*}
    &J_i^N(\pi^*_{n}, \ldots, \pi^*_{n}) \geq \\
    &\quad \max_{\pi_i \in \Pi} J_i^N(\pi^*_{n}, \ldots, \pi^*_{n}, \pi_i, \pi^*_{n}, \ldots, \pi^*_{n}) - \varepsilon \, .
\end{align*}
\end{theorem}
\begin{deferredproof}
See Appendix.
\end{deferredproof}

If we can obtain contractivity for sufficiently low $\eta$, we can find good approximate Markov-Nash equilibria. As it is impossible to have both $\eta \to 0$ and $\eta \to \infty$, it depends on the problem and prior whether we can converge to a good solution. Nonetheless, we find that it is often possible to empirically find low $\eta$ that provide convergence as well as a good approximate MFE.

\subsection{Prior descent}
In principle, we can insert arbitrary prior policies $q \in \Pi$. Under Assumption~\ref{assumption}, by boundedness of both $\tilde Q_\eta$ and $Q^*$ (see Appendix), both $\eta$-RelEnt and $\eta$-Boltzmann MFE policies converge to the prior policy as $\eta \to \infty$. Therefore, in principle we can show that for any $\varepsilon > 0$, for sufficiently large $\eta$ and $N$, the $\eta$-RelEnt and $\eta$-Boltzmann MFE under $q$ will be at most an $\varepsilon$-worse approximate Nash equilibrium than the prior policy. Furthermore, we obtain guaranteed contractivity by Theorem~\ref{th:contractive}. Thus, any prior policy gives a worst-case bound on the performance achievable over all $\eta > 0$. On the other hand, if we obtain better results for sufficiently low $\eta$, we may iteratively improve our policy and thus our equilibrium quality.

\section{Related work}

The original work of \citet{huang2006large} introduces contractivity and uniqueness assumptions into the continuous MFG setting. Analogously, \citet{guo2019learning} and \citet{caines2019graphon} assume contractivity for discrete-time MFGs and dense graph limit MFGs respectively. Further existing work on discrete-time MFGs similarly assumes uniqueness of the MFE, which includes \citet{saldi2018markov} and \citet{gomes2010discrete} for approximate optimality and existence results, and \citet{anahtarci2020value} for an analysis on contractivity requirements. \citet{mguni2018decentralised} solve discrete-time continuous state MFG problems under the classical uniqueness conditions of \citet{lasry2007mean}. Further extensions of the MFG formula include partial observability (\citet{saldi2019approximate}) or major agents (\citet{nourian2013e}).

The work of \citet{anahtarci2020q} is related and studies theoretical properties of finite-$N$ regularized games and their limiting MFG. In their work, the existence and approximate Nash property of MFE in stationary regularized games is shown, and Q-Learning error propagation is investigated. In comparison, we consider the original, unregularized finite-$N$ game in a transient setting and perform extensive empirical evaluations. \citet{guo2019learning} and \citet{yang2018mean} previously proposed to apply Boltzmann policies. The former applies the approximation heuristically, while the latter focuses on directly solving finite-$N$ games.

An orthogonal approach to computing MFE is fictitious play. Rooted in game-theory and classical economic works (\citet{brown1951iterative}), it has since been adapted to MFGs. In fictitious play, all past mean fields (\cite{cardaliaguet2017learning}) and policies (\cite{perrin2020fictitious}) are averaged to produce a new mean field or policy. Importantly, convergence is guaranteed in certain special cases only (cf. \cite{elie2019approximate}). Although introduced in a differentiable setting, we evaluate fictitious play empirically in our setting and find that both our regularization and fictitious play may be combined successfully.

\begin{figure*}[t]
	\centering
    \includegraphics[width=\textwidth]{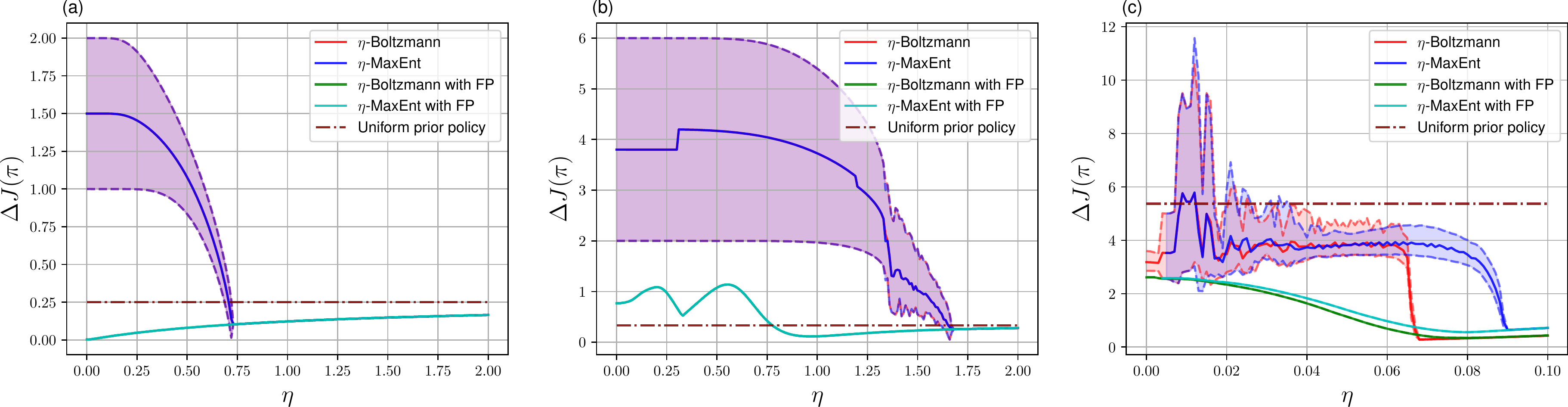}
    \caption{Mean exploitability over the final 10 iterations. Dashed lines represent maximum and minimum over the final 10 iterations. (a) LR, 10000 iterations; (b) RPS, 10000 iterations; (c) SIS, 10000 iterations. Maximum entropy (MaxEnt) results begin at higher temperatures due to limited floating point accuracy. Temperature zero depicts the exact fixed point iteration for both $\eta$-MaxEnt and $\eta$-Boltzmann MFE. In LR and RPS, $\eta$-MaxEnt and $\eta$-Boltzmann MFE coincide both with and without fictitious play (FP), here averaging both policy and mean field over all past iterations. The exploitability of the prior policy is indicated by the dashed horizontal line. } \label{fig:exact}
\end{figure*}

\section{Evaluation} \label{sec:eval}

In practice, we find that our approaches are capable of generating solutions of lower exploitability than otherwise obtained. Unless stated otherwise, we compute everything exactly, use the maximum entropy objective (MaxEnt) with the uniform prior policy $q$ where $q_t(a \mid s) = 1/|\mathcal A|$ for all $t \in \mathcal T, s \in \mathcal S, a \in \mathcal A$, and initialize with $\mu^0 = \Psi(q)$ generated by $q$. As the main evaluation metric, we define the exploitability of a policy $\pi \in \Pi$ with induced mean field $\mu \equiv \Psi(\pi)$ as
\begin{align*}
    \Delta J(\pi) \equiv \max_{\pi^*} J^{\mu}(\pi^*) - J^{\mu}(\pi) \ .
\end{align*}
Clearly, the exploitability of $\pi$ is zero if and only if $(\pi, \mu)$ is an MFE. Indeed, for any $\varepsilon > 0$, any policy $\pi \in \Pi$ is a $(\Delta J(\pi)+\varepsilon)$-Markov Nash equilibrium if $N$ sufficiently large, i.e. the exploitability translates directly to the limiting equilibrium quality in the finite-$N$ game, see also Theorem~\ref{th:epsepsopt} and its proof. 

We evaluate the algorithms on the LR, RPS, SIS and Taxi problems, ordered in increasing complexity. Details of the algorithms, hyperparameters, problems and experiment configurations as well as further experimental results can be found in the Appendix.

\subsection{Exploitability}

In Figure~\ref{fig:exact}, we plot the minimum, maximum and mean exploitability for varying temperatures $\eta$ during the last 10 fixed point iterations, i.e. a single value when the exploitability (and usually mean field) converges. Observe that the lowest convergent temperature outperforms not only the exact fixed point iteration (drawn at temperature zero), but also the uniform prior policy. 

Although developed for a different setting, we also show results of fictitious play similar to the version from \citet{perrin2020fictitious}, i.e. both policies and mean fields are averaged over all past iterations. It can be seen that fictitious play only converges to the optimal solution in the LR problem. In the other examples, supplementing fictitious play with entropy regularization is effective at producing better results. A non-existent fictitious play variant averaging only the policies finds the exact MFE in RPS, but nevertheless fails in SIS. See the Appendix for further results.

Evaluating and solving finite-$N$ games is highly intractable by the curse of dimensionality, as the local state is no longer sufficient to perform dynamic programming in the presence of the random empirical state measure. Since it has already been proven that the exploitability for $N \to \infty$ will converge to the exploitability of the corresponding mean field game, we refrain from evaluating on finite-$N$ games.

Note that the plots are entirely deterministic and not stochastic as it would seem at first glance, since the depicted shaded area visualizes the non-convergence of exploitability and is a result of the fixed point updates running into a limit cycle (cf. Figure~\ref{fig:convergence}). 

\subsection{Convergence}
\begin{figure*}[t]
	\centering
    \includegraphics[width=0.9\textwidth]{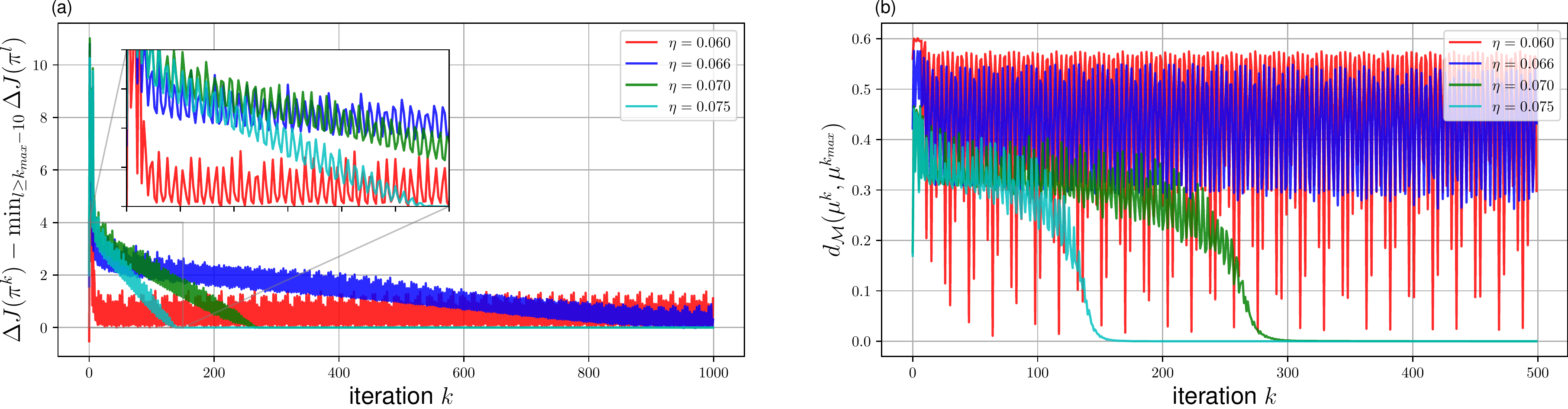}
    \caption{(a) Difference between current and final minimum exploitability over the last 10 iterations; (b) Distance between current and final mean field. Plotted for the $\eta$-Boltzmann MFE iterations in SIS for different indicated temperature settings. Note the periodicity of the lowest temperature setting, indicating a limit cycle. } \label{fig:convergence}
\end{figure*}

In Figure~\ref{fig:convergence}, the difference between the exploitability of the current policy and the minimal exploitability reached during the final 10 iterations is shown for $\eta$-Boltzmann MFE. As the temperature $\eta$ decreases, time to convergence increases until non-convergence is reached in form of a limit cycle. Analogous results for $\eta$-RelEnt MFE can be found in the Appendix. 

Note also that in LR, we can analytically find $K_Q = 1$ and $K_{\Psi} = 1$. Thus, we obtain guaranteed convergence via $\eta$-Boltzmann MFE iteration if $\eta > 1$. In Figure~\ref{fig:exact}, we see convergence already for $\eta \geq 0.7$. Note further that the non-converged regime can allow for lower exploitability. However, it is unclear a priori when to stop, and for approximate solutions where DQN is used for evaluation, the evaluation of exploitability may become inaccurate. 

\begin{figure*}[ht]
	\centering
    \includegraphics[width=1\textwidth]{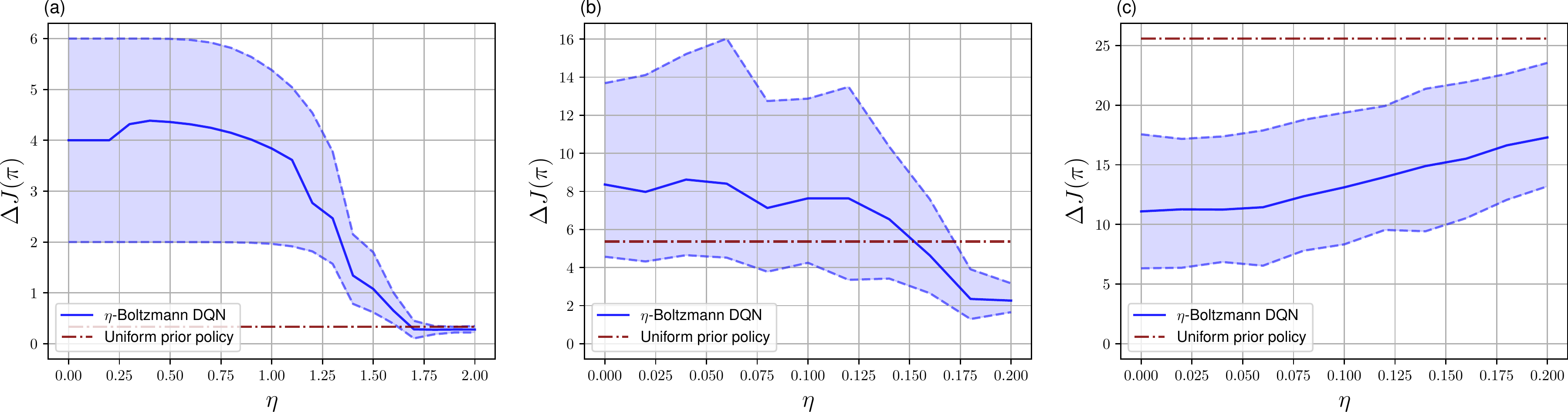}
    \caption{Mean exploitability over the final 5 iterations using DQN, averaged over 5 seeds. Dashed lines represent the averaged maximum and minimum exploitability over the last 5 iterations. (a) RPS, 1000 iterations; (b) SIS, 50 iterations; (c) Taxi, 15 iterations. Evaluation of exploitability is exact except in Taxi, which uses DQN and averages over 1000 episodes. The point of zero temperature depicts fixed point iteration using exact DQN policies. } \label{fig:approximate}
\end{figure*}
\begin{figure}[t]
	\centering
    \includegraphics[width=0.45\textwidth]{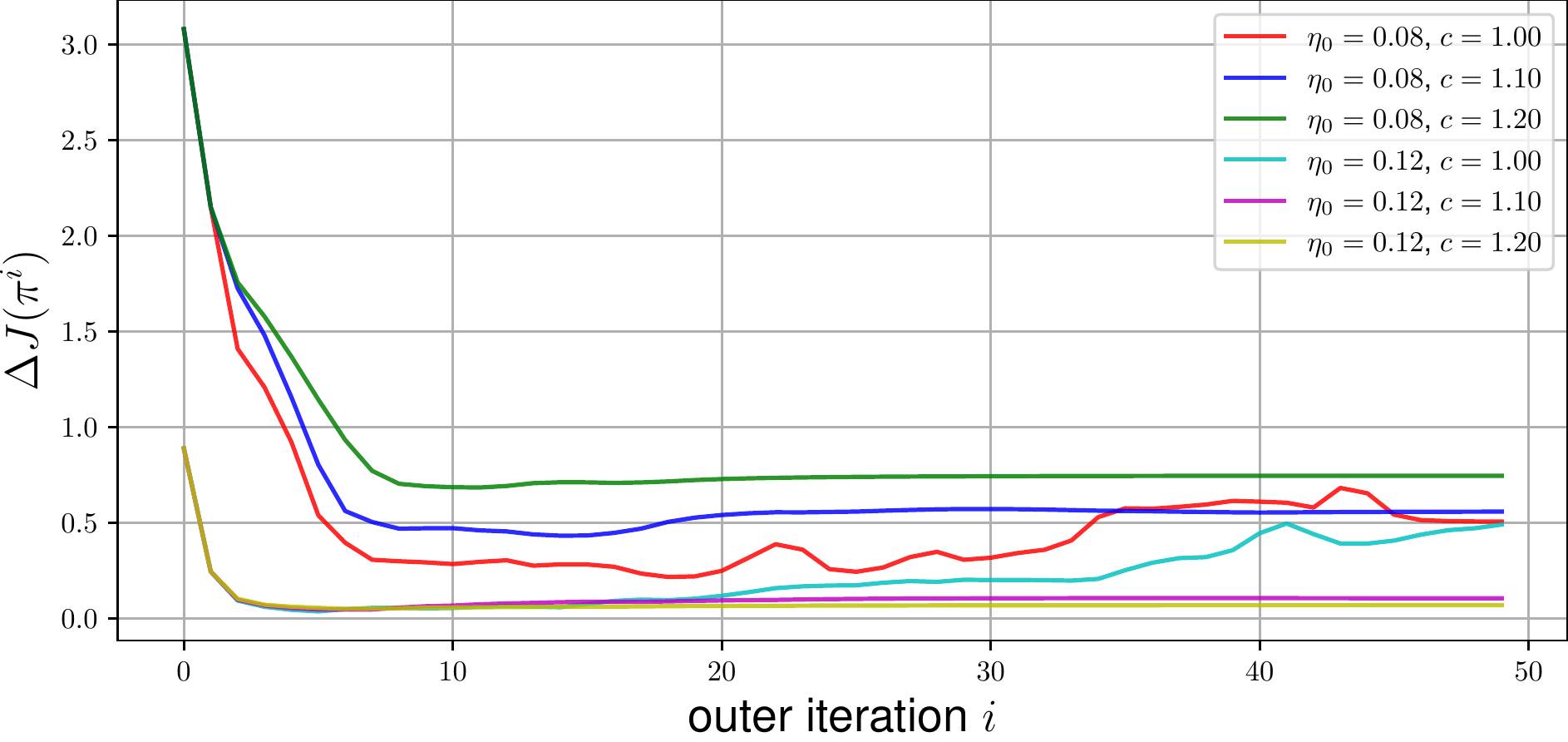}
    \caption{Exploitability over outer iterations in SIS, using 100 $\eta$-RelEnt MFE iterations per outer iteration. Note that the results are deterministic. Not shown: Running the fixed temperature settings $c = 1$ for longer does not converge for at least $1000$ iterations. } \label{fig:rep}
\end{figure}

\subsection{Deep reinforcement learning}
For problems with intractably large state spaces, we adopt the DQN algorithm (\citet{mnih2013playing}), using the implementation of \citet{cleanrl} as a base.  Particle-based simulations are used for the mean field, and stochastic performance evaluation on the induced MDP is performed (see Appendix). Note that the approximation introduces three sources of stochasticity into the otherwise deterministic algorithms, i.e. stochastic evaluation, mean field simulation and DQN. To counteract the randomness, we average our results over multiple runs. The hyperparameters and architectures used are standard and can be found in the Appendix.

Fitting the soft action-value function directly using a network is numerically problematic, as the log-exponential transformation of approximated action-values quickly fails due to limited floating point accuracy. Thus, we limit ourselves to the classical Bellman equation with Boltzmann policies only. 

In Figure~\ref{fig:approximate}, we evaluate the exploitability of Boltzmann DQN iteration, evaluated exactly in SIS and RPS, and stochastically in Taxi over 2000 realizations. Minimum, maximum and mean exploitability are taken over the final 5 iterations and averaged over 5 seeds. Note that it is very time-consuming to solve a full reinforcement learning problem using DQN repeatedly in every iteration. Nonetheless, we observe that a temperature larger than zero appears to improve exploitability and convergence in the SIS example. Both due to the noisy nature of approximate solutions and the lower number of iterations, it can be seen that a higher temperature is required to converge than in the exact case. 

In the intractable Taxi environment, the policy oscillates between two modes as in exact LR, and regularization fails to obtain better results, see also the Appendix. An important reason is that the prior policy performs extremely bad (exploitability of $\sim 35$) as most states require specific actions for optimality. Hence we cannot find an $\eta > 0$ for which the algorithm both converges and performs well. Using prior descent and iteratively refining a better prior policy would likely increase performance, but is deferred to future investigations as the required computations grow very large.

Fictitious play is expensive in combination with approximate Q-Learning and particle simulations, as policies and particles of past iterations must be kept to perform exact fictitious play. For this reason, we do not attempt approximate fictitious play with approximate solution methods. In theory, supervised learning for fitting summarizing policies and randomly sampling particles may help, but is out of scope of this paper.

\subsection{Prior descent}
In Figure~\ref{fig:rep}, we repeatedly perform outer iterations consisting of 100 $\eta$-RelEnt MFE iterations each with the indicated fixed temperature parameters in SIS. After each outer iteration, the prior policy is updated to the newest resulting policy. Note again that the results are entirely deterministic. 

Searching for a suitable $\eta$ dynamically every iteration would keep the exploitability from increasing, as for $\eta \to \infty$ we obtain the original prior policy. Since it is expensive to scan over all temperatures in each outer iteration, we use a heuristic. Intuitively, since the prior will become increasingly good, it will be increasingly difficult to obtain a better policy. Thus, increasing the temperature will help sticking close to the prior and converge. Consequently, we use the simple heuristic
\begin{align*}
    \eta_{i+1} = \eta_i \cdot c
\end{align*}
for each outer iteration $i$, where $c \geq 1$ adjusts the temperature after each outer iteration.
 
Importantly, even for our simple heuristic, prior descent already achieves an exploitability of $\sim 0.068$, whereas the best results for the fixed uniform policy from Figure~\ref{fig:exact} show an optimal mean exploitability of $\sim 0.281$. Furthermore, repeated prior policy updates succeed in computing the exact MFE in RPS and LR under a fixed temperature (see Appendix). 

Note that prior descent creates a double loop around solving the optimal control problem, becoming highly expensive under deep reinforcement learning. Hence, we refrain from prior descent with DQN. Automatically adjusting temperatures to monotonically improve exploitability is left for potential future work.

\section{Conclusion}
In this work, we have investigated the necessity and feasibility of approximate MFG solution approaches -- entropy regularization, Boltzmann policies and prior descent -- in the context of finite MFGs. We have shown that the finite MFG case typically cannot be solved by exact fixed point iteration or fictitious play alone. Entropy regularization and Boltzmann policies in combination with deep reinforcement learning may enable feasible computation of approximate MFE. We believe that lifting the restriction of inherent contractivity is an important step in ensuring applicability of MFG models in practical problems. We hope that entropy regularization and the insight for finite MFGs can help transfer the MFG formalism from its so-far mostly theory-focused context into real world application scenarios. Nonetheless, there still remain many restrictions to the applicability of the MFG formalism.

For future work, an efficient, automatic temperature adjustment for prior descent could be fruitful. Furthermore, it would be interesting to generalize relative entropy MFGs to infinite horizon discounted problems, continuous time, and continuous state and action spaces. Moreover, it could be of interest to investigate theoretical properties of fictitious play in finite MFGs in combination with entropy regularization. For non-Lipschitz mappings from policy to induced mean field, the proposed approach does not provide a solution. It could nonetheless be important to consider problems with threshold-type dynamics and rewards, e.g. majority vote problems. Most notably, the current formalism precludes common noise entirely, i.e. any games with common observations. In practice, many problems will allow for some type of common observation between agents, leading to non-independent agent distributions and stochastic as opposed to deterministic mean fields.

\subsubsection*{Acknowledgements}
This work has been funded by the LOEWE initiative (Hesse, Germany) within the emergenCITY center. The authors acknowledge the Lichtenberg high performance computing cluster of the TU Darmstadt for providing computational facilities for the calculations of this research.

\bibliographystyle{unsrtnat}
\bibliography{main}

\onecolumn

\appendix
\allowdisplaybreaks
\section{Experimental Details} \label{app:impl}
\subsection{Algorithms}

\begin{algorithm}[ht]
    \caption{\textbf{Exact fixed point iteration}}
    \label{alg:exact}
    \begin{algorithmic}[1]
        \STATE Initialize $\mu^0 = \Psi(q)$ as the mean field induced by the uniformly random policy $q$.
        \FOR {$k=0, 1, \cdots$}
        \STATE Compute the Q-function $Q^*(\mu^k,t,s,a)$ for fixed $\mu^k$.
        \STATE Choose $\pi^k \in \Pi$ such that $\pi^k_t(a \mid s) \implies a \in \argmax_{a \in \mathcal A} Q^k(\mu^k,t,s,a)$ for all $t \in \mathcal T, s \in \mathcal S, a \in \mathcal A$ by putting all probability mass on the first optimal action, or evenly on all optimal actions.
        \STATE \textbf{Optionally}: Overwrite $\pi^k \leftarrow \frac{1}{k+1} \pi^k + \frac{k}{k+1} \pi^{k-1}$. (FP averaged policy)
        \STATE Compute the mean field $\mu^{k+1} = \Psi(\pi^k)$ induced by $\pi^k$.
        \STATE \textbf{Optionally}: Overwrite $\mu^{k+1} \leftarrow \frac{1}{k+1} \mu^{k+1} + \frac{k}{k+1} \mu^{k}$. (FP averaged mean field)
        \ENDFOR
    \end{algorithmic}
\end{algorithm}

\begin{algorithm}[ht]
    \caption{\textbf{Boltzmann / RelEnt iteration}}
    \label{alg:boltzmann}
    \begin{algorithmic}[1]
        \STATE \textbf{Input}: Temperature $\eta > 0$, prior policy $q \in \Pi$.
        \STATE Initialize $\mu^0 = \Psi(q)$ as the mean field induced by $q$.
        \FOR {$k=0, 1, \cdots$}
        \STATE Compute the Q-function (Boltzmann) or soft Q-function (RelEnt) $Q(\mu^k,t,s,a)$ for fixed $\mu^k$.
        \STATE Define $\pi^k$ by $\pi^k_t(a \mid s) = \frac{ q_t(a \mid s) \exp \left( \frac{Q(\mu^k, t, s, a)}{\eta} \right) }{ \sum_{a' \in \mathcal A} q_t(a' \mid s) \exp \left( \frac{Q(\mu^k, t, s, a')}{\eta} \right) } $ for all $t \in \mathcal T, s \in \mathcal S, a \in \mathcal A$.
        \STATE \textbf{Optionally}: Overwrite $\pi^k \leftarrow \frac{1}{k+1} \pi^k + \frac{k}{k+1} \pi^{k-1}$. (FP averaged policy)
        \STATE Compute the mean field $\mu^{k+1} = \Psi(\pi^k)$ induced by $\pi^k$.
        \STATE \textbf{Optionally}: Overwrite $\mu^{k+1} \leftarrow \frac{1}{k+1} \mu^{k+1} + \frac{k}{k+1} \mu^{k}$. (FP averaged mean field)
        \ENDFOR
    \end{algorithmic}
\end{algorithm}

\begin{algorithm}[ht]
    \caption{\textbf{Boltzmann DQN iteration}}
    \label{alg:approx}
    \begin{algorithmic}[1]
        \STATE \textbf{Input}: Temperature $\eta > 0$, prior policy $q \in \Pi$.
        \STATE \textbf{Input}: Simulation parameters, DQN hyperparameters.
        \STATE Initialize $\mu^0 \approx \Psi(q)$ as the mean field induced by $q$ using Algorithm~\ref{alg:sim}.
        \FOR {$k=0, 1, \cdots$}
        \STATE Approximate the Q-function $Q^*(\mu^k,t,s,a)$ using Algorithm~\ref{alg:softdqn} on the MDP induced by $\mu^k$.
        \STATE Define $\pi^k$ by $\pi^k_t(a \mid s) = \frac{ q_t(a \mid s) \exp \left( \frac{Q^*(\mu^k, t, s, a)}{\eta} \right) }{ \sum_{a' \in \mathcal A} q_t(a' \mid s) \exp \left( \frac{Q^*(\mu^k, t, s, a')}{\eta} \right) } $ for all $t \in \mathcal T, s \in \mathcal S, a \in \mathcal A$.
        \STATE Approximately simulate mean field $\mu^{k+1} \approx \Psi(\pi^k)$ induced by $\pi^k$ using Algorithm~\ref{alg:sim}.
        \ENDFOR
    \end{algorithmic}
\end{algorithm}

\begin{algorithm}[ht]
    \caption{\textbf{DQN}}
    \label{alg:softdqn}
    \begin{algorithmic}[1]
        \STATE \textbf{Input}: Number of epochs $L$, mini-batch size $N$, target update frequency $M$, replay buffer size $D$.
        \STATE \textbf{Input}: Probability of random action $\epsilon$, Discount factor $\gamma$, ADAM and gradient clipping parameters.
        \STATE Initialize network $Q_\theta$, target network $Q_{\theta'} \leftarrow Q_\theta$ and replay buffer $\mathcal D$ of size $D$.
        \FOR {$L$ epochs}
            \FOR {$t = 1, \ldots, \mathcal T$}
                \STATE \textbf{One environment step}
                \STATE \quad Let new action $a_{t} \leftarrow \argmax_{a \in \mathcal A} Q_\theta(t,s,a)$, or with probability $\epsilon$ sample uniformly random instead.
                \STATE \quad Sample new state $s_{t+1} \sim p(\cdot \mid s_{t},a_{t})$.
                \STATE \quad Add transition tuple $(s_{t}, a_{t}, r(s_{t}, a_{t}), s_{t+1})$ to replay buffer $\mathcal D$.
                \STATE \textbf{One mini-batch descent step}
                \STATE \quad Sample from the replay buffer: $\{ (s_t^i, a_t^i, r_t^i, s_{t+1}^i) \}_{i=1, \ldots, N} \sim \mathcal D$.
                \STATE \quad Compute loss $J_Q = \sum_{i=1}^N \left( r_t^i + \gamma \max_{a' \in \mathcal A} Q(t+1,s_{t+1}^i,a') - Q(t,s_t^i,a_t^i) \right)^2$.
                \STATE \quad Update $\theta$ according to $\nabla_\theta J_Q$ using ADAM with gradient norm clipping.
                \IF{number of steps $\operatorname{mod} M = 0$}
                    \STATE Update target network $\theta' \leftarrow \theta$.
                    
                \ENDIF
            \ENDFOR
        \ENDFOR
    \end{algorithmic}
\end{algorithm}

\begin{algorithm}[ht]
    \caption{\textbf{Stochastic mean field simulation}}
    \label{alg:sim}
    \begin{algorithmic}[1]
        \STATE \textbf{Input}: Number of mean fields $K$, number of particles $M$, policy $\pi$.
        \FOR {$k = 1, \ldots, K$}
            \STATE Initialize particles $x_m^0 \sim \mu_0$ for all $m = 1, \ldots, M$.
            \FOR {$t \in \mathcal T$}
                \STATE Define empirical measure $\mathbb G_{t}^k \leftarrow \frac 1 M \sum_{m=1}^M \delta_{x_m^t}$.
                \FOR {$m = 1, \ldots, M$}
                    \STATE Sample action $a \sim \pi_t(\cdot \mid x_m^t)$.
                    \STATE Sample new particle state $x_m^{t+1} \sim p(\cdot \mid x_m^t, a, \mathbb G_{t}^k)$.
                \ENDFOR
            \ENDFOR
        \ENDFOR
        \RETURN average empirical mean field $(\frac{1}{K} \sum_{k=1}^K \mathbb G_{t}^k)_{t \in \mathcal T}$
    \end{algorithmic}
\end{algorithm}

\begin{algorithm}[ht]
    \caption{\textbf{Prior descent}}
    \label{alg:prior}
    \begin{algorithmic}[1]
        \STATE \textbf{Input}: Number of outer iterations $I$.
        \STATE \textbf{Input}: Initial prior policy $q \in \Pi$.
        \FOR {outer iteration $i=1, \ldots, I$}
            \STATE Find $\eta$ heuristically or minimally such that Algorithm~\ref{alg:boltzmann} with temperature $\eta$ and prior $q$ converges.
            \IF{no such $\eta$ exists}
                \RETURN $q$
            \ENDIF
            \STATE $q \leftarrow $ solution of Algorithm~\ref{alg:boltzmann} with temperature $\eta$ and prior $q$.
        \ENDFOR
    \end{algorithmic}
\end{algorithm}

\subsection{Implementation details}
For all the DQN experiments, we use the configurations given in Table~\ref{tab:params} and hyperparameters given in Table~\ref{hyperparameters}. Note that we add epsilon scheduling and a discount factor to DQN for stability reasons, i.e. the loss term has an additional factor smaller than one before the maximum operation, cf. \citet{mnih2013playing}. For the action-value network, we use a fully connected dueling architecture (\citet{wang2016dueling}) with one shared hidden layer of 256 neurons, and one separate hidden layer of 256 neurons for value and advantage stream each. As the activation function, we use ReLU. Further, we use gradient norm clipping and the ADAM optimizer. To allow for time-dependent policies, we append the current time to the observations. 

We transform all discrete-valued observations except time to corresponding one-hot vectors, except in the intractably large Taxi environment where we simply observe one value in $\{0,1\}$ for each tile's passenger status. For evaluation of exploitability, we compare the values of the optimal policy and the evaluated policy in the MDP induced by the mean field generated by the evaluated policy. In intractable cases, we use DQN to approximately obtain the optimal policy. In this case, we obtain the values by averaging over many episodes in the MDP induced by the mean field generated by the evaluated policy via Algorithm~\ref{alg:sim}.

\begin{table}[ht]
  \caption{Boltzmann DQN Iteration Parameters}
  \label{tab:params}
  \centering
  \begin{tabular}{lllll}
    \\
    Parameter     & RPS & SIS & Taxi \\
    \midrule
    Fixed point iteration count & $1000$ & $50$ & $15$ \\
    Number of particles for mean field & $1000$ & $1000$ & $200$ \\
    Number of mean fields & $5$ & $5$ & $5$ \\
    Number of episodes for evaluation & $2000$ & $2000$ & $500$
  \end{tabular}
\end{table}

\begin{table}[ht]
  \caption{DQN Hyperparameters}
  \label{hyperparameters}
  \centering
  \begin{tabular}{ll}
    \\
    Hyperparameter     & Value \\
    \midrule
    Replay buffer size & $10000$ \\
    ADAM Learning rate     & $0.0005$ \\
    Discount factor     & $0.99$ \\
    Target update frequency     & $500$ \\
    Gradient clipping norm     & $40$ \\
    Mini-batch size     & $128$ \\
    Epsilon schedule     & $1$ linearly down to $0.02$ at $0.8$ times maximum steps \\
    Total epochs     & $1000$
  \end{tabular}
\end{table}

\subsection{Problems}

\begin{table}[ht]
  \caption{Problem Properties}
  \label{tab:problems}
  \centering
  \begin{tabular}{llll}
    \\
    Problem        & $\left| \mathcal T \right|$    & $\left| \mathcal S \right|$    & $|\mathcal A|$ \\
    \midrule
    LR          & $2$               & $3$               & $2$ \\
    RPS         & $2$               & $4$               & $3$ \\
    SIS         & $50$              & $2$               & $2$ \\
    Taxi        & $100$             & ${\sim} 2^{27}$   & $5$
  \end{tabular}
\end{table}

Summarizing properties of the considered problems are given in Table~\ref{tab:problems}.

\paragraph{LR.}
Similar to the example mentioned in the main text, we let a large number of agents choose simultaneously between going left ($L$) or right ($R$). Afterwards, each agent shall be punished proportional to the number of agents that chose the same action, but more-so for choosing right than left. 

More formally, let $\mathcal S = \{C, L, R\}$, $\mathcal A = \mathcal S \setminus \{C\}$, $\mu_0(C) = 1$, $r(s,a,\mu_t) = - \mathbf 1_{\{L\}}(s) \cdot \mu_t(L) - 2 \cdot \mathbf 1_{\{R\}}(s) \cdot \mu_t(R)$ and $\mathcal T = \{0, 1\}$. Note the difference to the toy example in the main text: right is punished more than left. The transition function allows picking the next state directly, i.e. for all $s,s' \in \mathcal S, a \in \mathcal A$,
\begin{align*}
    \mathbb P(S_{t+1} = s' \mid S_{t} = s, A_t = a) &= \mathbf 1_{\{s'\}}(a) \, .
\end{align*}

For this example, we have $K_Q = 1$ since the return $Q$ of the initial state changes linearly with $\mu_1$ and lies between $0$ and $-2$, while the distance between two mean fields is also bounded by $2$. Analogously, $K_{\Psi} = 1$ since $(\Psi(\pi))_1$ similarly changes linearly with $\pi_0$, and both can change at most by $2$. Thus, we obtain guaranteed convergence via Boltzmann iteration if $\eta > 1$. In numerical evaluations, we see convergence already for $\eta \geq 0.7$.

\paragraph{RPS.}
This game is inspired by \citet{shapley1964some} and their generalized non-zero-sum version of Rock-Paper-Scissors, for which classical fictitious play would not converge. Each of the agents can choose between rock, paper and scissors, and obtains a reward proportional to double the number of beaten agents minus the number of agents beating the agent. We modify the proportionality factors such that a uniformly random prior policy does not constitute a mean field equilibrium.

Let $\mathcal S = \{0, R, P, S\}$, $\mathcal A = \mathcal S \setminus \{0\}$, $\mu_0(0) = 1$, $\mathcal T = \{0, 1\}$, and for any $a \in \mathcal A, \mu_t \in \mathcal P(\mathcal S)$,
\begin{align*}
    r(R,a,\mu_t) &= 2 \cdot \mu_t(S) - 1 \cdot \mu_t(P), \\
    r(P,a,\mu_t) &= 4 \cdot \mu_t(R) - 2 \cdot \mu_t(S), \\
    r(S,a,\mu_t) &= 6 \cdot \mu_t(P) - 3 \cdot \mu_t(R) \, .
\end{align*} 
The transition function allows picking the next state directly, i.e. for all $s,s' \in \mathcal S, a \in \mathcal A$,
\begin{align*}
    \mathbb P(S_{t+1} = s' \mid S_{t} = s, A_t = a) &= \mathbf 1_{\{s'\}}(a) \, .
\end{align*}

\paragraph{SIS.}
In this problem, a large number of agents can choose between social distancing (D) or going out (U). If a susceptible (S) agent chooses social distancing, they may not become infected (I). Otherwise, an agent may become infected with a probability proportional to the number of agents being infected. If infected, an agent will recover with a fixed chance every time step. Both social distancing and being infected have an associated cost.

Let $\mathcal S = \{S, I\}$, $\mathcal A = \{U, D\}$, $\mu_0(I) = 0.6$, $r(s,a,\mu_t) = - \mathbf 1_{\{I\}}(s) - 0.5 \cdot \mathbf 1_{\{D\}}(a)$ and $\mathcal T = \{0, \ldots, 50\}$. We find that similar parameters produce similar results, and set the transition probability mass functions as
\begin{align*}
    \mathbb P(S_{t+1} = S \mid S_{t} = I) &= 0.3 \\
    \mathbb P(S_{t+1} = I \mid S_{t} = S, A_{t} = U) &= 0.9^2 \cdot \mu_t(I) \\
    \mathbb P(S_{t+1} = I \mid S_{t} = S, A_{t} = D) &= 0 \, .
\end{align*}

\paragraph{Taxi.}
In this problem, we consider a $K \times L$ grid. The state is described by a tuple $(x,y,x',y',p,B)$ where $(x,y)$ is the agent's position, $(x',y')$ indicates the current desired destination of the passenger or is $(0,0)$ otherwise, and $p \in \{0,1\}$ indicates whether a passenger is in the taxi or not. Finally, $B$ is a $K \times L$ matrix indicating whether a new passenger is available for the taxi on the corresponding tile. All taxis start on the same tile and have no passengers in the queue or on the map at the beginning. The problem runs for 100 time steps.

The taxi can choose between five actions $W, U, D, L, R$, where $W$ (Wait) allows the taxi to pick up / deliver passengers, and $U, D, L, R$ (Up, Down, Left, Right) allows it to move in all four directions. As there are many taxis, there is a chance of a jam on tile $s$ given by $\min(0.7, 10 \cdot \mu_t(s))$, i.e. the taxi will not move with this probability. The taxi also cannot move into walls or back into the starting tile, in which case it will stay on its current tile. With a probability of $0.8$, a new passenger spawns on one randomly chosen free tile of each region. On picking up a passenger, the destination is generated by randomly picking any free tile of the same region. Delivering passengers to a destination and picking them up gives a reward of $1$ in region $1$ and $1.2$ in region $2$. 

For our experiments, we use the following small map, where $S$ denotes the starting tile, $1$ denotes a free tile from region 1, $2$ denotes a free tile from region 2 and $H$ denotes an impassable wall:
\begin{align*}
    \begin{pmatrix}
    1 & 1 & 1 \\
    1 & 1 & 1 \\
    1 & 1 & 1 \\
    H & S & H \\
    2 & 2 & 2 \\
    2 & 2 & 2 \\
    2 & 2 & 2
    \end{pmatrix}
\end{align*}

This produces a similar situation as in LR, where a fraction of taxis should choose each region so the values balance out, while also requiring solution of a problem that is intractable to solve exactly via dynamic programming. 

\newpage
\subsection{Further experiments}

\begin{figure}[ht]
	\centering
    \includegraphics[width=\textwidth]{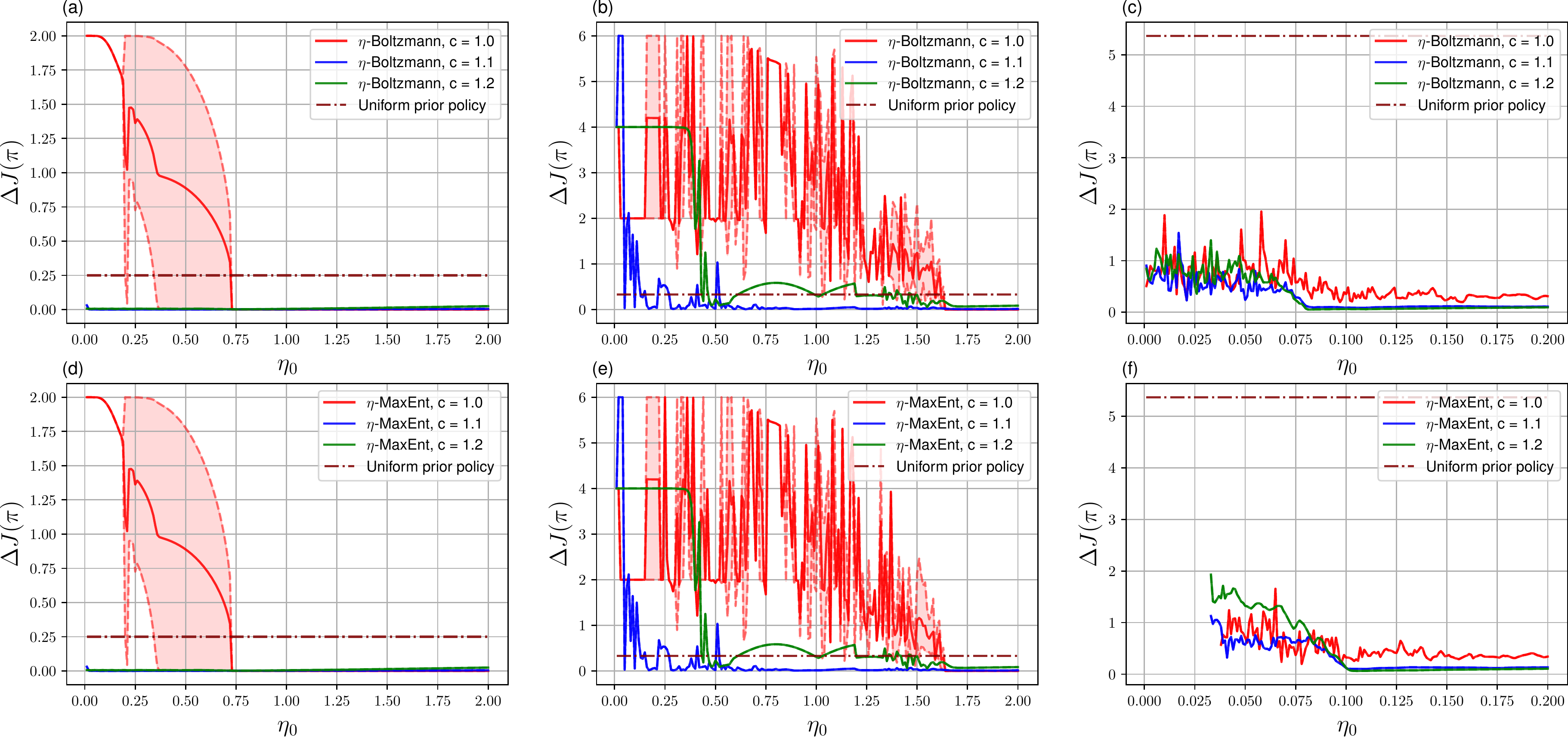}
    \caption{Mean exploitability (straight lines), maximum and minimum (dashed lines) over the final 10 iterations of the last outer iteration. 50 outer iterations and 100 inner iterations each; (a, d) LR; (b, e) RPS; (c, f) SIS. Maximum entropy (MaxEnt) results begin at higher temperatures due to limited floating point accuracy. The exploitability of the initial uniform prior policy is indicated by the dashed horizontal line. } \label{fig:exploitability-a1}
\end{figure}

\begin{figure}[ht]
	\centering
    \includegraphics[width=\textwidth]{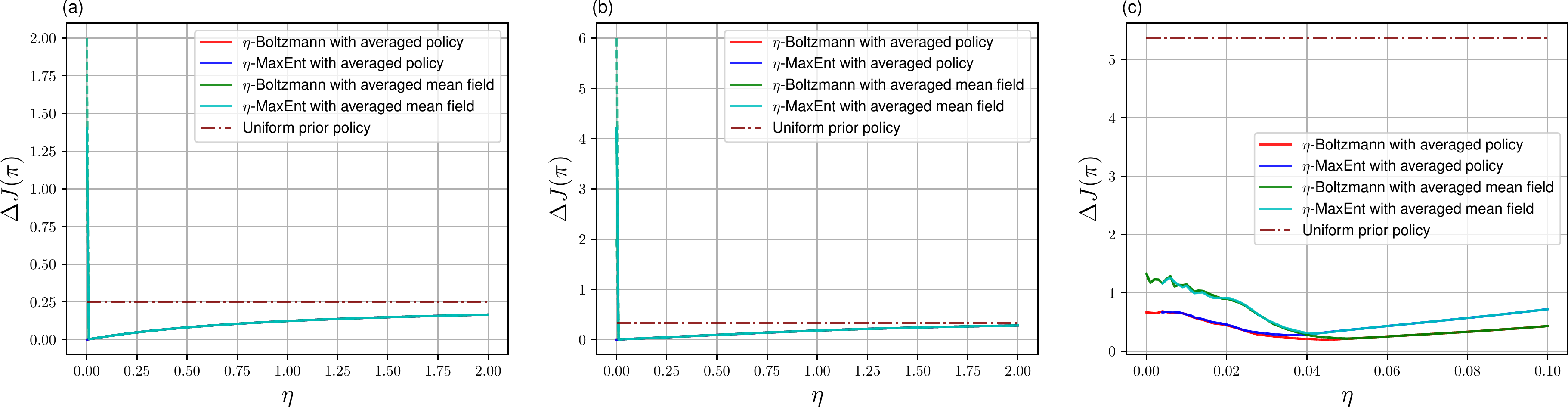}
    \caption{Mean exploitability over the final 10 iterations. Dashed lines represent maximum and minimum over the final 10 iterations. (a) LR, 10000 iterations; (b) RPS, 10000 iterations; (c) SIS, 1000 iterations. The exploitability of the uniform prior policy is indicated by the dashed horizontal line. } \label{fig:exploitability-a2}
\end{figure}

\begin{figure}[ht]
	\centering
    \includegraphics[width=\textwidth]{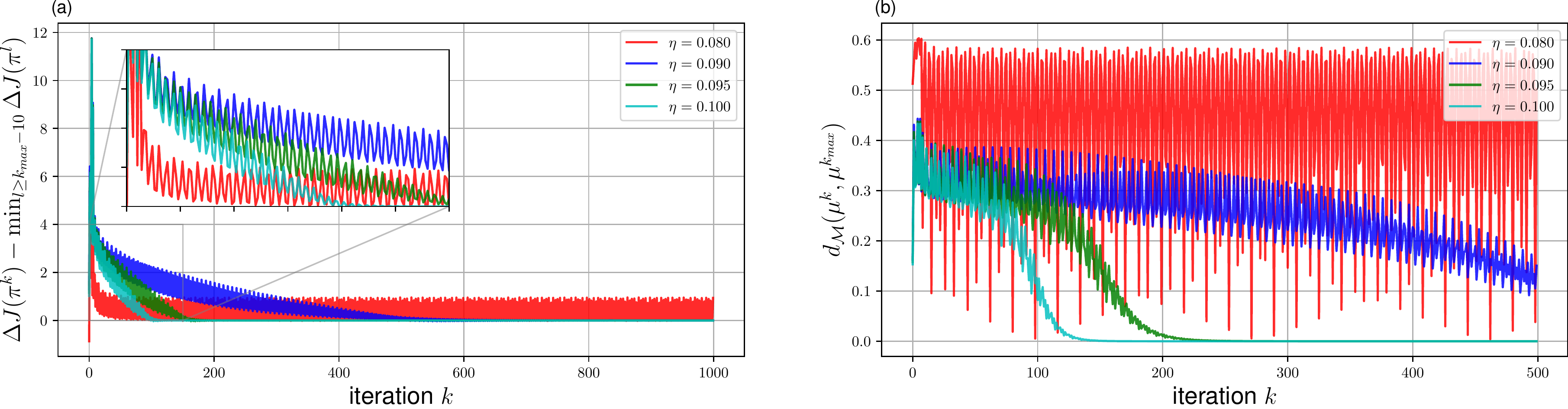}
    \caption{(a) Difference between current and final minimum exploitability over the last 10 iterations; (b) Distance between current and final mean field, cut off at 500 iterations for readability. Plotted for the $\eta$-RelEnt iterations in SIS for the indicated temperature settings and uniform prior policy. } \label{fig:convergence-a3}
\end{figure}

\begin{figure}[ht]
	\centering
    \includegraphics[width=\textwidth]{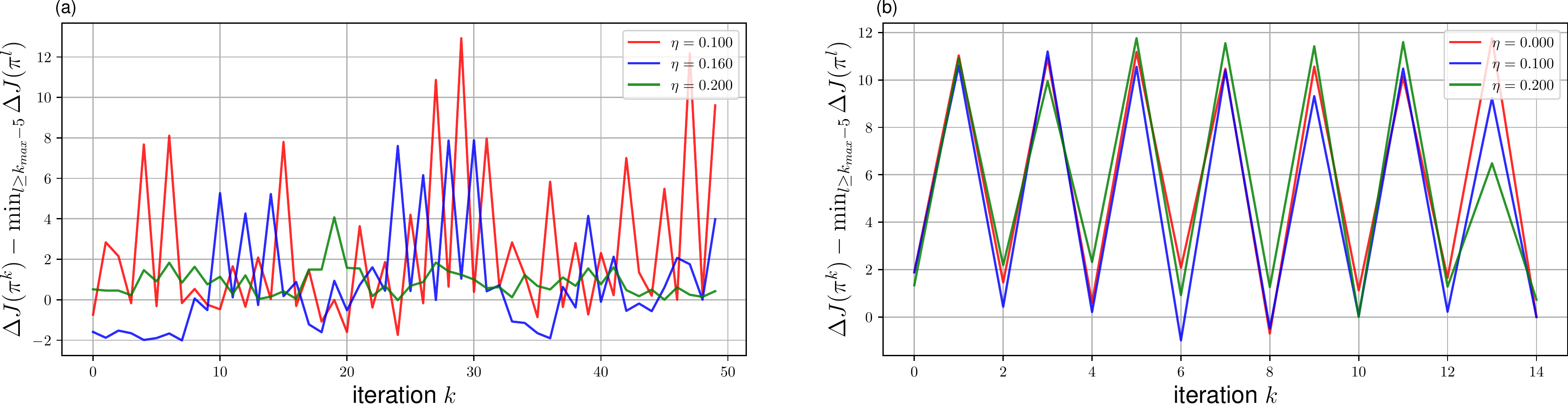}
    \caption{Difference between current and final estimated minimum exploitability over the last 5 iterations. (a) SIS, 50 iterations; (b) Taxi, 15 iterations. Plotted for the $\eta$-Boltzmann DQN iteration for the indicated temperature settings and uniform prior policy. } \label{fig:convergence-a4}
\end{figure}

In Figure~\ref{fig:exploitability-a1}, we observe that prior descent for both Boltzmann and RelEnt MFE with the same uniform prior policy performs qualitatively similarly, and coincide in LR and SIS except for numerical inaccuracies. It can be seen that using a temperature sufficiently low to converge in LR and RPS allows prior descent to descend to the exact MFE iteratively. In SIS on the other hand, picking a fixed temperature that converges for the initial uniform prior policy does not guarantee monotonic improvement of exploitability afterwards. Instead, by applying the heuristic
\begin{align*}
    \eta_{i+1} = \eta_i \cdot c
\end{align*}
for each outer iteration $i$, where $c \geq 1$ adjusts the temperature after each outer iteration, we avoid scanning over all temperatures in each step and reach convergence to a good approximate mean field equilibrium for both Boltzmann and MaxEnt iteration.

In Figure~\ref{fig:exploitability-a2} empirical results are shown for fictitious play variants averaging only policy or mean field. In the simple one-step toy problems LR and RPS, averaging the policies appears to converge to the exact solution without regularization and to the regularized solution with regularization. Averaging the mean fields on the other hand fails, since this method can only produce deterministic policies. By applying any amount of regularization, averaging the mean fields is led to success in LR and SIS. Nonetheless, both methods fail to converge to the MFE in SIS and produce worse results than obtained by prior descent in Figure~\ref{fig:exploitability-a1}.

In Figure~\ref{fig:convergence-a3} we depict the convergence of exploitability and mean field of MaxEnt iteration in SIS. The results are qualitatively similar with Boltzmann iteration and, as in the main text, show the convergence behaviour near the critical temperature leading to convergence.

In Figure~\ref{fig:convergence-a4} we depict the convergence of exploitability for Boltzmann DQN iteration in SIS and Taxi during one of the runs. All 4 other runs show similar qualitative behaviour. As can be seen, the highest temperature of $0.2$ shows less oscillatory behaviour, stabilizing Boltzmann DQN iteration. In Taxi, it can be seen that the used temperatures are insufficient to allow Boltzmann DQN iteration to converge. We believe that using prior descent could allow for better results. We could not verify this due to the high computational cost, as this includes repeatedly and sequentially solving an expensive reinforcement learning problem.
	
\begin{figure}[t]
	\centering
    \includegraphics[width=1\textwidth]{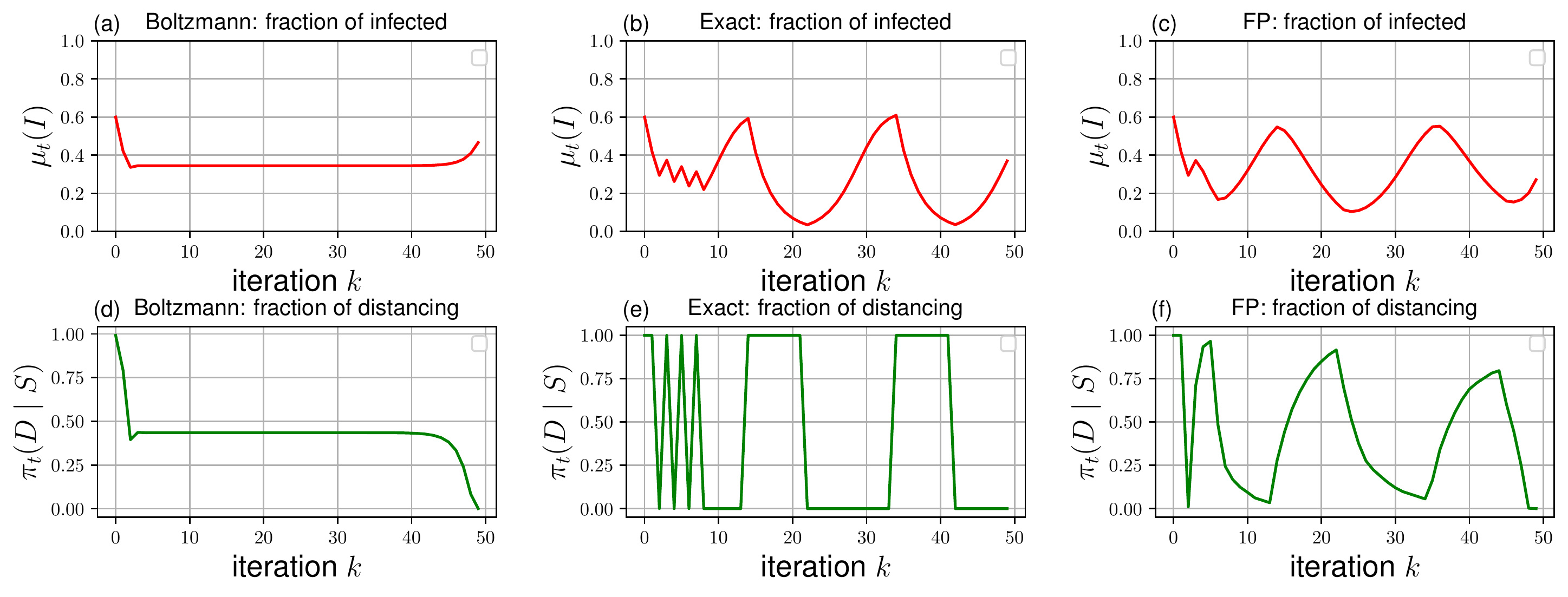}
    \caption{Fraction of infected agents and fraction of susceptible agents picking social distancing over time. (a, d): Boltzmann iteration ($\eta = 0.07$); (b, e): exact fixed point iteration; (c, f): fictitious play (averaging both policy and mean field) results in SIS after 500 iterations. More iterations and averaging only policy or mean field show same qualitative results. } \label{fig:sis-example}
\end{figure}

Finally, in Figure~\ref{fig:sis-example} we depict the resulting behavior in the SIS case. In the Boltzmann iteration result, at the beginning the number of infected is high enough to make social distancing the optimal action to take. As the number of infected falls, it reaches an equilibrium point where both social distancing or potentially getting infected are of equal value. Finally, as the game ends at time $t=T=50$, there is no point in social distancing any more. Our approach yields intuitive results here, while exact fixed point iteration and FP fail to converge.

\section{Proofs} 
\subsection{Completeness of mean field and policy space} \label{app:comp}
\begin{lemma}
The metric spaces $(\Pi, d_{\Pi})$ and $(\mathcal M, d_{\mathcal M})$ are complete metric spaces.
\end{lemma}
\begin{proof}
The metric space $(\mathcal M, d_{\mathcal M})$ is a complete metric space. Let $(\mu^n)_{n \in \mathbb N} \in \mathcal M^{\mathbb N}$ be a Cauchy sequence of mean fields. Then by definition, for any $\varepsilon > 0$ there exists integer $N > 0$ such that for any $m,n > N$ we have
\begin{align*}
    d_{\mathcal M}(\mu^n, \mu^m) &< 0.5\varepsilon \\
    \implies \forall t \in \mathcal T: d_{TV}(\mu^n_t, \mu^m_t) = \frac{1}{2} \sum_{s \in \mathcal S} | \mu^n_t(s) - \mu^m_t(s) | &< 0.5\varepsilon \\
    \implies \forall t \in \mathcal T, s \in \mathcal S: | \mu^n_t(s) - \mu^m_t(s) | &< \varepsilon \, .
\end{align*}
By completeness of $\mathbb R$ there exists the limit of $(\mu^n_t(s))_{n \in \mathbb N}$ for all $t \in \mathcal T, s \in \mathcal S$, suggestively denoted by $\mu_t(s)$. The mean field $\mu = \{ \mu_t \}_{t \in \mathcal T}$ with the probabilities defined by the aforementioned limits fulfills $\mu^n \to \mu$ and is in $\mathcal M$, showing completeness of $\mathcal M$. 

We do this analogously for $(\Pi, d_{\Pi})$. Thus, $(\Pi, d_{\Pi})$ and $(\mathcal M, d_{\mathcal M})$ are complete metric spaces.
\end{proof}

\subsection{Lipschitz continuity} \label{app:lipschitz}
\begin{lemma} \label{lem:lip}
Assume bounded and Lipschitz functions $f: X \to \mathbb R$ and $g: X \to \mathbb R$ mapping from a metric space $(X,d_X)$ into $\mathbb R$ with Lipschitz constants $C_f, C_g$ and bounds $\left| f(x) \right| \leq M_f$, $\left| g(x) \right| \leq M_g$. The sum of both functions $f+g$, the product of both functions $f \cdot g$ and the maximum of both functions $\max(f,g)$ are all Lipschitz and bounded with Lipschitz constants $C_f + C_g$, $(M_f C_g + M_g C_f)$, $\max(C_f, C_g)$ and bounds $M_f + M_g$, $M_f M_g$, $\max(M_f, M_g)$.
\end{lemma}
\begin{proof}
Let $x,y \in X$ be arbitrary. By the triangle inequality, we obtain 
\begin{align*}
    \lvert f(x) + g(x) - (f(y) + g(y)) \rvert 
    &\leq \lvert f(x) - f(y) \rvert + \lvert g(x) - g(y) \rvert \leq (C_f + C_g) d_X(x,y) \, .
\end{align*}
Analogously, we obtain
\begin{align*}
    \lvert f(x)g(x) - f(y)g(y) \rvert 
    &\leq \lvert f(x)g(x) - f(x)g(y) \rvert + \lvert f(x)g(y) - f(y)g(y) \rvert \leq (M_f C_g + M_g C_f) d_X(x,y) \, .
\end{align*}
For the maximum of both functions, consider case by case. If $f(x) \geq g(x)$ and $f(y) \geq g(y)$ we obtain
\begin{align*}
    \lvert \max(f(x),g(x)) - \max(f(y),g(y)) \rvert = \lvert f(x) - f(y) \rvert
    &\leq C_f d_X(x,y)
\end{align*}
and analogously for $g(x) \geq f(x)$ and $g(y) \geq f(y)$ 
\begin{align*}
    \lvert \max(f(x),g(x)) - \max(f(y),g(y)) \rvert = \lvert g(x) - g(y) \rvert
    &\leq C_g d_X(x,y) \, .
\end{align*}
On the other hand, if $g(x) < f(x)$ and $g(y) \geq f(y)$ , we have either $g(y) \geq f(x)$ and thus
\begin{align*}
    \lvert \max(f(x),g(x)) - \max(f(y),g(y)) \rvert = \lvert f(x) - g(y) \rvert = g(y) - f(x) < g(y) - g(x)
    &\leq C_g d_X(x,y)
\end{align*}
or $g(y) < f(x)$ and thus
\begin{align*}
    \lvert \max(f(x),g(x)) - \max(f(y),g(y)) \rvert = \lvert f(x) - g(y) \rvert = f(x) - g(y) \leq f(x) - f(y)
    &\leq C_f d_X(x,y) \, .
\end{align*}
The case for $f(x) < g(x)$ and $f(y) \geq g(y)$ as well as boundedness is analogous.
\end{proof}

\subsection{Proof of Proposition~1} \label{app:kakutani}
\begin{proof}
Since we work with finite $\mathcal T, \mathcal S, \mathcal A$, we identify the space of mean fields $\mathcal M$ with the $|\mathcal T| (|\mathcal S|-1)$-dimensional simplex $S_{|\mathcal T| (|\mathcal S|-1)} \subseteq \mathbb R^{|\mathcal T| (|\mathcal S|-1)}$ via the values of the probability mass functions at all times and states. Analogously the space of policies $\Pi$ is identified with $S_{|\mathcal T| |\mathcal S| (|\mathcal A|-1)} \subseteq \mathbb R^{|\mathcal T| |\mathcal S| (|\mathcal A|-1)}$.

Define the set-valued map $\hat \Gamma: S_{|\mathcal T| |\mathcal S| (|\mathcal A|-1)} \to 2^{S_{|\mathcal T| |\mathcal S| (|\mathcal A|-1)}}$ mapping from a policy $\pi$ represented by the input vector, to the set of vector representations of optimal policies in the MDP induced by $\Psi(\pi)$.

A policy $\pi$ is optimal in the MDP induced by $\mu \in \mathcal M$ if and only if its value function defined by
\begin{align*}
    V^{\pi}(\mu, t, s) &= \sum_{a \in \mathcal A} \pi_t(a \mid s) \left( r(s,a,\mu_t) + \sum_{s' \in \mathcal S} p(s'\mid s, a, \mu_t) V^{\pi}(\mu, t+1,s') \right) \, ,
\end{align*}
is equal to the optimal action-value function defined by
\begin{align*}
    V^*(\mu, t, s) &= \max_{a \in \mathcal A} \left( r(s,a,\mu_t) + \sum_{s' \in \mathcal S} p(s'\mid s, a, \mu_t) V^*(\mu, t+1, s') \right)
\end{align*}
for every $t \in \mathcal T, s \in \mathcal S$, with terminal conditions $V^*(\mu,T,s) \equiv V^{\pi}(\mu,T,s) \equiv 0$. Moreover, an optimal policy always exists. For more details, see e.g. \citet{puterman2014markov}. Define the optimal action-value function for every $t \in \mathcal T, s \in \mathcal S, a \in \mathcal A$ via
\begin{align*}
    Q^*(\mu, t,s,a) &= r(s,a,\mu_t) + \sum_{s' \in \mathcal S} p(s'\mid s, a, \mu_t) V^*(\mu, t+1, s')
\end{align*}
with terminal condition $Q^*(\mu,T,s,a) \equiv 0$. Then, the following lemma characterizes optimality of policies.

\begin{lemma} \label{lem:policyoptimality}
A policy $\pi$ fulfills $\pi \in \hat \Gamma(\hat \pi)$ if and only if 
\begin{align*}
    \pi_t(a \mid s) > 0 \implies a \in \argmax_{a' \in \mathcal A} Q^*(\Psi(\hat \pi), t, s, a')
\end{align*} 
for all $t \in \mathcal T, s \in \mathcal S, a \in \mathcal A$.
\end{lemma}
\begin{subproof}
To see the implication, consider $\pi \in \hat \Gamma(\hat \pi)$. Then, if the right-hand side was false, there exists a maximal $t \in \mathcal T$ and $s \in \mathcal S, a \in \mathcal A$ such that $\pi_t(a \mid s) > 0$ but $a \centernot\in \argmax_{a' \in \mathcal A} Q^*(\Psi(\hat \pi), t, s, a')$. Since for any $t' > t$ we have optimality, $V^{\pi}(\mu, t+1,s') = V^*(\mu, t+1,s')$ by induction. However, $V^{\pi}(\mu, t, s) < V^*(\mu, t, s)$ since the suboptimal action is assigned positive probability, contradicting optimality of $\pi$. On the other hand, if the right-hand side is true, then $V^{\pi}(\mu, t, s) = V^*(\mu, t, s)$ by induction, which implies that $\pi$ is optimal.
\end{subproof}

We will now check that the requirements of Kakutani's fixed point theorem hold for $\hat \Gamma$. The finite-dimensional simplices are convex, closed and bounded, hence compact. $\hat \Gamma$ maps to a non-empty set, as the induced mean field is uniquely defined and any finite MDP (induced by this mean field) has an optimal policy. 

For any $\pi$, $\hat \Gamma(\pi)$ is convex, since the set of optimal policies is convex as shown in the following. Consider a convex combination $\tilde \pi = \lambda \pi + (1-\lambda) \pi'$ of optimal policies $\pi, \pi'$ for $\lambda \in [0,1]$. Then, the resulting policy will be optimal, since we have
\begin{align*}
    \tilde \pi_t(a \mid s) > 0 \implies \pi_t(a \mid s) > 0 \vee \pi_t'(a \mid s) > 0 \implies a \in \argmax_{a \in \mathcal A} Q^*(\Psi(\hat \pi), t, s, a)
\end{align*}
for any $t \in \mathcal T, s \in \mathcal S, a \in \mathcal A$ and thus optimality by Lemma~\ref{lem:policyoptimality}.

Finally, we show that $\hat \Gamma$ has a closed graph. Consider arbitrary sequences $(\pi_n, \pi_n') \to (\pi, \pi')$ with $\pi_n' \in \hat \Gamma(\pi_n)$. It is then sufficient to show that $\pi' \in \hat \Gamma(\pi)$. By the standing assumption, we have continuity of $\Psi$ and $\mu \to Q^*(\mu, t, s, a)$ for any $t \in \mathcal T, s \in \mathcal S, a \in \mathcal A$, as sums, products and compositions of continuous functions remain continuous. Therefore, the composition $\pi \to Q^*(\Psi(\pi), t, s, a)$ is continuous. To show that $\pi' \in \hat \Gamma(\pi)$, assume that $\pi' \centernot\in \hat \Gamma(\pi)$. By Lemma~\ref{lem:policyoptimality} there exists $t \in \mathcal T, s \in \mathcal S, a \in \mathcal A$ such that $\pi_t'(a \mid s) > 0$ and further there exists $a' \in \mathcal A$ such that $Q^*(\Psi(\pi), t, s, a') > Q^*(\Psi(\pi), t, s, a)$. Fix such an $a' \in \mathcal A$. Let $\delta \equiv Q^*(\Psi(\pi), t, s, a') - Q^*(\Psi(\pi), t, s, a)$, then by continuity there exists $\varepsilon > 0$ such that for all $\hat \pi \in \Pi$ we have 
\begin{align*}
    d_{\Pi}(\hat \pi, \pi) < \varepsilon \implies \left| Q^*(\Psi(\hat \pi), t, s, a) - Q^*(\Psi(\pi), t, s, a) \right| < \frac{\delta}{2} \, .
\end{align*} 
By convergence, there is an integer $N \in \mathbb N$ such that for all $n > N$ we have $d_{\Pi}(\pi_n, \pi) < \varepsilon$ and therefore
\begin{align*}
    Q^*(\Psi(\pi_n), t, s, a') > Q^*(\Psi(\pi), t, s, a') - \frac{\delta}{2} = Q^*(\Psi(\pi), t, s, a) + \frac{\delta}{2} > Q^*(\Psi(\pi_n), t, s, a) \, .
\end{align*} 
Since $(\pi'_n)_t(a \mid s) \to \pi'_t(a \mid s) > 0$, there also exists $M \in \mathbb N$ such that for all $m > M$,
\begin{align*}
    \left| (\pi'_m)_t(a \mid s) - \pi'_t(a \mid s) \right| < \pi'_t(a \mid s) \, .
\end{align*} 
Let $n > \max(N, M)$, then it follows that $(\pi'_n)_t(a \mid s) > 0$ but $a \centernot\in \argmax_{a' \in \mathcal A} Q^*(\Psi(\pi), t, s, a')$ since we have $Q^*(\Psi(\pi_n), t, s, a') > Q^*(\Psi(\pi_n), t, s, a)$, contradicting $\pi'_n \in \hat\Gamma(\pi_n)$ by Lemma~\ref{lem:policyoptimality}. Hence, $\hat \Gamma$ must have a closed graph.

By Kakutani's fixed point theorem, there exists a fixed point $\pi^*$ that generates some mean field $\Psi(\pi^*)$. The associated pair $(\pi^*, \Psi(\pi^*))$ is an MFE by definition.
\end{proof}

\subsection{Proof of Proposition~3} \label{app:existence}
\begin{proof}
The space of mean fields $(\mathcal M, d_{\mathcal M})$ is equivalent to convex and compact finite-dimensional simplices. In this representation, each coordinate of the operators $\tilde \Gamma_\eta(\mu)$ and $\Gamma_\eta(\mu)$ consists of compositions, sums and products of continuous functions, since the functions $r(s,a,\mu_t)$ and $p(s' \mid s,a,\mu_t)$ are assumed to be continuous. Existence of a fixed point follows immediately by Brouwer's fixed point theorem. 
\end{proof}

\subsection{Proof of Theorem~1} \label{app:epsopt}
\begin{proof}
The proof is a slightly simplified version of the one found in \citet{saldi2018markov}. Note that we require the results later, so for convenience we give the full details.

The empirical measure $\mathbb G^N_{S_t}$ is a random variable on $\mathcal P(\mathcal S)$, i.e. its law $\mathcal L(\mathbb G^N_{S_t}) \in \mathcal P(\mathcal P(\mathcal S))$ is a distribution over probability measures. Since we want to show convergence of the empirical measure to the mean field, let us pick a metric on $\mathcal P(\mathcal P(\mathcal S))$. Remember that we metrized $\mathcal P(\mathcal S)$ with the total variation distance. We metrize $\mathcal P(\mathcal P(\mathcal S))$ with the 1-Wasserstein metric defined for any $\Phi, \Psi \in \mathcal P(\mathcal P(\mathcal S))$ by the infimum over couplings
\begin{align*}
    W_1(\Phi, \Psi) \equiv \inf_{\mathcal L(X_1) = \Phi, \mathcal L(X_2) = \Psi} \mathbb E \left[ d_{TV}(X_1, X_2) \right] \, .
\end{align*}

\begin{lemma} \label{lem:equiv-conv}
Let $\{ \Phi_n \}_{n \in \mathbb N}$ be a sequence of measures with $\Phi_n \in \mathcal P(\mathcal P(\mathcal S))$ for all $n \in \mathbb N$. Further, let $\mu \in \mathcal P(\mathcal S)$ arbitrary. Then, the following are equivalent.
\begin{enumerate}
    \item[(a)] $W_1(\Phi_n, \delta_\mu) \to 0$ as $n \to \infty$
    \item[(b)] $\mathbb E \left[ \left| F(X_n) - F(X) \right| \right] \to 0$ as $n \to \infty$ for any continuous, bounded $F: \mathcal P(\mathcal S) \to \mathbb R$, any sequence $\{ X_n \}_{n \in \mathbb N}$ of $\mathcal P(\mathcal S)$-valued random variables and any $\mathcal P(\mathcal S)$-valued random variable $X$ with $\mathcal L(X_n) = \Phi_n$ and $\mathcal L(X) = \delta_\mu$.
    \item[(c)] $\mathbb E \left[ \left| X_n(f) - X(f) \right| \right] \to 0$ as $n \to \infty$ for any $f: \mathcal S \to \mathbb R$, any sequence $\{ X_n \}_{n \in \mathbb N}$ of $\mathcal P(\mathcal S)$-valued random variables and any $\mathcal P(\mathcal S)$-valued random variable $X$ with $\mathcal L(X_n) = \Phi_n$ and $\mathcal L(X) = \delta_\mu$.
\end{enumerate}
\end{lemma}
\begin{subproof} 
Define the only possible coupling $\Delta_n \equiv \Phi_n \times \delta_\mu$.

(b), (c) $\implies$ (a):

Define $F_s(x) \equiv x(s)$ and $f_s(s') \equiv \mathbf 1_{\{s\}}(s')$ for all $s \in \mathcal S$, where $F_s$ is continuous. By assumption,
\begin{align*}
    W_1(\Phi_n, \delta_\mu) 
    &= \inf_{\mathcal L(X_n) = \Phi_n, \mathcal L(X) = \delta_\mu} \mathbb E \left[ d_{TV}(X_n, X) \right] \\
    &= \frac{1}{2} \int_{\mathcal P(\mathcal S) \times \mathcal P(\mathcal S)} \sum_{s \in \mathcal S} | X_n(s) - X(s) | \, d\Delta_n \\
    &= \frac{1}{2} \sum_{s \in \mathcal S} \mathbb E \left[ \left| X_n(s) - X(s) \right| \right] \to 0
\end{align*}
since for any $s \in \mathcal S$, we have
\begin{align*}
    \mathbb E \left[ \left| X_n(s) - X(s) \right| \right] = \mathbb E \left[ \left| F_s(X_n) - F_s(X) \right| \right] = \mathbb E \left[ \left| X_n(f_s) - X(f_s) \right| \right] \, .
\end{align*}

(a) $\implies$ (b), (c):

We have
\begin{align*}
    \mathbb E \left[ \left| F(X_n) - F(X) \right| \right] &= \int_{\mathcal P(\mathcal S) \times \mathcal P(\mathcal S)} | F(\nu) - F(\nu') | \, \Delta_n(d\nu, d\nu') \\
    &= \int_{\mathcal P(\mathcal S)} | F(\nu) - F(\mu) | \, \Phi_n(d\nu) \\
    &\to \int_{\mathcal P(\mathcal S)} | F(\nu) - F(\mu) | \, \delta_\mu(d\nu) = 0
\end{align*}
by continuity and boundedness of $| F(\nu) - F(\mu) |$, and convergence in $W_1$ implying weak convergence. Analogously,
\begin{align*}
    \mathbb E \left[ \left| X_n(f) - X(f) \right| \right] 
    &= \int_{\mathcal P(\mathcal S)} | \nu(f) - \mu(f) | \, \Phi_n(d\nu) 
    \to \int_{\mathcal P(\mathcal S)} | \nu(f) - \mu(f) | \, \delta_\mu(d\nu) = 0
\end{align*}
since $f$ and thus $| \nu(f) - \mu(f) |$ is automatically bounded from finiteness of $\mathcal S$, and $\nu(f) = \sum_{s \in \mathcal S} \nu(s) f(s) \to \sum_{s \in \mathcal S} \mu(s) f(s)$ as $\nu \to \mu$ in total variation distance implies continuity of $| \nu(f) - \mu(f) |$.
\end{subproof}

First, it is shown that when all other agents follow the same policy $\pi$, then the empirical distribution is essentially the deterministic mean field as $N \to \infty$, i.e. $\mathcal L(\mathbb G^N_{S_t}) \to \mathcal L(\mu_t) \equiv \delta_{\mu_t}$ with $\mu = \Psi(\pi)$

\begin{lemma} \label{lem:law-conv}
Consider a set of policies $(\tilde \pi, \pi, \ldots, \pi) \in \Pi^N$ for all agents. Under this set of policies, the law of the empirical distribution $\mathcal L(\mathbb G^N_{S_t}) \in \mathcal P(\mathcal M)$ converges to $\delta_{\mu_t}$ where $\mu = \Psi(\pi)$ as $N \to \infty$ in 1-Wasserstein distance. 
\end{lemma}
\begin{subproof}
Define the Markov kernel $P^\pi_{t,\nu}$ such that its probability mass function fulfills
\begin{align*}
    P^\pi_{t,\nu}(s' \mid s) \equiv \sum_{a \in \mathcal A} \pi_t(a \mid s) p(s' \mid s, a, \nu) 
\end{align*}
for any $t \in \mathcal T, s \in \mathcal S, \nu \in \mathcal P(\mathcal S), \pi \in \Pi$ and analogously
\begin{align*}
    \tilde \nu P^\pi_{t,\nu}(s') \equiv \sum_{s \in \mathcal S} \tilde \nu(s) \sum_{a \in \mathcal A} \pi_t(a \mid s) p(s' \mid s, a, \nu)
\end{align*}
for any $\tilde \nu \in \mathcal P(\mathcal S)$. Note that $\mu_{t+1} = \mu_t P^\pi_{t,\mu_t}(g)$ for mean fields $\mu = \Psi(\pi)$ induced by $\pi$.

We show that $\mathbb E \left[ \left| \mathbb G^N_{S_t}(f) - \mu_t(f) \right| \right] \to 0$ as $N \to \infty$ for any function $f: \mathcal S \to \mathbb R$ and any time $t \in \mathcal T$. From this, the desired result follows by Lemma~\ref{lem:equiv-conv}. Since $\mathbb G^N_{S_t}(\cdot) \equiv \frac{1}{N} \sum_{i=1}^N \delta_{S_t^i}(\cdot)$ and $S_0^i\sim \mu_0$ we have at time $t = 0$ that
\begin{align*}
    \lim_{N \to \infty} \mathbb E \left[ \left| \mathbb G^N_{S_0}(f) - \mu_0(f) \right| \right]
    &= \lim_{N \to \infty} \mathbb E \left[  \left| \frac{1}{N} \sum_{i=1}^N f(S^i_0) - \mathbb E \left[  f(S^i_0)  \right] \right|  \right] = 0
\end{align*}
by the strong law of large numbers and the dominated convergence theorem.

Assuming this holds for $t$, then for $t+1$ we have
\begin{align*}
    \mathbb E \left[ \left| \mathbb G^N_{S_{t+1}}(f) - \mu_{t+1}(f) \right| \right]
    &\leq \mathbb E \left[ \left| \mathbb G^N_{S_{t+1}}(f) - \mathbb G^{N-1}_{S_{t+1}}(f) \right| \right] \\
    &\quad + \mathbb E \left[ \left| \mathbb G^{N-1}_{S_{t+1}}(f) - \mathbb G^{N-1}_{S_{t}} P_{t, \mathbb G^{N}_{S_{t}}}^\pi (f) \right| \right] \\
    &\quad + \mathbb E \left[ \left| \mathbb G^{N-1}_{S_{t}} P_{t, \mathbb G^{N}_{S_{t}}}^\pi (f) - \mathbb G^{N}_{S_{t}} P_{t, \mathbb G^{N}_{S_{t}}}^\pi (f) \right| \right] \\
    &\quad + \mathbb E \left[ \left| \mathbb G^{N}_{S_{t}} P_{t, \mathbb G^{N}_{S_{t}}}^\pi (f) - \mu_{t} P_{t,\mu_{t}}^\pi (f) \right| \right]
\end{align*}
where we defined $\mathbb G^{N-1}_{S_t}(\cdot) \equiv \frac{1}{N-1} \sum_{i=2}^N \delta_{S_t^i}(\cdot)$. 

For the first term, we have as $N \to \infty$
\begin{align*}
    \mathbb E \left[ \left| \mathbb G^N_{S_{t+1}}(f) - \mathbb G^{N-1}_{S_{t+1}}(f) \right| \right]
    &= \mathbb E \left[ \left| \frac{1}{N} \sum_{i=1}^N f(S_{t+1}^i) - \frac{1}{N-1} \sum_{i=2}^N f(S_{t+1}^i) \right| \right] \\
    &\leq \frac{1}{N} \mathbb E \left[ \left| f(S_{t+1}^1) \right| \right] + \left| \frac{1}{N} - \frac{1}{N-1} \right| \sum_{i=2}^N \mathbb E \left[ \left| f(S_{t+1}^i) \right| \right] \\
    &\leq \left( \frac{1}{N} + \frac{N-1}{N (N-1)} \right) \max_{s \in \mathcal S} | f(s) | \to 0 \, .
\end{align*}

For the second term, as $N \to \infty$ we have by Jensen's inequality and bounds $|f| \leq M_f$ (by finiteness of $\mathcal S$)
\begin{align*}
    \mathbb E \left[ \left| \mathbb G^{N-1}_{S_{t+1}}(f) - \mathbb G^{N-1}_{S_{t}} P^\pi_{t,\mathbb G^{N-1}_{S_{t}}} (f) \right| \right]^2
    &= \mathbb E \left[ \mathbb E \left[ \left| \mathbb G^{N-1}_{S_{t+1}}(f) - \mathbb G^{N-1}_{S_{t}} P^\pi_{t,\mathbb G^{N-1}_{S_{t}}} (f) \right| \mid S_{t} \right] \right]^2 \\
    &= \mathbb E \left[ \mathbb E \left[ \left| \frac{1}{{N-1}} \sum_{i=2}^N \left( f(S_{t+1}^i) - \mathbb E \left[ f(S_{t+1}^i) \mid S_{t} \right] \right) \right| \mid S_{t} \right] \right]^2 \\
    &\leq \frac{1}{(N-1)^2} \sum_{i=2}^N \mathbb E \left[ \mathbb E \left[ \left( f(S_{t+1}^i) - \mathbb E \left[ f(S_{t+1}^i) \mid S_{t} \right] \right)^2 \mid S_{t} \right] \right] \\
    &\leq \frac{1}{{N-1}} \cdot 4M_f^2 \to 0 \, .
\end{align*}

For the third term, we again have as $N \to \infty$
\begin{align*}
    \mathbb E \left[ \left| \mathbb G^{N-1}_{S_{t}} P_{t, \mathbb G^{N}_{S_{t}}}^\pi (f) - \mathbb G^{N}_{S_{t}} P_{t, \mathbb G^{N}_{S_{t}}}^\pi (f) \right| \right]
    &= \mathbb E \left[ \left| \sum_{s \in \mathcal S} \left( \mathbb G^{N-1}_{S_{t}}(s) - \mathbb G^{N}_{S_{t}}(s) \right) \sum_{a \in \mathcal A} \pi_t(a \mid s) \sum_{s' \in \mathcal S} p(s' \mid s, a, \mathbb G^{N}_{S_{t}}) f(s') \right| \right] \\
    &\leq \mathbb E \left[ \left| \left( \frac{1}{{N-1}} - \frac{1}{N} \right) \sum_{i=2}^N \sum_{a \in \mathcal A} \pi_t(a \mid S_t^i) \sum_{s' \in \mathcal S} p(s' \mid S_t^i, a, \mathbb G^{N}_{S_{t}}) f(s') \right| \right] \\
    &\quad + \mathbb E \left[ \left| \frac{1}{{N}} \sum_{a \in \mathcal A} \pi_t(a \mid S_t^1) \sum_{s' \in \mathcal S} p(s' \mid S_t^1, a, \mathbb G^{N}_{S_{t}}) f(s') \right| \right] \\
    &\leq \left( \frac{N-1}{N (N-1)} + \frac{1}{N} \right) \max_{s \in \mathcal S} | f(s) | \to 0 \, .
\end{align*}

For the fourth term, define $F: \mathcal P(\mathcal S) \to \mathbb R$, $F(\nu) = \nu P^\pi_{t,\nu}(f)$ and observe that $F$ is continuous, since $\nu \to \nu'$ if and only if $\nu(s) \to \nu'(s)$ for all $s \in \mathcal S$, and therefore (as $p$ is assumed continuous by Assumption~1)
\begin{align*}
    F(\nu) = \nu P^\pi_{t,\nu}(f) = \sum_{s \in \mathcal S} \nu(s) \sum_{a \in \mathcal A} \pi_t(a \mid s) \sum_{s' \in \mathcal S} p(s' \mid s, a, \nu) f(s')
\end{align*}
is continuous. By Lemma~\ref{lem:equiv-conv}, we have from the induction hypothesis $\mathbb G^N_{S_{t}} \to \mu_t$ that
\begin{align*}
    \mathbb E \left[ \left| \mathbb G^N_{S_{t}} P^\pi_{t,\mathbb G^N_{S_{t}}} (f) - \mu_{t} P_{t,\mu_{t}}^\pi (f) \right| \right] \to 0 \ .
\end{align*}

Therefore, $\mathbb E \left[ \left| \mathbb G^N_{S_{t+1}}(f) - \mu_{t+1}(f) \right| \right] \to 0$ which implies the desired result by induction.
\end{subproof}

Consider the case where all agents follow a set of policies $(\pi^N, \pi, \ldots, \pi) \in \Pi^N$ for each $N \in \mathbb N$. Define new single-agent random variables $S_t^{\mu}$ and $A_t^{\mu}$ with $S_0^{\mu} \sim \mu_0$ and
\begin{align*} 
    \mathbb P(A_{t}^{\mu} = a \mid S_{t}^{\mu} = s) &= \pi^N_t(a \mid s), \\
    \mathbb P(S_{t+1}^{\mu} = s' \mid S_{t}^{\mu} = s, A_{t}^{\mu} = a) &= p(s' \mid s, a, \mu_t) \, ,
\end{align*}
where the deterministic mean field $\mu$ is used instead of the empirical distribution.

\begin{lemma} \label{lem:equicont}
Consider an equicontinuous, uniformly bounded family of functions $\mathcal F$ on $\mathcal P(\mathcal S)$ and define
\begin{align*}
    F_t(\nu) \equiv \sup_{f \in \mathcal F} | f(\nu) - f(\mu_t) | 
\end{align*}
for any $t \in \mathcal T$. Then, $F_t$ is continuous and bounded and by Lemma~\ref{lem:equiv-conv} we have
\begin{align*}
    \lim_{N \to \infty} \mathbb E \left[ \sup_{f \in \mathcal F} \left| f(\mathbb G^N_{S_t}) - f(\mu) \right| \right] = 0
\end{align*}
\end{lemma}
\begin{subproof}
$F_t$ is continuous, since for $\nu_n \to \nu$
\begin{align*}
    | F_t(\nu_n) - F_t(\nu) | = \left| \sup_{f \in \mathcal F} | f(\nu_n) - f(\mu_t) | - \sup_{f \in \mathcal F} | f(\nu) - f(\mu_t) | \right| \leq \sup_{f \in \mathcal F} | f(\nu_n) - f(\nu) | \to 0
\end{align*}
by equicontinuity. Further, $F_t$ is bounded since $|F_t(\nu)| \leq \sup_{f \in \mathcal F} | f(\nu) | + | f(\mu_t) |$ is uniformly bounded. By Lemma~\ref{lem:law-conv}, we have $W_1(\mathbb G^N_{S_t}, \delta_{\mu_t}) \to 0$ as $N \to \infty$, therefore Lemma~\ref{lem:equiv-conv} applies. 
\end{subproof}

\begin{lemma} \label{lem:equiconv}
Suppose that at some time $t \in \mathcal T$, it holds that
\begin{align*}
    \lim_{N \to \infty} \left| \mathcal L(S_t^1)(g_N) - \mathcal L(S_t^{\mu})(g_N) \right| = 0
\end{align*}
for any sequence of functions $\{ g_N \}_{N \in \mathbb N}$ from $\mathcal S$ to $\mathbb R$ that is uniformly bounded. Then, we have
\begin{align*}
    \lim_{N \to \infty} \left| \mathcal L(S_t^1, \mathbb G^N_{S_t})(T_N) - \mathcal L(S_t^{\mu}, \mu_t)(T_N) \right| = 0
\end{align*}
for any sequence of functions $\{ T_N \}_{N \in \mathbb N}$ from $\mathcal S \times \mathcal P(\mathcal S)$ to $\mathbb R$ that is equicontinuous and uniformly bounded.
\end{lemma}
\begin{subproof}
We have
\begin{align*}
    \left| \mathcal L(S_t^1, \mathbb G^N_{S_t})(T_N) - \mathcal L(S_t^{\mu}, \mu_t)(T_N) \right| 
    &\leq \left| \mathcal L(S_t^1, \mathbb G^N_{S_t})(T_N) - \mathcal L(S_t^1, \mu_t)(T_N) \right| + \left| \mathcal L(S_t^1, \mu_t)(T_N) - \mathcal L(S_t^{\mu}, \mu_t)(T_N) \right| 
\end{align*}
The first term becomes 
\begin{align*}
    \left| \mathcal L(S_t^1, \mathbb G^N_{S_t})(T_N) - \mathcal L(S_t^1, \mu_t)(T_N) \right|
    &= \left| \int T_N(x,\nu) \mathcal L(S_t^1, \mathbb G^N_{S_t})(dx,d\nu) - \int T_N(x,\nu) \mathcal L(S_t^1, \mu_t)(dx,d\nu) \right| \\
    &\leq \mathbb E \left[ \mathbb E \left[ \left| T_N(S_t^1, G^N_{S_t}) - T_N(S_t^1,\mu_t) \right| S_t^1 \right] \right] \\
    &\leq \mathbb E \left[ \sup_{f \in \{ T_N(x, \cdot) \}_{x \in \mathcal X, N \in \mathbb N}} \left| f(G^N_{S_t}) - f(\mu_t) \right| \right] \to 0
\end{align*}
by Lemma~\ref{lem:equicont}, since $\{ T_N \}_{N \in \mathbb N}$ is equicontinuous and uniformly bounded. Similarly for the second term,
\begin{align*}
    \left| \mathcal L(S_t^1, \mu_t)(T_N) - \mathcal L(S_t^{\mu}, \mu_t)(T_N) \right| 
    &= \left| \mathbb E \left[ T_N(S_t^1,\mu_t) - T_N(S_t^{\mu},\mu_t) \right] \right| \to 0
\end{align*}
by the assumption, since $T_N$ fulfills the condition of being uniformly bounded.
\end{subproof}

\begin{lemma} \label{lem:agent-conv}
For any sequence $\{ g_N \}_{N \in \mathbb N}$ of functions from $\mathcal S$ to $\mathbb R$ that is uniformly bounded, we have 
\begin{align*}
    \lim_{N \to \infty} \left| \mathcal L(S_t^1)(g_N) - \mathcal L(S_t^{\mu})(g_N) \right| = 0
\end{align*}
for all times $t \in \mathcal T$.
\end{lemma}
\begin{subproof}
Define $l_{N,t}$ as
\begin{align*}
    l_{N,t}(s, \nu) \equiv \sum_{a \in \mathcal A} \pi^N_t(a \mid s) \sum_{s' \in \mathcal S} p(s' \mid s, a, \nu) g_N(s') \, .
\end{align*}
$\{ l_{N,t}(s, \cdot) \}_{s \in \mathcal S, N \in \mathbb N}$ is equicontinuous, since for any $\nu, \nu' \in \mathcal M$ with $d_{TV}(\nu, \nu') \to 0$,
\begin{align*}
    \sup_{s \in \mathcal S, N \in \mathbb N} \left| l_{N,t}(s, \nu) - l_{N,t}(s, \nu') \right| &\leq M_g \sup_{s \in \mathcal S, N \in \mathbb N} \left| \sum_{a \in \mathcal A} \pi^N_t(a \mid s) \sum_{s' \in \mathcal S} \left( p(s' \mid s, a, \nu) - p(s' \mid s, a, \nu') \right) \right| \\ 
    &\leq M_g |\mathcal S| \max_{s \in \mathcal S} \max_{a \in \mathcal A} \max_{s' \in \mathcal S} | p(s' \mid s, a, \nu) - p(s' \mid s, a, \nu') | \to 0
\end{align*}
since $|g_N| < M_g$ is uniformly bounded and $p$ is continuous by assumption. Furthermore, $l_{N,t}(s, \nu)$ is always uniformly bounded by $M_g$. Now the result can be shown by induction.

For $t=0$, $\mathcal L(S_0^{\mu}) = \mathcal L(S_0^1)$ fulfills the hypothesis. Assume this holds for $t$, then
\begin{align*}
    \left| \mathcal L(S_{t+1}^1)(g_N) - \mathcal L(S_{t+1}^{\mu})(g_N) \right|
    &= \left| \mathcal L(S_t^1, \mathbb G^N_{S_t})(l_{N,t}) - \mathcal L(S_t^{\mu}, \mu_t)(l_{N,t}) \right| \to 0
\end{align*}
as $N \to \infty$ by Lemma~\ref{lem:equiconv}.
\end{subproof}

Thus, for any sequence of policies $\{ \pi^N \}_{N \in \mathbb N}$ with $\pi^N \in \Pi$ for all $N \in \mathbb N$, the achieved return of the $N$-agent game converges to the return of the mean field game under the mean field generated by the other agent's policy $\pi$ as $N \to \infty$.

\begin{lemma} \label{lem:obj-conv}
Let $\{ \pi^N \}_{N \in \mathbb N}$ with $\pi^N \in \Pi$ for all $N \in \mathbb N$ be an arbitrary sequence of policies and $\pi \in \Pi$ an arbitrary policy. Further, let the mean field $\mu = \Psi(\pi)$ be generated by $\pi$. Then, under the joint policy $(\pi^N, \pi, \ldots, \pi)$, we have as $N \to \infty$ that
\begin{align*}
    \left| J_1^N(\pi^N, \pi, \ldots, \pi) - J^{\mu}(\pi^N) \right| \to 0 \, .
\end{align*}
\end{lemma}
\begin{subproof}
Define for any $t \in \mathcal T$, $N \in \mathbb N$
\begin{align*}
    r_{\pi^N_t}(s, \nu) \equiv \sum_{a \in \mathcal A} r(s,a,\nu) \pi^N_t(a \mid s)
\end{align*}
such that the family $\{ r_{\pi^N_t}(s, \cdot) \}_{s \in \mathcal S, N \in \mathbb N}$ is equicontinuous, since for any $\nu_n, \nu' \in \mathcal M$ as $d_{\mathcal M}(\nu_n, \nu') \to 0$,
\begin{align*}
    \max_{s \in \mathcal S, N \in \mathbb N} \left| r_{\pi^N_t}(s, \nu_n) - r_{\pi^N_t}(s, \nu') \right| \leq \max_{s \in \mathcal S, a \in \mathcal A} \left|  r(s,a,\nu_n) -  r(s,a,\nu') \right| \to 0
\end{align*}
by continuity of $r$. The function $r_{\pi^N_t}$ is uniformly bounded for all $N \in \mathbb N$ by assumption of uniformly bounded $r$. By Lemma~\ref{lem:equiconv} and Lemma~\ref{lem:agent-conv}, 
\begin{align*}
    &\lim_{N \to \infty} \left| \mathbb E \left[ r(S_t^1, A_t^1, \mathbb G^N_{S_t}) \right] - \mathbb E \left[ r(S_t^{\mu}, A_t^{\mu}, \mu_t) \right] \right| \\
    &= \lim_{N \to \infty} \left| \mathbb E \left[ r_{\pi^N_t}(S_t^1, \mathbb G^N_{S_t}) \right] - \mathbb E \left[ r_{\pi^N_t}(S_t^{\mu}, \mu_t) \right] \right| = 0 \, .
\end{align*}
such that we have
\begin{align*}
    \lim_{N \to \infty} \left| J_1^N(\pi^N, \pi, \ldots, \pi) - J^{\mu}(\pi^N) \right|
    &\leq \sum_{t \in \mathcal T} \lim_{N \to \infty} \left | \mathbb E \left[ r(S_t^1, A_t^1, \mathbb G^N_{S_t}) \right] - \mathbb E \left[ r(S_t^{\mu}, A_t^{\mu}, \mu_t) \right] \right| = 0 \, .
\end{align*}
which is the desired result.
\end{subproof}

From Lemma~\ref{lem:obj-conv}, it follows that for any sequence of optimal exploiting policies $\{ \pi^N \}_{N \in \mathbb N}$ with $\pi^N \in \Pi$ for all $N \in \mathbb N$ and
\begin{align*}
    \pi^N \in \argmax_{\pi \in \Pi} J_1^N(\pi, \pi^*, \ldots, \pi^*)
\end{align*}
for all $N \in \mathbb N$, it holds that for any MFE $(\pi^*, \mu^*) \in \Pi \times \mathcal M$,
\begin{align*}
    \lim_{N \to \infty} J_1^N(\pi^N, \pi^*, \ldots, \pi^*)
    &\leq \max_{\pi \in \Pi} J^{\mu^*}(\pi) \\
    &= J^{\mu^*}(\pi^*) \\
    &= \lim_{N \to \infty} J_1^N(\pi^*, \ldots, \pi^*)
\end{align*}
and by instantiating for arbitrary $\epsilon > 0$, for sufficiently large $N$ we obtain 
\begin{align*}
    J_1^N(\pi^N, \pi^*, \ldots, \pi^*) - \epsilon &= \max_{\pi \in \Pi} J_1^N(\pi, \pi^*, \ldots, \pi^*) - \epsilon \\
    &\leq \max_{\pi \in \Pi} J^{\mu^*}(\pi) - \frac{\epsilon}{2} \\
    &= J^{\mu^*}(\pi^*) - \frac{\epsilon}{2} \\
    &= J_1^N(\pi^*, \pi^*, \ldots, \pi^*)
\end{align*}
which is the desired approximate Nash property that applies to all agents by symmetry.
\end{proof}

\subsection{Proof of Theorem~2} \label{app:int}
\begin{proof}
If $\Phi$ or $\Psi$ is constant, or if the restriction $\Psi \restriction_{\Pi_{\Phi}}$ of $\Psi$ to $\Pi_{\Phi}$ is constant, then $\Gamma = \Psi \circ \Phi$ is constant. Assume that this is not the case. 

Then there exist distinct $\pi, \pi' \in \Pi_{\Phi}$ such that $\Psi(\pi) \neq \Psi(\pi')$. By definition of $\Pi_{\Phi}$ there also exist distinct $\mu, \mu' \in \mathcal M$ such that $\Phi(\mu) = \pi$ and $\Phi(\mu') = \pi'$. Note that for any $\nu, \nu' \in \mathcal M$ with $\Gamma(\nu) \neq \Gamma(\nu')$,
\begin{align*}
    d_{\mathcal M}(\Gamma(\nu), \Gamma(\nu')) \geq \min_{\pi, \pi' \in \Pi_{\Phi}, \Psi(\pi) \neq \Psi(\pi')} d_{\mathcal M}(\Psi(\pi), \Psi(\pi'))
\end{align*} 
where the right-hand side is greater zero by finiteness of $\Pi_{\Phi}$. This holds for $\mu, \mu'$. 

To show that $\Gamma$ cannot be Lipschitz continuous, assume that $\Gamma$ has a Lipschitz constant $C > 0$. We can find an integer $N$ such that
\begin{align*}
    d_{\mathcal M}(\mu^i, \mu^{i+1}) &= \frac{d_{\mathcal M}(\mu, \mu')}{N-1} < \frac{\min_{\pi, \pi' \in \Pi_{\Phi}, \Psi(\pi) \neq \Psi(\pi')} d_{\mathcal M}(\Psi(\pi), \Psi(\pi'))}{C}
\end{align*} 
for all $i \in \{0, \ldots, N-1\}$ by defining 
\begin{align*}
    \mu^i = \frac{i}{N} \mu + \frac{N-i}{N} \mu'
\end{align*} 
for all $i \in \{0, \ldots, N\}$, and $\mu^i \in \mathcal M$ holds. By the triangle inequality 
\begin{align*}
    d_{\mathcal M}(\Gamma(\mu), \Gamma(\mu')) &\leq d_{\mathcal M}(\Gamma(\mu^{0}), \Gamma(\mu^{1})) + \ldots + d_{\mathcal M}(\Gamma(\mu^{N-1}), \Gamma(\mu^{N}))
\end{align*} 
there exists a pair $(\mu^i, \mu^{i+1})$ with $\Gamma(\mu^{i}) \neq \Gamma(\mu^{i+1})$. Therefore, for this pair, by the prequel
\begin{align*}
    d_{\mathcal M}(\Gamma(\mu^{i}), \Gamma(\mu^{i+1})) \geq \min_{\pi, \pi' \in \Pi_{\Phi}, \Psi(\pi) \neq \Psi(\pi')} d_{\mathcal M}(\Psi(\pi), \Psi(\pi')) \, .
\end{align*} 
On the other hand, since $\Gamma$ is Lipschitz with constant $C$, we have 
\begin{align*}
    d_{\mathcal M}(\Gamma(\mu^{i}), \Gamma(\mu^{i+1})) &\leq C \cdot d_{\mathcal M}(\mu^i, \mu^{i+1}) < \min_{\pi, \pi' \in \Pi_{\Phi}, \Psi(\pi) \neq \Psi(\pi')} d_{\mathcal M}(\Psi(\pi), \Psi(\pi'))
\end{align*}
which is a contradiction. Thus, $\Gamma$ cannot be Lipschitz continuous and by extension cannot be contractive.
\end{proof}

\subsection{Proof of Theorem~3} \label{app:th2}
\begin{proof}
For all $\eta > 0, \mu \in \mathcal M, t \in \mathcal T, s \in \mathcal S, a \in \mathcal A$, the soft action-value function of the MDP induced by $\mu \in \mathcal M$ is given by
\begin{align*}
    \tilde Q_\eta(\mu, t, s, a) = r(s, a, \mu_t) + \sum_{s' \in \mathcal S} p(s'\mid s, a, \mu_t) \eta \log \sum_{a' \in \mathcal A} q_{t+1}(a' \mid s') \exp \left( \frac{\tilde Q_\eta(\mu, t+1, s', a')}{\eta} \right)
\end{align*} 
and terminal condition $\tilde Q_\eta(\mu, T-1, s, a) \equiv r(s, a, \mu_{T-1})$. Analogously, the action-value function of the MDP induced by $\mu \in \mathcal M$ is given by
\begin{align*}
    Q^*(\mu, t,s,a) &= r(s,a,\mu_t) + \sum_{s' \in \mathcal S} p(s'\mid s, a, \mu_t) \max_{a' \in \mathcal A} Q^*(\mu, t+1,s',a')
\end{align*}
and the similarly defined policy action-value function for $\pi \in \Pi$ is given by
\begin{align*}
    Q^{\pi}(\mu, t,s,a) &= r(s,a,\mu_t) + \sum_{s' \in \mathcal S} p(s'\mid s, a, \mu_t) \sum_{a' \in \mathcal A} \pi_{t+1}(a' \mid s') Q^{\pi}(\mu, t+1,s',a') \, ,
\end{align*}
with terminal conditions $Q^*(\mu,T-1,s,a) \equiv Q^{\pi}(\mu,T-1,s,a) \equiv r(s,a,\mu_{T-1})$. 

We will show that we can find a Lipschitz constant $K_{\tilde Q_\eta}$ of $\tilde Q_\eta$ that is independent of $\eta$ if $\eta$ is not arbitrarily small. To show this, we will explicitly compute such a Lipschitz constant. Note first that $\tilde Q_\eta$, $Q^*$ and $Q^{\pi}$ are all uniformly bounded by $M_Q \equiv |\mathcal T| M_r$ by assumption, where $M_r$ is the uniform bound of $r$. 

\begin{lemma} \label{lem:qbound}
The functions $\tilde Q_\eta(\mu, t, s, a)$, $Q^*(\mu, t, s, a)$ and $Q^{\pi}(\mu, t, s, a)$ are uniformly bounded for all $\eta > 0, \mu \in \mathcal M, t \in \mathcal T, s \in \mathcal S, a \in \mathcal A$ by
\begin{align*}
    \left| \tilde Q_\eta(\mu, t, s, a) \right| \leq (T-t) M_r \leq T M_r =: M_Q
\end{align*}
where $M_r$ is the uniform bound of $| r(s, a, \mu_t) | \leq M_r$, and $T = |\mathcal T|$.
\end{lemma}
\begin{subproof}
Make the induction hypothesis for all $t \in \mathcal T$ that
\begin{align*}
    \left| \tilde Q_\eta(\mu, t, s, a) \right| \leq (T-t) M_r
\end{align*}
for all $\eta > 0, \mu \in \mathcal M, s \in \mathcal S, a \in \mathcal A$ and note that this holds for $t=T-1$, as by assumption
\begin{align*}
    \left| \tilde Q_\eta(\mu, T-1, s, a) \right| = | r(s, a, \mu_t) | \leq M_r \, .
\end{align*}
The induction step from $t+1$ to $t$ holds by
\begin{align*}
    \left| \tilde Q_\eta(\mu, t, s, a) \right| &= \left| r(s, a, \mu_t) + \sum_{s' \in \mathcal S} p(s'\mid s, a, \mu_t) \eta \log \sum_{a' \in \mathcal A} q_{t+1}(a' \mid s') \exp \left( \frac{\tilde Q_\eta(\mu, t+1, s', a')}{\eta} \right) \right| \\
    &\leq \left| r(s, a, \mu_t) \right| + \eta \max_{s' \in \mathcal S} \left| \log \sum_{a' \in \mathcal A} q_{t+1}(a' \mid s') \exp \left( \frac{\tilde Q_\eta(\mu, t+1, s', a')}{\eta} \right) \right| \\
    &\leq  M_r + \eta \left| \log \left( \exp \left( \frac{(T-t-1) M_r}{\eta} \right) \right) \right| \\
    &=  M_r + (T-t-1) M_r = (T-t) M_r \, .
\end{align*}
By maximizing over all $t \in \mathcal T$, we obtain the uniform bound. The other cases are analogous.
\end{subproof}

Now we can find a Lipschitz constant of $\tilde Q_\eta(\mu, t, s, a)$ that is independent of $\eta$.

\begin{lemma} \label{lem:qetalip}
Let $C_r$ be a Lipschitz constant of $\mu \to r(s,a,\mu_t)$ and $C_p$ a Lipschitz constant of $\mu \to p(s' \mid s,a,\mu_t)$. Further, let $\eta_{\mathrm{min}} > 0$. Then, for all $\eta > \eta_{\mathrm{min}}, t \in \mathcal T$, the map $\mu \mapsto \tilde Q_\eta(\mu, t, s, a)$ is Lipschitz for all $s \in \mathcal S, a \in \mathcal A$ with a Lipschitz constant $K_{\tilde Q_\eta}^t$ independent of $\eta$. Therefore, by picking $K_{\tilde Q_\eta} \equiv \max_{t \in \mathcal T} K_{\tilde Q_\eta}^t$, we have one single Lipschitz constant for all $\eta > \eta_{\mathrm{min}}, t \in \mathcal T, s \in \mathcal S, a \in \mathcal A$.
\end{lemma} 
\begin{subproof}
We show by induction that for all $t \in \mathcal T, s \in \mathcal S, a \in \mathcal A$, we can find Lipschitz constants such that $\tilde Q_\eta(\mu, t, s, a)$ is Lipschitz in $\mu$ with a Lipschitz constant that does not depend on $\eta$.

To see this, note that this is true for $t = T-1$ and any $s \in \mathcal S, a \in \mathcal A$, as for any $\mu, \mu'$ we have 
\begin{align*}
    \left| \tilde Q_\eta(\mu, T-1, s, a) - \tilde Q_\eta(\mu', T-1, s, a) \right| 
    &= \left| r(s, a, \mu_{T-1}) - r(s, a, \mu'_{T-1}) \right| \leq C_r d_{\mathcal M}(\mu, \mu') \, .
\end{align*}
The induction step from $t+1$ to $t$ is
\begin{align*}
    &\left| \tilde Q_\eta(\mu, t, s, a) - \tilde Q_\eta(\mu, t, s, a) \right| \\
    &\leq \left| r(s, a, \mu_t) - r(s, a, \mu'_t)\right| + \sum_{s' \in \mathcal S} \left| p(s'\mid s, a, \mu_t) \eta \log \sum_{a' \in \mathcal A} q_{t+1}(a' \mid s') \exp \left( \frac{\tilde Q_\eta(\mu, t+1, s', a')}{\eta} \right) \right. \\
    &\qquad \left. - p(s'\mid s, a, \mu'_t) \eta \log \sum_{a' \in \mathcal A} q_{t+1}(a' \mid s') \exp \left( \frac{\tilde Q_\eta(\mu', t+1, s', a')}{\eta} \right) \right| \\
    &\leq C_r d_{\mathcal M}(\mu, \mu') + \eta |\mathcal S| \max_{s' \in \mathcal S} 1 \cdot \left| \log \sum_{a' \in \mathcal A} q_{t+1}(a' \mid s') \exp \left( \frac{\tilde Q_\eta(\mu, t+1, s', a')}{\eta} \right) \right. \\
    &\qquad \left. - \log \sum_{a' \in \mathcal A} q_{t+1}(a' \mid s') \exp \left( \frac{\tilde Q_\eta(\mu', t+1, s', a')}{\eta} \right) \right| \\
    &\quad + \eta |\mathcal S| \max_{s' \in \mathcal S} \frac{M_Q}{\eta} \cdot \left| p(s'\mid s, a, \mu_t) - p(s'\mid s, a, \mu'_t) \right| \\
    &\leq C_r d_{\mathcal M}(\mu, \mu') + \eta |\mathcal S| \max_{s' \in \mathcal S} \sum_{a' \in \mathcal A} \left| \frac{ \frac{1}{\eta} q_{t+1}(a' \mid s') \exp \left( \frac{\xi_{a'}}{\eta} \right) }{ \sum_{a'' \in \mathcal A} q_{t+1}(a'' \mid s') \exp \left( \frac{\xi_{a''}}{\eta} \right) } \right| \left| \tilde Q_\eta(\mu, t+1, s', a') - \tilde Q_\eta(\mu', t+1, s', a') \right| \\
    &\quad + |\mathcal S| M_Q \cdot C_p d_{\mathcal M}(\mu, \mu') \\
    &\leq C_r d_{\mathcal M}(\mu, \mu') + \frac{| \mathcal A | q_{\mathrm{max}}}{| \mathcal A | q_{\mathrm{min}}} \exp \left(2 \cdot \frac{M_Q}{\eta} \right) K_{\tilde Q_\eta}^{t+1} d_{\mathcal M}(\mu, \mu') + |\mathcal S| M_Q C_p d_{\mathcal M}(\mu, \mu') \\
    &< \left( C_r + \frac{q_{\mathrm{max}}}{q_{\mathrm{min}}} \exp \left( \frac{2 M_Q}{\eta_{\mathrm{min}}} \right) K_{\tilde Q_\eta}^{t+1} + |\mathcal S| M_Q C_p \right) d_{\mathcal M}(\mu, \mu')
\end{align*}
where we use the mean value theorem to obtain some $\xi_a \in [-M_Q, M_Q]$ for all $a \in \mathcal A$ bounded by Lemma~\ref{lem:qbound}, Lemma~\ref{lem:lip} for the second inequality, and defined $q_{\mathrm{max}} = \max_{t \in \mathcal T, s \in \mathcal S, a \in \mathcal A} q_t(a \mid s)$, $q_{\mathrm{min}} = \min_{t \in \mathcal T, s \in \mathcal S, a \in \mathcal A} q_t(a \mid s)$. Since $s \in \mathcal S, a \in \mathcal A$ were arbitrary, this holds for all $s \in \mathcal S, a \in \mathcal A$.

Thus, as long as $\eta > \eta_{\mathrm{min}}$, we have the Lipschitz constant $K_{\tilde Q_\eta}^{t} \equiv \left( C_r + \frac{q_{\mathrm{max}}}{q_{\mathrm{min}}} \exp \left( \frac{2 M_Q}{\eta_{\mathrm{min}}} \right) K_{\tilde Q_\eta}^{t+1} + |\mathcal S| M_Q C_p \right) $ independent of $\eta$, since by induction assumption $K_{\tilde Q_\eta}^{t+1}$ is independent of $\eta$.
\end{subproof}

The optimal action-value function and the policy action-value function for any fixed policy are Lipschitz in $\mu$.

\begin{lemma} \label{lem:qlip}
The functions $\mu \mapsto Q^*(\mu, t, s, a)$ and $\mu \mapsto Q^\pi(\mu, t, s, a)$ for any fixed $\pi \in \Pi, t \in \mathcal T, s \in \mathcal S, a \in \mathcal A$ are Lipschitz continuous. Therefore, for any fixed $\pi \in \Pi$ we can choose a Lipschitz constant $K_Q$ for all $t \in \mathcal T, s \in \mathcal S, a \in \mathcal A$ by taking the maximum over all Lipschitz constants.
\end{lemma}
\begin{subproof}
The action-value function is given by the recursion
\begin{align*}
    Q^*(\mu, t,s,a) &= r(s,a,\mu_t) + \sum_{s' \in \mathcal S} p(s'\mid s, a, \mu_t) \max_{a' \in \mathcal A} Q^*(\mu, t+1,s',a')
\end{align*}
with terminal condition $Q^*(\mu,T-1,s,a) \equiv r(s,a,\mu_{T-1})$. The functions $r(s,a,\mu_t)$ and $p(s' \mid s,a,\mu_t)$ are Lipschitz continuous by Assumption~2. Note that for any $\mu, \mu' \in \mathcal M$ and any $t \in \mathcal T$, $d_{TV}(\mu_t, \mu'_t) \leq d_{M}(\mu, \mu')$. Therefore, the terminal condition and all terms in the above recursion are Lipschitz. Further, $Q^*(\mu, t,s,a)$ is uniformly bounded, since $r$ is assumed uniformly bounded.

Since a finite maximum, product and sum of Lipschitz and bounded functions is again Lipschitz and bounded by Lemma~\ref{lem:lip}, we obtain Lipschitz constants $K_{Q,t,s,a}$ of the maps $\mu \to Q^*(\mu, t, s, a)$ for any $t \in \mathcal T, s \in \mathcal S, a \in \mathcal A$ and define $K_Q \equiv \max_{t \in \mathcal T,s \in \mathcal S,a \in \mathcal A} K_{Q,t,s,a}$. The case for $Q^\pi$ with fixed $\pi \in \Pi$ is analogous.
\end{subproof}

The same holds for $\Psi(\pi)$ mapping from policy $\pi$ to its induced mean field.

\begin{lemma} \label{lem:mulip}
The function $\Psi(\pi)$ is Lipschitz with some Lipschitz constant $K_{\Psi}$.
\end{lemma}
\begin{subproof}
Recall that $\Psi(\pi)$ maps to the mean field $\mu$ starting with $\mu_0$ and obtained by the recursion
\begin{align*}
    \mu_{t+1}(s') = \sum_{s \in \mathcal S} \sum_{a \in \mathcal A} p(s' \mid s, a, \mu_t) \pi_t(a \mid s) \mu_t(s) \, .
\end{align*}
We proceed analogously to Lemma~\ref{lem:qlip}. $\mu$ is uniformly bounded by normalization. The constant function $\pi \mapsto \mu_0(s)$ is Lipschitz and bounded for any $s \in \mathcal S$. The functions $r(s,a,\mu_t)$ and $p(s' \mid s,a,\mu_t)$ are Lipschitz continuous by Assumption~2. Since a finite sum, product and composition of Lipschitz and bounded functions is again Lipschitz and bounded by Lemma~\ref{lem:lip}, we obtain Lipschitz constants $K_{\Psi,t,s}$ of the maps $\pi \to \mu_t(s)$ for any $t \in \mathcal T, s \in \mathcal S$ and define $K_{\Psi} \equiv \max_{t \in \mathcal T,s \in \mathcal S} K_{\Psi,t,s}$, which is the desired Lipschitz constant of $\Psi$. 
\end{subproof}

Finally, the map from an energy function to its associated Boltzmann distribution is Lipschitz for any $\eta > 0$ with a Lipschitz constant explicitly depending on $\eta$.

\begin{lemma} \label{lem:boltzmannlip}
Let $\eta > 0$ arbitrary and $f_a: \mathcal M \to \mathbb R$ be a Lipschitz continuous function with Lipschitz constant $K_f$ for any $a \in \mathcal A$. Further, let $g: \mathcal A \to \mathbb R$ be bounded by $g_{\mathrm{max}} > g(a) > g_{\mathrm{min}} > 0$ for any $a \in \mathcal A$. The function 
\begin{align*}
    \mu \mapsto \frac{ g(a) \exp \left( \frac{f_{a}(\mu)}{\eta} \right) }{ \sum_{a' \in \mathcal A} g(a') \exp \left( \frac{f_{a'}(\mu)}{\eta} \right) }
\end{align*}
is Lipschitz with Lipschitz constant $K = \frac{(\left| \mathcal A \right| - 1) K_f g_{\mathrm{max}}^2}{ 2 \eta g_{\mathrm{min}}^2}$ for any $a \in \mathcal A$.
\end{lemma}
\begin{subproof}
Let $\mu, \mu' \in \mathcal M$ be arbitrary and define
\begin{align*}
    \Delta_a f_{a'}(\mu) \equiv f_{a'}(\mu) - f_a(\mu) 
\end{align*}
for any $a' \in \mathcal A$, which is Lipschitz with constant $2K_f$. Then, we have
\begin{align*}
    &\left| \frac{ g(a) \exp \left( \frac{f_a(\mu)}{\eta} \right) }{ \sum_{a' \in \mathcal A} g(a') \exp \left( \frac{f_{a'}(\mu)}{\eta} \right) } - \frac{  g(a) \exp \left( \frac{f_a(\mu')}{\eta} \right) }{ \sum_{a' \in \mathcal A} g(a') \exp \left( \frac{f_{a'}(\mu')}{\eta} \right) } \right| \\
    &= \left| \frac{ 1 }{ 1 + \sum_{a' \neq a} \frac{g(a')}{g(a)} \exp \left( \frac{\Delta_a f_{a'}(\mu)}{\eta} \right) } - \frac{ 1 }{ 1 + \sum_{a' \neq a} \frac{g(a')}{g(a)} \exp \left( \frac{\Delta_a f_{a'}(\mu')}{\eta}  \right) } \right| \\
    &\leq \left| \sum_{a' \neq a} \frac{ \frac{g(a')}{g(a)} \cdot \frac{1}{\eta} \exp \left( \frac{\xi_{a'}}{\eta} \right) }{ \left( 1 + \sum_{a'' \neq a} \frac{g(a'')}{g(a)} \exp \left( \frac{\xi_{a''}}{\eta} \right) \right)^2 } \cdot \left( \Delta_a f_{a'}(\mu) - \Delta_a f_{a'}(\mu')  \right) \right| \\
    &\leq  \sum_{a' \neq a} \left| \frac{ \frac{g_{\mathrm{max}}}{g_{\mathrm{min}}} \cdot \frac{1}{\eta} \exp \left( \frac{\xi_{a'}}{\eta} \right) }{ \left( 1 + \frac{g_{\mathrm{min}}}{g_{\mathrm{max}}} \exp \left( \frac{\xi_{a'}}{\eta} \right) \right)^2 } \right| \cdot \left| \Delta_a f_{a'}(\mu) - \Delta_a f_{a'}(\mu') \right| \\
    &\leq \frac{g_{\mathrm{max}}^2}{4 \eta g_{\mathrm{min}}^2} \cdot \sum_{a' \neq a} 2 K_f d_{\mathcal M}(\mu, \mu') = \frac{(\left| \mathcal A \right| - 1) K_f g_{\mathrm{max}}^2}{ 2 \eta g_{\mathrm{min}}^2} \cdot d_{\mathcal M}(\mu, \mu')
\end{align*}
where we applied the mean value theorem to obtain some $\xi_{a'} \in \mathbb R$ for all $a' \in \mathcal A$ and used the maximum $\frac{1}{4c}$ of the function $\tilde f(x) = \frac{\exp (x/\eta)}{(1 + c \cdot \exp (x/\eta))^2}$ at $x = 0$.
\end{subproof}

For RelEnt MFE, by Lemma~\ref{lem:qetalip} we obtain a Lipschitz constant $K_{\tilde Q_\eta}$ of $\mu \to \tilde Q_\eta(\mu, t, s, a)$ as long as $\eta > \eta_{\mathrm{min}}$ for some $\eta_{\mathrm{min}} > 0$. Furthermore, note that for $\tilde \pi^{\mu,\eta} \equiv \tilde \Phi_\eta(\mu)$, we have
\begin{align*}
    \left| \tilde \pi_t^{\mu,\eta}(a \mid s) - \tilde \pi_t^{\mu',\eta}(a \mid s)) \right| &= \left| \frac{ q_{t}(a \mid s) \exp \left( \frac{\tilde Q_\eta(\mu, t, s, a)}{\eta} \right) }{ \sum_{a' \in \mathcal A} q_{t}(a' \mid s) \exp \left( \frac{\tilde Q_\eta(\mu, t, s, a')}{\eta} \right) } - \frac{ q_{t}(a \mid s) \exp \left( \frac{\tilde Q_\eta(\mu', t, s, a)}{\eta} \right) }{ \sum_{a' \in \mathcal A} q_{t}(a' \mid s) \exp \left( \frac{\tilde Q_\eta(\mu', t, s, a')}{\eta} \right) } \right| \, .
\end{align*}
We obtain the Lipschitz constant of $\tilde \Phi_\eta$ by applying Lemma~\ref{lem:boltzmannlip} to each of the maps given by
\begin{align*}
    \mu \mapsto \frac{ q_{t}(a \mid s) \exp \left( \frac{\tilde Q_\eta(\mu, t, s, a)}{\eta} \right) }{ \sum_{a' \in \mathcal A} q_{t}(a' \mid s) \exp \left( \frac{\tilde Q_\eta(\mu, t, s, a')}{\eta} \right) }
\end{align*} 
for all $t \in \mathcal T, s \in \mathcal S, a \in \mathcal A$, resulting in the Lipschitz property
\begin{align*}
    d_{\Pi}(\tilde \Phi_\eta(\mu), \tilde \Phi_\eta(\mu')) &= \max_{s \in \mathcal S} \max_{t \in \mathcal T} \sum_{a \in \mathcal A} \left| \tilde \pi_t^{\mu,\eta}(a \mid s) - \tilde \pi_t^{\mu',\eta}(a \mid s)) \right| \\
    &\leq \sum_{a \in \mathcal A}\frac{(\left| \mathcal A \right| - 1) K_{\tilde Q_\eta} q_{\mathrm{max}}^2}{ 2 \eta q_{\mathrm{min}}^2} \cdot d_{\mathcal M}(\mu, \mu') = \frac{\left| \mathcal A \right| (\left| \mathcal A \right| - 1) K_{\tilde Q_\eta} q_{\mathrm{max}}^2}{ 2 \eta q_{\mathrm{min}}^2} \cdot d_{\mathcal M}(\mu, \mu') \, ,
\end{align*} 
where we define $q_{\mathrm{max}} = \max_{t \in \mathcal T, s \in \mathcal S, a \in \mathcal A} q_t(a \mid s)$ and analogously $q_{\mathrm{min}} = \min_{t \in \mathcal T, s \in \mathcal S, a \in \mathcal A} q_t(a \mid s)$.

By Lemma~\ref{lem:mulip}, $\Psi(\pi)$ is Lipschitz with some Lipschitz constant $K_{\Psi}$. Therefore, the resulting Lipschitz constant of the composition $\tilde \Gamma_\eta = \Psi \circ \tilde \Phi_\eta$ is $\frac{\left| \mathcal A \right| (\left| \mathcal A \right| - 1) K_{\tilde Q_\eta} K_{\Psi} q_{\mathrm{max}}^2}{ 2 \eta q_{\mathrm{min}}^2}$ and leads to a contraction for any
\begin{align*}
    \eta > \max \left( \eta_{\mathrm{min}}, \frac{\left| \mathcal A \right| (\left| \mathcal A \right| - 1) K_{\tilde Q_\eta} K_{\Psi} q_{\mathrm{max}}^2}{ 2 q_{\mathrm{min}}^2} \right) \, .
\end{align*}

Analogously for Boltzmann MFE, by Lemma~\ref{lem:qlip} the mapping $\mu \to Q^*(\mu, t, s, a)$ is Lipschitz with some Lipschitz constant $K_{Q^*}$ for all $t \in \mathcal T, s \in \mathcal S, a \in \mathcal A$. For $\pi^{\mu,\eta} \equiv \Phi_\eta(\mu)$, we have
\begin{align*}
    \left|  \pi_t^{\mu,\eta}(a \mid s) - \pi_t^{\mu',\eta}(a \mid s)) \right| &= \left| \frac{ q_{t}(a \mid s) \exp \left( \frac{Q^*(\mu, t, s, a)}{\eta} \right) }{ \sum_{a' \in \mathcal A} q_{t}(a' \mid s) \exp \left( \frac{Q^*(\mu, t, s, a')}{\eta} \right) } - \frac{ q_{t}(a \mid s) \exp \left( \frac{Q^*(\mu', t, s, a)}{\eta} \right) }{ \sum_{a' \in \mathcal A} q_{t}(a' \mid s) \exp \left( \frac{Q^*(\mu', t, s, a')}{\eta} \right) } \right| \, .
\end{align*}
We obtain the Lipschitz constant of $\Phi_\eta$ by applying Lemma~\ref{lem:boltzmannlip} to each of the maps given by 
\begin{align*}
    \mu \mapsto \frac{ q_{t}(a \mid s) \exp \left( \frac{Q^*(\mu, t, s, a)}{\eta} \right) }{ \sum_{a' \in \mathcal A} q_{t}(a' \mid s) \exp \left( \frac{Q^*(\mu, t, s, a')}{\eta} \right) }
\end{align*}
for all $t \in \mathcal T, s \in \mathcal S, a \in \mathcal A$, resulting in the Lipschitz property
\begin{align*}
    d_{\Pi}(\Phi_\eta(\mu), \Phi_\eta(\mu')) &= \max_{s \in \mathcal S} \max_{t \in \mathcal T} \sum_{a \in \mathcal A} \left| \pi_t^{\mu,\eta}(a \mid s) - \pi_t^{\mu',\eta}(a \mid s)) \right| \\
    &\leq \sum_{a \in \mathcal A}\frac{(\left| \mathcal A \right| - 1) K_{Q^*} q_{\mathrm{max}}^2}{ 2 \eta q_{\mathrm{min}}^2} \cdot d_{\mathcal M}(\mu, \mu') = \frac{\left| \mathcal A \right| (\left| \mathcal A \right| - 1) K_{Q^*} q_{\mathrm{max}}^2}{ 2 \eta q_{\mathrm{min}}^2} \cdot d_{\mathcal M}(\mu, \mu') \, .
\end{align*} 
By Lemma~\ref{lem:mulip}, $\Psi(\pi)$ is Lipschitz with some Lipschitz constant $K_{\Psi}$. The resulting Lipschitz constant of the composition $\Gamma_\eta = \Psi \circ \Phi_\eta$ is $\frac{\left| \mathcal A \right| (\left| \mathcal A \right| - 1) K_{Q^*} K_{\Psi} q_{\mathrm{max}}^2}{ 2 \eta q_{\mathrm{min}}^2}$ and leads to a contraction for any 
\begin{align*}
    \eta > \frac{\left| \mathcal A \right| (\left| \mathcal A \right| - 1) K_{Q^*} K_{\Psi} q_{\mathrm{max}}^2}{ 2 q_{\mathrm{min}}^2}
\end{align*} 
where for the uniform prior policy, $q_{\mathrm{max}} = q_{\mathrm{min}}$. If required, the Lipschitz constants can be computed recursively according to Lemma~\ref{lem:lip}.
\end{proof}

\subsection{Proof of Theorem~4} \label{app:epsepsopt}
\begin{proof}
Consider any sequence $(\pi^*_{n}, \mu^*_{n})_{n \in \mathbb N}$ of ${\eta_n}$-Boltzmann or ${\eta_n}$-RelEnt MFE with $\eta_n \to 0^+$ as $n \to \infty$. Note that a pair $(\pi^*_{n}, \mu^*_{n})$ is completely specified by $\mu^*_{n}$, since $\pi^*_{n} = \Phi_{\eta_n}(\mu^*_{n})$ or $\pi^*_{n} = \tilde \Phi_{\eta_n}(\mu^*_{n})$ uniquely. Therefore, it suffices to show that the associated functions $(\mu \mapsto Q^{\Phi_{\eta_n}(\mu)}(\mu, t, s, a))_{n \in \mathbb N}$ and $(\mu \mapsto Q^{\tilde \Phi_{\eta_n}(\mu)}(\mu, t, s, a))_{n \in \mathbb N}$ converge uniformly to $\mu \mapsto Q^*(\mu, t, s, a)$, from which the desired result will follow. For definitions of the different action-value functions, see Appendix~\ref{app:th2}.

Note that pointwise convergence is insufficient, since there is no guarantee that $\mu^*_{n}$ itself will converge as $n \to \infty$. However, we can obtain uniform convergence by pointwise convergence and equicontinuity. For RelEnt MFE, we will additionally require uniform convergence of the sequence $(\mu \mapsto \tilde Q_{\eta_n}(\mu, t, s, a))_{n \in \mathbb N}$ with $\eta_n \to 0^+$. We begin with pointwise convergence of $(\mu \mapsto Q^{\Phi_{\eta_n}(\mu)}(\mu, t, s, a))_{n \in \mathbb N}$ to the optimal action-value function $\mu \mapsto Q^*(\mu, t, s, a)$.

\begin{lemma} \label{lem:boltzmann-pointwise}
Any sequence of functions $(\mu \mapsto Q^{\Phi_{\eta_n}(\mu)}(\mu, t, s, a))_{n \in \mathbb N}$ with $\eta_n \to 0^+$ converges pointwise to $\mu \mapsto Q^*(\mu, t, s, a)$ for all $t \in \mathcal T, s \in \mathcal S, a \in \mathcal A$.
\end{lemma}
\begin{subproof}
Fix $\mu \in \mathcal M$. We make the induction hypothesis for arbitrary $t \in \mathcal T$ that for all $s \in \mathcal S, a \in \mathcal A, \varepsilon > 0$, there exists $n' \in \mathbb N$ such that for any $n > n'$ we have
\begin{align*}
    \left| Q^{\Phi_{\eta_n}(\mu)}(\mu, t,s,a) - Q^*(\mu, t,s,a) \right| < \varepsilon \, .
\end{align*}

The induction hypothesis is fulfilled for $t=T-1$, as by definition
\begin{align*}
    \left| Q^{\Phi_{\eta_n}(\mu)}(\mu, t,s,a) - Q^*(\mu, t,s,a) \right| = \left| r(s,a,\mu_t) - r(s,a,\mu_t) \right| = 0 \, .
\end{align*}

Assume that the induction hypothesis is fulfilled for $t+1$, then at time $t$ let $s \in \mathcal S, a \in \mathcal A, \varepsilon > 0$ arbitrary. Furthermore, let $s' \in \mathcal S$ arbitrary. Collect all optimal actions into a set $\mathcal A_{\mathrm{opt}}^{s'} \subseteq \mathcal A$, i.e. for $a' \in \mathcal A_{\mathrm{opt}}^{s'}$ we have
\begin{align*}
    Q^*(\mu, t, s', a_{\mathrm{opt}}) = \max_{a \in \mathcal A} Q^*(\mu, t, s', a) \, .
\end{align*}

We define the minimal action gap
\begin{align*}
    \Delta Q_{\mathrm{min}}^{s', \mu} \equiv \min_{a_{\mathrm{opt}} \in \mathcal A_{\mathrm{opt}}^{s'}, a_{\mathrm{sub}} \in \mathcal A \setminus \mathcal A_{\mathrm{opt}}^{s'}} \left( Q^*(\mu, t, s', a_{\mathrm{opt}}) - Q^*(\mu, t, s', a_{\mathrm{sub}}) \right) > 0
\end{align*}
such that for arbitrary suboptimal actions $a_{\mathrm{sub}} \in \mathcal A \setminus \mathcal A_{\mathrm{opt}}^{s'}$ and optimal actions $a_{\mathrm{opt}} \in \mathcal A_{\mathrm{opt}}^{s'}$,
\begin{align*}
    Q^*(\mu, t, s', a_{\mathrm{opt}}) - Q^*(\mu, t, s', a_{\mathrm{sub}}) \geq \Delta Q_{\mathrm{min}}^{s', \mu} \, .
\end{align*}

This is well defined if there are suboptimal actions, since there is always at least one optimal action. If all actions are optimal, we can skip bounding the probability of taking suboptimal actions and the result will hold trivially. Thus, we assume henceforth that there exists a suboptimal action.

It follows that the probability of taking suboptimal actions $a_{\mathrm{sub}} \in \mathcal A \setminus \mathcal A_{\mathrm{opt}}^{s'}$ disappears, since
\begin{align*}
    (\Phi_{\eta_n}(\mu))_t(a_{\mathrm{sub}} \mid s')
    &= \frac{ q_t(a_{\mathrm{sub}} \mid s) }{ \sum_{a' \in \mathcal A} q_t(a' \mid s) \exp \left( \frac{Q^*(\mu, t, s, a') - Q^*(\mu, t, s, a_{\mathrm{sub}})}{\eta} \right) } \\
    &\leq \frac{ 1 }{ 1 + \sum_{a' \in \mathcal A} \frac{q_t(a' \mid s)}{q_t(a_{\mathrm{sub}} \mid s)} \exp \left( \frac{Q^*(\mu, t, s, a') - Q^*(\mu, t, s, a_{\mathrm{sub}})}{\eta} \right) } \\
    &\leq \frac{ 1 \mid s) }{ 1 + \frac{q_t(a_{\mathrm{opt}} \mid s)}{q_t(a_{\mathrm{sub}} \mid s)} \exp \left( \frac{Q^*(\mu, t, s, a_{\mathrm{opt}}) - Q^*(\mu, t, s, a_{\mathrm{sub}})}{\eta} \right) } \\
    &\leq \frac{ 1 \mid s) }{ 1 + \frac{q_t(a_{\mathrm{opt}} \mid s)}{q_t(a_{\mathrm{sub}} \mid s)} \exp \left( \frac{\Delta Q_{\mathrm{min}}^{s', \mu}}{\eta} \right) } \to 0
\end{align*}
as $\eta \to 0^+$ for some arbitrary optimal action $a_{\mathrm{opt}} \in \mathcal A_{\mathrm{opt}}^{s'}$. Since $s' \in \mathcal S$ was arbitrary, this holds for all $s' \in \mathcal S$. Therefore, by finiteness of $\mathcal S$ and $\mathcal A$ we can choose $n_1 \in \mathbb N$ such that for all $n > n_1$ and for all $a_{\mathrm{sub}} \in \mathcal A \setminus \mathcal A_{\mathrm{opt}}^{s'}$ we have $\eta_n$ sufficiently small such that
\begin{align*}
    (\Phi_{\eta_n}(\mu))_t(a_{\mathrm{sub}} \mid s') < \frac{\varepsilon}{2|\mathcal A|M_Q}
\end{align*}
where $M_Q$ is the uniform bound of $Q^{\Phi_{\eta_n}(\mu)}$.

Further, by induction assumption, we can choose $n_{s',a'}$ for any $s' \in \mathcal S, a' \in \mathcal A$ such that for all $n > n_{s',a'}$ we have
\begin{align*}
    \left| Q^{\Phi_{\eta_n}(\mu)}(\mu, t+1,s',a') - Q^*(\mu, t+1,s',a') \right| < \frac{\varepsilon}{3}
\end{align*}

Therefore, as long as $n > n' \equiv \max(n_1, \max_{s' \in \mathcal S, a' \in \mathcal A} n_{s',a'}) $, we have
\begin{align*}
    &\left| Q^{\Phi_{\eta_n}(\mu)}(\mu, t,s,a) - Q^*(\mu, t,s,a) \right| \\
    &= \left| \sum_{s' \in \mathcal S} p(s'\mid s, a, \mu_t) \left( \sum_{a' \in \mathcal A} (\Phi_{\eta_n}(\mu))_t(a' \mid s') Q^{\Phi_{\eta_n}(\mu)}(\mu, t+1,s',a') - \max_{a'' \in \mathcal A} Q^*(\mu, t+1,s',a'') \right) \right| \\
    &\leq \max_{s' \in \mathcal S} \left| \sum_{a' \in \mathcal A} (\Phi_{\eta_n}(\mu))_t(a' \mid s') Q^{\Phi_{\eta_n}(\mu)}(\mu, t+1,s',a') - \max_{a'' \in \mathcal A} Q^*(\mu, t+1,s',a'') \right| \\
    &\leq \max_{s' \in \mathcal S} \left| \sum_{a' \in \mathcal A_{\mathrm{opt}}^{s'}} (\Phi_{\eta_n}(\mu))_t(a' \mid s') Q^{\Phi_{\eta_n}(\mu)}(\mu, t+1,s',a') - \max_{a'' \in \mathcal A} Q^*(\mu, t+1,s',a'') \right| \\
    &\quad + \max_{s' \in \mathcal S} \left| \sum_{a' \in \mathcal A \setminus \mathcal A_{\mathrm{opt}}^{s'}} (\Phi_{\eta_n}(\mu))_t(a' \mid s') Q^{\Phi_{\eta_n}(\mu)}(\mu, t+1,s',a') \right| \\
    &\leq \max_{s' \in \mathcal S} \left| \sum_{a' \in \mathcal A_{\mathrm{opt}}^{s'}} (\Phi_{\eta_n}(\mu))_t(a' \mid s') Q^{\Phi_{\eta_n}(\mu)}(\mu, t+1,s',a') - \sum_{a' \in \mathcal A_{\mathrm{opt}}^{s'}} (\Phi_{\eta_n}(\mu))_t(a' \mid s') \max_{a'' \in \mathcal A} Q^*(\mu, t+1,s',a'') \right| \\
    &\quad + \max_{s' \in \mathcal S} \left| \sum_{a' \in \mathcal A_{\mathrm{opt}}^{s'}} (\Phi_{\eta_n}(\mu))_t(a' \mid s') \max_{a'' \in \mathcal A} Q^*(\mu, t+1,s',a'') - \max_{a'' \in \mathcal A} Q^*(\mu, t+1,s',a'') \right| \\
    &\quad + \max_{s' \in \mathcal S} \left| \sum_{a' \in \mathcal A \setminus \mathcal A_{\mathrm{opt}}^{s'}} (\Phi_{\eta_n}(\mu))_t(a' \mid s') Q^{\Phi_{\eta_n}(\mu)}(\mu, t+1,s',a') \right| \\
    &\leq \max_{s' \in \mathcal S} \max_{a' \in \mathcal A_{\mathrm{opt}}^{s'}} \left|  Q^{\Phi_{\eta_n}(\mu)}(\mu, t+1,s',a') - \max_{a'' \in \mathcal A} Q^*(\mu, t+1,s',a'') \right| \\
    &\quad + \max_{s' \in \mathcal S} M_Q \left| -\sum_{a' \in \mathcal A \setminus \mathcal A_{\mathrm{opt}}^{s'}} (\Phi_{\eta_n}(\mu))_t(a' \mid s') \right| + \max_{s' \in \mathcal S} M_Q \left| \sum_{a' \in \mathcal A \setminus \mathcal A_{\mathrm{opt}}^{s'}} (\Phi_{\eta_n}(\mu))_t(a' \mid s') \right| \\
    &< \frac{\varepsilon}{3} + \frac{\varepsilon}{3|\mathcal A|M_Q} \cdot |\mathcal A|M_Q + \frac{\varepsilon}{3|\mathcal A|M_Q} \cdot |\mathcal A|M_Q = \varepsilon \, .
\end{align*}

Since $s \in \mathcal S, a \in \mathcal A, \varepsilon > 0$ were arbitrary, the desired result follows immediately by induction.
\end{subproof}

As we have no control over $\mu^*_{n}$ and the sequence $(\pi^*_{n}, \mu^*_{n})_{n \in \mathbb N}$ may not even converge, pointwise convergence is insufficient. To obtain uniform convergence, we shall use compactness of $\mathcal M$ and equicontinuity.

\begin{lemma} \label{lem:boltzmann-equi}
The family of functions $\mathcal F \equiv \{\mu \mapsto Q^{\Phi_{\eta}(\mu)}(\mu, t, s, a)\}_{\eta > 0, t \in \mathcal T, s \in \mathcal S, a \in \mathcal A}$ is equicontinuous, i.e. for any $\varepsilon > 0$ and any $\mu \in \mathcal M$, we can choose a $\delta > 0$ such that for all $\mu' \in \mathcal M$ with $d_{\mathcal M}(\mu, \mu') < \delta$ and any $f \in \mathcal F$ we have
\begin{align*}
    \left| f(\mu) - f(\mu') \right| < \varepsilon \, .
\end{align*}
\end{lemma}
\begin{subproof}
Fix an arbitrary $\mu \in \mathcal M$. We make the (backwards in time) induction hypothesis for all $t \in \mathcal T$ that for any $s \in \mathcal S, a \in \mathcal A, \varepsilon_{t,s,a} > 0$, there exists $\delta_{t,s,a} > 0$ such that for any $\mu' \in \mathcal M$ with $d_{\mathcal M}(\mu, \mu') < \delta_{t,s,a}$ and any $f \in \mathcal F$ we have
\begin{align*}
    \left| Q^{\Phi_{\eta}(\mu)}(\mu, t,s,a) - Q^{\Phi_{\eta}(\mu')}(\mu', t,s,a) \right| < \varepsilon_{t,s,a} \, .
\end{align*}

The induction hypothesis is fulfilled for $t=T-1$, as by assumption, $\nu \to r(s,a,\nu_t)$ is Lipschitz with constant $C_r > 0$. Therefore, for all $s \in \mathcal S, a \in \mathcal A$ we can choose $\delta_{T-1,s,a} = \frac{\varepsilon_{t,s,a}}{C_r}$ such that for any $\mu, \mu'$ with $d_{\mathcal M}(\mu, \mu') < \delta'$ we have
\begin{align*}
    \left| Q^{\Phi_{\eta}(\mu)}(\mu, t,s,a) - Q^{\Phi_{\eta}(\mu')}(\mu', t,s,a) \right| 
    &= \left| r(s,a,\mu_t) - r(s,a,\mu'_t) \right| \leq C_r d_{\mathcal M}(\mu, \mu') < \varepsilon_{t,s,a} \, .
\end{align*}

Assume that the induction hypothesis holds for $t+1$, then at time $t$ let $\varepsilon_{t,s,a} > 0, s \in \mathcal S, a \in \mathcal A$ arbitrary. By definition, we have
\begin{align*}
    &\left| Q^{\Phi_{\eta}(\mu)}(\mu, t,s,a) - Q^{\Phi_{\eta}(\mu')}(\mu', t,s,a) \right| \\
    &= \left| r(s,a,\mu_t) + \sum_{s' \in \mathcal S} p(s'\mid s, a, \mu_t) \sum_{a' \in \mathcal A} (\Phi_{\eta}(\mu))_{t+1}(a' \mid s') Q^{\Phi_{\eta}(\mu)}(\mu, t+1,s',a') \right. \\
    &\quad \left. - r(s,a,\mu'_t) - \sum_{s' \in \mathcal S} p(s'\mid s, a, \mu'_t) \sum_{a' \in \mathcal A} (\Phi_{\eta}(\mu'))_{t+1}(a' \mid s') Q^{\Phi_{\eta}(\mu')}(\mu', t+1,s',a') \right| \\
    &\leq \left| r(s,a,\mu_t) - r(s,a,\mu'_t) \right| \\
    &\quad + \sum_{s' \in \mathcal S} \left| \left( p(s'\mid s, a, \mu_t) - p(s'\mid s, a, \mu'_t) \right) \sum_{a' \in \mathcal A} (\Phi_{\eta}(\mu))_{t+1}(a' \mid s') Q^{\Phi_{\eta}(\mu)}(\mu, t+1,s',a') \right| \\
    &\quad + \sum_{s' \in \mathcal S} \left| p(s'\mid s, a, \mu'_t) \sum_{a' \in \mathcal A} \left( (\Phi_{\eta}(\mu))_{t+1}(a' \mid s') Q^{\Phi_{\eta}(\mu)}(\mu, t+1,s',a') - (\Phi_{\eta}(\mu'))_{t+1}(a' \mid s') Q^{\Phi_{\eta}(\mu')}(\mu', t+1,s',a') \right) \right| \\
    &\leq \left| r(s,a,\mu_t) - r(s,a,\mu'_t) \right| \\
    &\quad + \sum_{s' \in \mathcal S} \left| \left( p(s'\mid s, a, \mu_t) - p(s'\mid s, a, \mu'_t) \right) \sum_{a' \in \mathcal A} (\Phi_{\eta}(\mu))_{t+1}(a' \mid s') Q^{\Phi_{\eta}(\mu)}(\mu, t+1,s',a') \right| \\
    &\quad + \max_{s' \in \mathcal S} \left| \sum_{a' \in \mathcal A_{\mathrm{opt}}^{s'}} \left( (\Phi_{\eta}(\mu))_{t+1}(a' \mid s') Q^{\Phi_{\eta}(\mu)}(\mu, t+1,s',a') - (\Phi_{\eta}(\mu'))_{t+1}(a' \mid s') Q^{\Phi_{\eta}(\mu')}(\mu', t+1,s',a') \right) \right| \\
    &\quad + \max_{s' \in \mathcal S} \left| \sum_{a' \in \mathcal A \setminus \mathcal A_{\mathrm{opt}}^{s'}} \left( (\Phi_{\eta}(\mu))_{t+1}(a' \mid s') Q^{\Phi_{\eta}(\mu)}(\mu, t+1,s',a') - (\Phi_{\eta}(\mu'))_{t+1}(a' \mid s') Q^{\Phi_{\eta}(\mu')}(\mu', t+1,s',a') \right) \right|
\end{align*}
where we define $\mathcal A_{\mathrm{opt}}^{s'} \subseteq \mathcal A$ for any $s' \in \mathcal S$ to include all optimal actions $a_{\mathrm{opt}} \in \mathcal A_{\mathrm{opt}}^{s'}$ such that
\begin{align*}
    Q^*(\mu, t, s', a_{\mathrm{opt}}) = \max_{a \in \mathcal A} Q^*(\mu, t, s', a) \, .
\end{align*}
We bound each of the four terms separately.

For the first term, we choose $\delta_{t,s,a}^1 = \frac{\varepsilon_{t,s,a}}{4C_r}$ by Lipschitz continuity such that 
\begin{align*}
    \left| r(s,a,\mu_t) - r(s,a,\mu'_t) \right| < \frac{\varepsilon_{t,s,a}}{4}
\end{align*} 
for all $\mu'$ with $d_{\mathcal M}(\mu, \mu') < \delta_{t,s,a}^1$.

For the second term, we choose $\delta_{t,s,a}^2 = \frac{1}{4|\mathcal S| M_Q C_p}$ such that for any $\mu' \in \mathcal M$ with $d_{\mathcal M}(\mu, \mu') < \delta_{t,s,a}^2$ we have
\begin{align*}
    &\sum_{s' \in \mathcal S} \left| \left( p(s'\mid s, a, \mu_t) - p(s'\mid s, a, \mu'_t) \right) \sum_{a' \in \mathcal A} (\Phi_{\eta}(\mu))_{t+1}(a' \mid s') Q^{\Phi_{\eta}(\mu)}(\mu, t+1,s',a') \right| \\
    &\leq |\mathcal S| C_p d_{\mathcal M}(\mu, \mu') M_Q 
    < \frac{\varepsilon_{t,s,a}}{4}
\end{align*}
where $M_Q$ denotes the uniform bound of $Q$ and $C_p$ is the Lipschitz constant of $\nu \mapsto p(s'\mid s, a, \nu_t)$.

For the third and fourth term, we first fix $s' \in \mathcal S$ and define the minimal action gap as
\begin{align*}
    \Delta Q_{\mathrm{min}}^{s', \mu} \equiv \min_{a_{\mathrm{opt}} \in \mathcal A_{\mathrm{opt}}^{s'}, a_{\mathrm{sub}} \in \mathcal A \setminus \mathcal A_{\mathrm{opt}}^{s'}} \left( Q^*(\mu, t, s', a_{\mathrm{opt}}) - Q^*(\mu, t, s', a_{\mathrm{sub}}) \right) \, .
\end{align*}

This is well defined if there are suboptimal actions, since there is always at least one optimal action. If all actions are optimal, we can skip bounding the probability of taking suboptimal actions and the result will still hold. Henceforth, we assume that there exists a suboptimal action.

By Lipschitz continuity of $\mu \mapsto Q^*(\mu, t, s, a)$ from Lemma~\ref{lem:qlip} implying uniform continuity, there exists some $\delta_{t,s,a}^{3,s'} > 0$ such that
\begin{align*}
    \left| Q^*(\mu', t, s', a) - Q^*(\mu, t, s', a) \right| < \frac{\Delta Q_{\mathrm{min}}^{s', \mu}}{4}
\end{align*}
for all $\mu' \in \mathcal M, a \in \mathcal A$ where $d_{\mathcal M}(\mu, \mu') < \delta_{t,s,a}^{3,s'}$, and thus
\begin{align*}
    \Delta Q_{\mathrm{min}}^{s', \mu'} = \min_{a_{\mathrm{opt}} \in \mathcal A_{\mathrm{opt}}^{s'}, a_{\mathrm{sub}} \in \mathcal A \setminus \mathcal A_{\mathrm{opt}}^{s'}} \left( Q^*(\mu', t, s', a_{\mathrm{opt}}) - Q^*(\mu', t, s', a_{\mathrm{sub}}) \right) > \frac{\Delta Q_{\mathrm{min}}^{s', \mu}}{2} \, .
\end{align*}
Under this condition, we can now show that the probability of any suboptimal action can be controlled. Define $R_q^{\mathrm{min}} \equiv \min_{t \in \mathcal T, s \in \mathcal S, a \in \mathcal A, a' \in \mathcal A} \frac{q_t(a' \mid s)}{q_t(a \mid s)} > 0$ and $R_q^{\mathrm{max}} \equiv \max_{t \in \mathcal T, s \in \mathcal S, a \in \mathcal A, a' \in \mathcal A} \frac{q_t(a' \mid s)}{q_t(a \mid s)} > 0$. Let $a_{\mathrm{sub}} \in \mathcal A \setminus \mathcal A_{\mathrm{opt}}^{s'}$, then we either have
\begin{align*}
    &\left| (\Phi_{\eta}(\mu))_{t+1}(a_{\mathrm{sub}} \mid s') - (\Phi_{\eta}(\mu'))_{t+1}(a_{\mathrm{sub}} \mid s') \right| \\
    &= \left| \frac{ 1 }{ 1 + \sum_{a' \neq a_{\mathrm{sub}}} \frac{q_t(a' \mid s')}{q_t(a_{\mathrm{sub}} \mid s')} \exp \left( \frac{Q^*(\mu, t, s', a') - Q^*(\mu, t, s', a_{\mathrm{sub}})}{\eta} \right) } \right. \\
    &\quad \left. - \frac{ 1 }{ 1 + \sum_{a' \neq a_{\mathrm{sub}}} \frac{q_t(a' \mid s')}{q_t(a_{\mathrm{sub}} \mid s')} \exp \left( \frac{Q^*(\mu', t, s', a') - Q^*(\mu', t, s', a_{\mathrm{sub}})}{\eta} \right) } \right| \\
    &\leq \frac{ 1 }{ 1 + \max_{a' \neq a_{\mathrm{sub}}} R_q^{\mathrm{min}} \exp \left( \frac{Q^*(\mu, t, s', a') - Q^*(\mu, t, s', a_{\mathrm{sub}})}{\eta} \right) } \\
    &\quad + \frac{ 1 }{ 1 + \max_{a' \neq a_{\mathrm{sub}}} R_q^{\mathrm{min}} \exp \left( \frac{Q^*(\mu', t, s', a') - Q^*(\mu', t, s', a_{\mathrm{sub}})}{\eta} \right) } \\
    &< \frac{ 1 }{ 1 + R_q^{\mathrm{min}} \exp \left( \frac{\Delta Q_{\mathrm{min}}^{s', \mu}}{\eta} \right) } + \frac{ 1 }{ 1 + R_q^{\mathrm{min}} \exp \left( \frac{\Delta Q_{\mathrm{min}}^{s', \mu}}{2\eta} \right) } \\ 
    &\leq \frac{ 2 }{ 1 + R_q^{\mathrm{min}} \exp \left( \frac{\Delta Q_{\mathrm{min}}^{s', \mu}}{2\eta} \right) }
    < \frac{\varepsilon_{t,s,a}}{8M_Q |\mathcal A|} 
\end{align*}
if $\varepsilon_{t,s,a} > 16M_Q |\mathcal A|$ trivially, or otherwise if $\eta < \eta_{\mathrm{min}}^{s'}$ with
\begin{align*}
    \eta_{\mathrm{min}}^{s'} \equiv \frac{\Delta Q_{\mathrm{min}}^{s', \mu}}{2\log \left( \frac{16M_Q |\mathcal A|}{\varepsilon_{t,s,a}R_q^{\mathrm{min}}} - \frac{1}{R_q^{\mathrm{min}}} \right)} \, ,
\end{align*}
in which case we arbitrarily define $\delta_{t,s,a}^{4,s'} = 1$, or if neither apply, then $\eta \geq \eta_{\mathrm{min}}^{s'}$ and thus
\begin{align*} 
    &\left| (\Phi_{\eta}(\mu))_{t+1}(a_{\mathrm{sub}} \mid s') - (\Phi_{\eta}(\mu'))_{t+1}(a_{\mathrm{sub}} \mid s') \right| \\
    &= \left| \frac{ 1 }{ 1 + \sum_{a' \neq a_{\mathrm{sub}}} \frac{q_t(a' \mid s')}{q_t(a_{\mathrm{sub}} \mid s')} \exp \left( \frac{Q^*(\mu, t, s', a') - Q^*(\mu, t, s', a_{\mathrm{sub}})}{\eta} \right) } \right. \\
    &\quad \left. - \frac{ 1 }{ 1 + \sum_{a' \neq a_{\mathrm{sub}}} \frac{q_t(a' \mid s')}{q_t(a_{\mathrm{sub}} \mid s')} \exp \left( \frac{Q^*(\mu', t, s', a') - Q^*(\mu', t, s', a_{\mathrm{sub}})}{\eta} \right) }\right| \\
    &= \left| \frac{ \sum_{a' \neq a_{\mathrm{sub}}} \frac{q_t(a' \mid s)}{q_t(a_{\mathrm{sub}} \mid s')} \left( \exp \left( \frac{Q^*(\mu', t, s', a') - Q^*(\mu', t, s', a_{\mathrm{sub}})}{\eta} \right) - \exp \left( \frac{Q^*(\mu, t, s', a') - Q^*(\mu, t, s', a_{\mathrm{sub}})}{\eta} \right) \right) }{ \left( 1 + \cdots \right) \cdot \left( 1 + \cdots \right) } \right| \\
    &\leq R_q^{\mathrm{max}} \sum_{a' \neq a_{\mathrm{sub}}} \left| \exp \left( \frac{Q^*(\mu', t, s', a') - Q^*(\mu', t, s', a_{\mathrm{sub}})}{\eta} \right) - \exp \left( \frac{Q^*(\mu, t, s', a') - Q^*(\mu, t, s', a_{\mathrm{sub}})}{\eta} \right) \right| \\
    &\leq R_q^{\mathrm{max}} \sum_{a' \neq a_{\mathrm{sub}}} \left| \frac{1}{\eta} \exp \left( \frac{\xi_{a'}}{\eta} \right) \right| | (Q^*(\mu', t, s', a') - Q^*(\mu', t, s', a_{\mathrm{sub}})) - (Q^*(\mu, t, s', a') - Q^*(\mu, t, s', a_{\mathrm{sub}})) | \\
    &\leq R_q^{\mathrm{max}} |\mathcal A| \cdot \frac{1}{\eta_{\mathrm{min}}^{s'}} \exp \left( \frac{2 M_Q}{\eta_{\mathrm{min}}^{s'}} \right) \left( | Q^*(\mu', t, s', a') - Q^*(\mu, t, s', a') | + | Q^*(\mu, t, s', a_{\mathrm{sub}}) - Q^*(\mu', t, s', a_{\mathrm{sub}}) | \right) \\
    &\leq R_q^{\mathrm{max}} |\mathcal A| \cdot \frac{1}{\eta_{\mathrm{min}}^{s'}} \exp \left( \frac{2 M_Q}{\eta_{\mathrm{min}}^{s'}} \right) \cdot 2 K_Q d_{\mathcal M}(\mu, \mu') 
    < \frac{\varepsilon_{t,s,a}}{8M_Q |\mathcal A|} 
\end{align*}
by the mean value theorem with some $\xi_{a'} \in [-2M_Q, 2M_Q]$ for all $a' \in \mathcal A$, where we abbreviated the denominator $\left( 1 + \cdots \right) \cdot \left( 1 + \cdots \right) \geq 1$, as long as we choose
\begin{align*}
    \delta_{t,s,a}^{4,s'} = \frac{\varepsilon_{t,s,a} \eta_{\mathrm{min}}^{s'}}{8M_Q |\mathcal A|^2 R_q^{\mathrm{max}} \cdot \exp \left( \frac{2 M_Q}{\eta_{\mathrm{min}}^{s'}} \right) \cdot 2 K_Q} 
\end{align*}
and $d_{\mathcal M}(\mu, \mu') < \delta_{t,s,a}^{4,s'}$, where $K_Q$ is the Lipschitz constant of $\mu \mapsto Q^*(\mu, t, s, a)$ given by Lemma~\ref{lem:qlip}.

Since $s' \in \mathcal S$ was arbitrary, we now define $\delta_{t,s,a}^3 \equiv \min_{s' \in \mathcal S} \delta_{t,s,a}^{3,s'}$, $\delta_{t,s,a}^4 \equiv \min_{s' \in \mathcal S} \delta_{t,s,a}^{4,s'}$ and let $d_{\mathcal M}(\mu, \mu') < \min(\delta_{t,s,a}^3, \delta_{t,s,a}^4)$. Under these assumptions, for the third term we have approximate optimality for all optimal actions in $\mathcal A_{\mathrm{opt}}^{s'}$, since by induction assumption we can choose $\delta_{t+1, s', a'}$ for all $s' \in \mathcal S, a' \in \mathcal A$ such that for all $\mu' \in \mathcal M$ with $d_{\mathcal M}(\mu, \mu') < \delta_{t+1, s', a'}$ it holds that
\begin{align*}
    \left| Q^{\Phi_{\eta}(\mu)}(\mu, t+1,s',a') - Q^{\Phi_{\eta}(\mu')}(\mu', t+1,s',a') \right| < \frac{\varepsilon_{t,s,a}}{16 |\mathcal A| + 8} \, .
\end{align*}
and therefore for all $\mu' \in \mathcal M$, as long as $d_{\mathcal M}(\mu, \mu') < \min_{s' \in \mathcal S, a' \in \mathcal A} \delta_{t+1, s', a'}$, we have
\begin{align*}
    &\max_{s' \in \mathcal S} \left| \sum_{a' \in \mathcal A_{\mathrm{opt}}^{s'}} (\Phi_{\eta}(\mu))_{t+1}(a' \mid s') Q^{\Phi_{\eta}(\mu)}(\mu, t+1,s',a') - \sum_{a' \in \mathcal A_{\mathrm{opt}}^{s'}} (\Phi_{\eta}(\mu'))_{t+1}(a' \mid s') Q^{\Phi_{\eta}(\mu')}(\mu', t+1,s',a') \right| \\
    &\leq \max_{s' \in \mathcal S} \left| \sum_{a' \in \mathcal A_{\mathrm{opt}}^{s'}} (\Phi_{\eta}(\mu))_{t+1}(a' \mid s') Q^{\Phi_{\eta}(\mu)}(\mu, t+1,s',a') - \sum_{a' \in \mathcal A_{\mathrm{opt}}^{s'}} (\Phi_{\eta}(\mu))_{t+1}(a' \mid s') Q^{\Phi_{\eta}(\mu')}(\mu', t+1,s',a') \right| \\
    &\quad + \max_{s' \in \mathcal S} \left| \sum_{a' \in \mathcal A_{\mathrm{opt}}^{s'}} (\Phi_{\eta}(\mu))_{t+1}(a' \mid s') Q^{\Phi_{\eta}(\mu')}(\mu', t+1,s',a') - \sum_{a' \in \mathcal A_{\mathrm{opt}}^{s'}} (\Phi_{\eta}(\mu'))_{t+1}(a' \mid s') Q^{\Phi_{\eta}(\mu')}(\mu', t+1,s',a') \right| \\
    &\leq \max_{s' \in \mathcal S} \max_{a' \in \mathcal A} \left| Q^{\Phi_{\eta}(\mu)}(\mu, t+1,s',a') - Q^{\Phi_{\eta}(\mu')}(\mu', t+1,s',a') \right| \\
    &\quad + \max_{s' \in \mathcal S} \left| \sum_{a' \in \mathcal A_{\mathrm{opt}}^{s'}} \left( (\Phi_{\eta}(\mu))_{t+1}(a' \mid s') - (\Phi_{\eta}(\mu'))_{t+1}(a' \mid s') \right) \left( Q^{\Phi_{\eta}(\mu')}(\mu', t+1,s',a') - Q^{\Phi_{\eta}(\mu)}(\mu, t+1,s',a') \right) \right| \\
    &\quad + \max_{s' \in \mathcal S} \left| \sum_{a' \in \mathcal A_{\mathrm{opt}}^{s'}} \left( (\Phi_{\eta}(\mu))_{t+1}(a' \mid s') - (\Phi_{\eta}(\mu'))_{t+1}(a' \mid s') \right) Q^{\Phi_{\eta}(\mu)}(\mu, t+1,s',a') \right| \\
    &\leq \max_{s' \in \mathcal S} \max_{a' \in \mathcal A} \left| Q^{\Phi_{\eta}(\mu)}(\mu, t+1,s',a') - Q^{\Phi_{\eta}(\mu')}(\mu', t+1,s',a') \right| \\
    &\quad + \max_{s' \in \mathcal S} \max_{a' \in \mathcal A} 2 |\mathcal A| \left| Q^{\Phi_{\eta}(\mu')}(\mu', t+1,s',a') - Q^{\Phi_{\eta}(\mu)}(\mu, t+1,s',a') \right| \\
    &\quad + \max_{s' \in \mathcal S} \max_{a'' \in \mathcal A} \left| Q^{\Phi_{\eta}(\mu)}(\mu, t+1,s',a'') \right| \cdot \left| \sum_{a' \in \mathcal A \setminus \mathcal A_{\mathrm{opt}}^{s'}} \left( (\Phi_{\eta}(\mu'))_{t+1}(a' \mid s') - (\Phi_{\eta}(\mu))_{t+1}(a' \mid s') \right) \right| \\
    &< \left( 1 + 2 |\mathcal A| \right) \cdot \frac{\varepsilon_{t,s,a}}{16 |\mathcal A| + 8} + M_Q |\mathcal A| \cdot \frac{\varepsilon_{t,s,a}}{8M_Q |\mathcal A|} < \frac{\varepsilon_{t,s,a}}{4}
\end{align*}
where we use that for any $a' \in \mathcal A_{\mathrm{opt}}^{s'}$ we have
\begin{align*}
    Q^{\Phi_{\eta}(\mu)}(\mu, t+1,s',a') = \max_{a'' \in \mathcal A} Q^{\Phi_{\eta}(\mu)}(\mu, t+1,s',a'') \, .
\end{align*}

Analogously, for the fourth term we have 
\begin{align*}
    &\max_{s' \in \mathcal S} \left| \sum_{a' \in \mathcal A \setminus \mathcal A_{\mathrm{opt}}^{s'}} ( (\Phi_{\eta}(\mu))_{t+1}(a' \mid s') Q^{\Phi_{\eta}(\mu)}(\mu, t+1,s',a') - (\Phi_{\eta}(\mu'))_{t+1}(a' \mid s') Q^{\Phi_{\eta}(\mu')}(\mu', t+1,s',a') ) \right| \\
    &\leq \max_{s' \in \mathcal S} \sum_{a' \in \mathcal A \setminus \mathcal A_{\mathrm{opt}}^{s'}} \left| (\Phi_{\eta}(\mu))_{t+1}(a' \mid s') Q^{\Phi_{\eta}(\mu)}(\mu, t+1,s',a') - (\Phi_{\eta}(\mu))_{t+1}(a' \mid s') Q^{\Phi_{\eta}(\mu')}(\mu', t+1,s',a') \right| \\
    &\quad + \max_{s' \in \mathcal S} \sum_{a' \in \mathcal A \setminus \mathcal A_{\mathrm{opt}}^{s'}} \left| (\Phi_{\eta}(\mu))_{t+1}(a' \mid s') Q^{\Phi_{\eta}(\mu')}(\mu', t+1,s',a') - (\Phi_{\eta}(\mu'))_{t+1}(a' \mid s') Q^{\Phi_{\eta}(\mu')}(\mu', t+1,s',a') \right| 
    \\
    &\leq \max_{s' \in \mathcal S} \max_{a' \in \mathcal A} \left| Q^{\Phi_{\eta}(\mu)}(\mu, t+1,s',a') - Q^{\Phi_{\eta}(\mu')}(\mu', t+1,s',a') \right| \\
    &\quad + \max_{s' \in \mathcal S} M_Q \sum_{a' \in \mathcal A \setminus \mathcal A_{\mathrm{opt}}^{s'}} \left| (\Phi_{\eta}(\mu))_{t+1}(a' \mid s') - (\Phi_{\eta}(\mu'))_{t+1}(a' \mid s') \right|   
    \\
    &< \frac{\varepsilon_{t,s,a}}{8} + M_Q |\mathcal A| \cdot \frac{\varepsilon_{t,s,a}}{8M_Q |\mathcal A|}
    = \frac{\varepsilon_{t,s,a}}{4}
\end{align*}
under the previous conditions, since as long as we have $d_{\mathcal M}(\mu, \mu') < \delta_{t+1, s', a'}$ for all $s' \in \mathcal S, a' \in \mathcal A$ from before, we have
\begin{align*}
    \left| Q^{\Phi_{\eta}(\mu)}(\mu, t+1,s',a') - Q^{\Phi_{\eta}(\mu')}(\mu', t+1,s',a') \right| < \frac{\varepsilon_{t,s,a}}{16 |\mathcal A| + 8} < \frac{\varepsilon_{t,s,a}}{8} \, .
\end{align*}

Finally, by choosing $\delta_{t,s,a}$ such that all conditions are fulfilled, i.e.
\begin{align*}
    \delta_{t,s,a} \equiv \min \left( \delta_{t,s,a}^1, \delta_{t,s,a}^2, \delta_{t,s,a}^3, \delta_{t,s,a}^4, \min_{s' \in \mathcal S, a' \in \mathcal A} \delta_{t+1, s', a'} \right) > 0 \, ,
\end{align*}
the induction hypothesis is fulfilled, since then for any $\mu'$ with $d_{\mathcal M}(\mu, \mu') < \delta_{t,s,a}$ we have
\begin{align*}
    \left| Q^{\Phi_{\eta}(\mu)}(\mu, t,s,a) - Q^{\Phi_{\eta}(\mu')}(\mu', t,s,a) \right| < \varepsilon_{t,s,a} \, .
\end{align*}

Since $\eta > 0$ is arbitrary, the desired result follows immediately, as we can set $\varepsilon_{t,s,a} = \varepsilon$ for each $t \in \mathcal T, s \in \mathcal S, a \in \mathcal A$ and obtain $\delta \equiv \max_{t \in \mathcal T, s \in \mathcal S, a \in \mathcal A} \delta_{t,s,a}$, fulfilling the required equicontinuity property at $\mu$.
\end{subproof}

From equicontinuity, we get the desired uniform convergence via compactness.

\begin{lemma} \label{lem:uniform}
If $(f_n)_{n \in \mathbb N}$ with $f_n: \mathcal M \to \mathbb R$ is an equicontinuous sequence of functions and for all $\mu \in \mathcal M$ we have $f_n(\mu) \to f(\mu)$ pointwise, then $f_n(\mu) \to f(\mu)$ uniformly.
\end{lemma}
\begin{subproof}
Let $\varepsilon > 0$ arbitrary, then there exists by equicontinuity for any point $\mu \in \mathcal M$ a $\delta(\mu)$ such that for all $\mu' \in \mathcal M$ with $d_{\mathcal M}(\mu, \mu') < \delta(\mu)$ we have for all $n \in \mathbb N$
\begin{align*}
    \left| f_n(\mu) - f_n(\mu') \right| < \frac{\varepsilon}{3}
\end{align*} 
which via pointwise convergence implies 
\begin{align*}
    \left| f(\mu) - f(\mu') \right| \leq \frac{\varepsilon}{3} \, .
\end{align*} 
Since $\mathcal M$ is compact, it is separable, i.e. there exists a countable dense subset $(\mu_j)_{j \in \mathbb N}$ of $\mathcal M$. Let $\delta(\mu)$ be as defined above and cover $\mathcal M$ by the open balls $(B_{\delta(\mu_j)}(\mu_j))_{j \in \mathbb N}$. By the compactness of $\mathcal M$, finitely many of these balls $B_{\delta(\mu_{n_1})}(\mu_{n_1}), \ldots, B_{\delta(\mu_{n_k})}(\mu_{n_k})$ cover $\mathcal M$. By pointwise convergence, for any $i = 1, \ldots, k$ we can find an integer $n_i$ such that for all $n > n_i$ we have 
\begin{align*}
    \left| f_n(\mu_{n_i}) - f(\mu_{n_i}) \right| < \frac{\varepsilon}{3} \, .
\end{align*}
Taken together, we find that for $n > \max_{i = 1, \ldots, k} n_i$ and arbitrary $\mu \in \mathcal M$, we have
\begin{align*}
    \left| f_n(\mu) - f(\mu) \right| < \left| f_n(\mu) - f_n(\mu_{n_i}) \right| + \left| f_n(\mu_{n_i}) - f(\mu_{n_i}) \right| + \left| f(\mu_{n_i}) - f(\mu) \right| < \frac{\varepsilon}{3} + \frac{\varepsilon}{3} + \frac{\varepsilon}{3} = \varepsilon
\end{align*} 
for some center point $\mu_{n_i}$ of a ball containing $\mu$ from the finite cover.
\end{subproof}

Therefore, a sequence of Boltzmann MFE with vanishing $\eta$ is approximately optimal in the MFG.

\begin{lemma} \label{lem:boltzmann-expl}
For any sequence $(\pi^*_{n}, \mu^*_{n})_{n \in \mathbb N}$ of ${\eta_n}$-Boltzmann MFE with $\eta_n \to 0^+$ and for any $\varepsilon > 0$ there exists integer $N \in \mathbb N$ such that for all integers $n > N$ we have
\begin{align*}
    J^{\mu^*_{n}}(\pi^*_{n}) \geq \max_\pi J^{\mu^*_{n}}(\pi) - \varepsilon \ .
\end{align*}
\end{lemma}
\begin{subproof}
By Lemma~\ref{lem:boltzmann-equi}, $\mathcal F \equiv (\mu \mapsto Q^{\Phi_{\eta}(\mu)}(\mu, t, s, a))_{\eta > 0, t \in \mathcal T, s \in \mathcal S, a \in \mathcal A}$ is equicontinuous. Therefore, any sequence $(\mu \mapsto Q^{\Phi_{\eta_n}(\mu)}(\mu, t, s, a))_{n \in \mathbb N}$ with $\eta_n \to 0^+$ is also equicontinuous for any $t \in \mathcal T, s \in \mathcal S, a \in \mathcal A$.

Furthermore, by Lemma~\ref{lem:boltzmann-pointwise}, the sequence $(\mu \mapsto Q^{\Phi_{\eta_n}(\mu)}(\mu, t, s, a))_{n \in \mathbb N}$ converges pointwise to $\mu \to Q^*(\mu, t, s, a)$ for any $t \in \mathcal T, s \in \mathcal S, a \in \mathcal A$.

By Lemma~\ref{lem:uniform}, we thus have $\left| Q^{\Phi_{\eta_n}(\mu)}(\mu, t, s, a) - Q^*(\mu, t, s, a) \right| \to 0$ uniformly. Therefore, for any $\varepsilon > 0$, there exists an integer $N$ by uniform convergence such that for all integers $n > N$ we have
\begin{align*}
    Q^{\pi^*_{n}}(\mu^*_{n}, t,s,a) \geq Q^*(\mu^*_{n}, t,s,a) - \varepsilon = \max_{\pi \in \Pi} Q^{\pi}(\mu^*_{n}, t,s,a) - \varepsilon \, ,
\end{align*}
and since by Lemma~\ref{lem:policyoptimality} we have
\begin{align*}
    J^{\mu^*_{n}}(\pi^*_{n}) = \sum_{s \in \mathcal S} \mu_0(s) \cdot \sum_{a \in \mathcal A} Q^{\pi^*_{n}}(\mu^*_{n}, t,s,a) \geq \sum_{s \in \mathcal S} \mu_0(s) \cdot \max_{\pi \in \Pi} \sum_{a \in \mathcal A} Q^{\pi}(\mu^*_{n}, t,s,a) - \varepsilon = \max_{\pi \in \Pi} J^{\mu^*_{n}}(\pi) - \varepsilon \, ,
\end{align*}
the desired result follows immediately.
\end{subproof}

Finally, we show approximate optimality in the actual $N$-agent game as long as a pair $(\pi^*, \mu^*) \in \Pi \times \mathcal M$ with $\mu^* = \Psi(\pi^*)$ has vanishing exploitability in the MFG. By Lemma~\ref{lem:boltzmann-expl}, for any sequence $(\pi^*_{n}, \mu^*_{n})_{n \in \mathbb N}$ of ${\eta_n}$-Boltzmann MFE with $\eta_n \to 0^+$ and for any $\varepsilon > 0$ there exists an integer $n' \in \mathbb N$ such that for all integers $n > n'$ we have
\begin{align*}
    J^{\mu^*_{n}}(\pi^*_{n}) \geq \max_\pi J^{\mu^*_{n}}(\pi) - \varepsilon \ .
\end{align*}
Let $\varepsilon' > 0$ be arbitrary and choose a sequence of optimal policies $\{ \pi^N \}_{N \in \mathbb N}$ such that for all $N \in \mathbb N$ we have 
\begin{align*}
    \pi^N \in \argmax_{\pi \in \Pi} J_1^N(\pi, {\pi^*_{n}}, \ldots, {\pi^*_{n}}) \, .
\end{align*} 
By Lemma~\ref{lem:obj-conv} there exists $N' \in \mathbb N$ such that for all $N > N'$ and all $n > n'$, we have
\begin{align*}
    \max_{\pi \in \Pi} J_1^N(\pi, {\pi^*_{n}}, \ldots, {\pi^*_{n}}) - \varepsilon - \varepsilon' &\leq \max_{\pi \in \Pi} J^{\mu^*_{n}}(\pi) - \varepsilon - \frac{\varepsilon'}{2} \\
    &\leq J^{\mu^*_{n}}({\pi^*_{n}}) - \frac{\varepsilon'}{2} \\
    &\leq J_1^N({\pi^*_{n}}, {\pi^*_{n}}, \ldots, {\pi^*_{n}})
\end{align*} 
which is the desired approximate Nash equilibrium property since $\varepsilon, \varepsilon'$ are arbitrary. This applies by symmetry to all agents. 

For RelEnt MFE, the same can be done by first showing the uniform convergence of the soft action-value function to the usual action-value function. For this, note that the smooth maximum Bellman recursion converges to the hard maximum Bellman recursion for any fixed $\mu$.

\begin{lemma} \label{lem:max}
For any $f: \mathcal A \to \mathbb R$ and any $g: \mathcal A \to \mathbb R$ with $g(a) > 0$ for all $a \in \mathcal A$, we have 
\begin{align*}
    \lim_{\eta \to 0^+} \eta \log \sum_{a \in \mathcal A} g(a) \exp \frac{f(a)}{\eta} = \max_{a \in \mathcal A} f(a) \, .
\end{align*}
\end{lemma}
\begin{subproof}
Let $\delta = \frac{1}{\eta} \to +\infty$. Then, by L'Hospital's rule we have
\begin{align*}
    \lim_{\delta \to +\infty} \frac{\log \sum_{a \in \mathcal A} g(a) \exp \left( \delta f(a) \right) }{\delta} &= \lim_{\delta \to +\infty} \frac{\sum_{a \in \mathcal A} g(a) \exp \left( \delta f(a) \right) f(a) }{\sum_{a \in \mathcal A} g(a) \exp \left( \delta f(a) \right) } \\
    &= \lim_{\delta \to +\infty} \frac{\sum_{a \in \mathcal A} g(a) \exp \left( \delta (f(a) - \max_{a \in \mathcal A} f(a)) \right) f(a) }{\sum_{a \in \mathcal A} g(a) \exp \left( \delta (f(a) - \max_{a \in \mathcal A} f(a)) \right) } \\
    &= \frac{ |\mathcal A_{\mathrm{max}}| \max_{a \in \mathcal A} f(a) }{|\mathcal A_{\mathrm{max}}|} = \max_{a \in \mathcal A} f(a)
\end{align*}
where $|\mathcal A_{\mathrm{max}}|$ is the number of elements in $\mathcal A$ that maximize $f$.
\end{subproof}

Using this result, we can show pointwise convergence of the soft action-value function to the action-value function.

\begin{lemma} \label{lem:tildeqpointwise}
Any sequence of functions $(\mu \mapsto \tilde Q_{\eta_n}(\mu, t, s, a))_{n \in \mathbb N}$ with $\eta_n \to 0^+$ converges pointwise to $\mu \mapsto Q^*(\mu, t, s, a)$ for all $t \in \mathcal T, s \in \mathcal S, a \in \mathcal A$.
\end{lemma}
\begin{subproof}
Fix $\mu \in \mathcal M$. We show by induction that for any $\varepsilon > 0$, there exists $\eta_t > 0$ such that for all $\eta < \eta_t$ we have $\left| \tilde Q_\eta(\mu, t, s, a) - Q^*(\mu, t, s, a) \right| < \varepsilon$ for all $t \in \mathcal T, s \in \mathcal S, a \in \mathcal A$. This holds for $t = T-1$ and arbitrary $s \in \mathcal S, a \in \mathcal A$ by Lemma~\ref{lem:max}, since $r(s, a, \mu_{T-1})$ is independent of $\eta$. Assume this holds for $t+1$ and consider $t$. Then, by the induction assumption we can choose $\eta_{t+1} > 0$ such that for $\eta < \eta_{t+1}$, as $\eta \to 0^+$ we have
\begin{align*}
    \tilde Q_\eta(\mu, t, s, a) &= r(s, a, \mu_t) + \sum_{s' \in \mathcal S} p(s'\mid s, a, \mu_t) \eta \log \sum_{a' \in \mathcal A} q_{t+1}(a' \mid s') \exp \left( \frac{\tilde Q_\eta(\mu, t+1, s', a')}{\eta} \right) \\
    &\leq r(s, a, \mu_t) + \sum_{s' \in \mathcal S} p(s'\mid s, a, \mu_t) \eta \log \sum_{a' \in \mathcal A} q_{t+1}(a' \mid s') \exp \left( \frac{Q^*(\mu, t+1, s', a') + \frac{\varepsilon}{2}}{\eta} \right) \\
    &\to r(s, a, \mu_t) + \sum_{s' \in \mathcal S} p(s'\mid s, a, \mu_t) \max_{a' \in \mathcal A} Q^*(\mu, t+1, s', a') + \frac{\varepsilon}{2}
\end{align*} 
by Lemma~\ref{lem:max} and monotonicity of $\log$ and $\exp$. Analogously,
\begin{align*}
    \tilde Q_\eta(\mu, t, s, a) 
    &\geq r(s, a, \mu_t) + \sum_{s' \in \mathcal S} p(s'\mid s, a, \mu_t) \eta \log \sum_{a' \in \mathcal A} q_{t+1}(a' \mid s') \exp \left( \frac{Q^*(\mu, t+1, s', a') - \frac{\varepsilon}{2}}{\eta} \right) \\
    &\to r(s, a, \mu_t) + \sum_{s' \in \mathcal S} p(s'\mid s, a, \mu_t) \max_{a' \in \mathcal A} Q^*(\mu, t+1, s', a') - \frac{\varepsilon}{2} \, .
\end{align*} 
Therefore, we can choose $\eta_t < \eta_{t+1}$ such that for all $\eta < \eta_{t}$ we have
\begin{align*}
    &\left| \tilde Q_\eta(\mu, t, s, a) - Q^*(\mu, t, s, a) \right| = \left| \tilde Q_\eta(\mu, t, s, a) - \left( r(s, a, \mu_t) + \sum_{s' \in \mathcal S} p(s'\mid s, a, \mu_t) \max_{a' \in \mathcal A} Q^*(\mu, t+1, s', a') \right) \right| < \varepsilon 
\end{align*} 
which is the desired result.
\end{subproof}

We can now show that the soft action-value function converges uniformly to the action-value function as $\eta \to 0^+$.

\begin{lemma} \label{lem:tildequniform}
Any sequence of functions $(\mu \mapsto \tilde Q_{\eta_n}(\mu, t, s, a))_{n \in \mathbb N}$ with $\eta_n \to 0^+$ converges uniformly to $\mu \mapsto Q^*(\mu, t, s, a)$ for all $t \in \mathcal T, s \in \mathcal S, a \in \mathcal A$.
\end{lemma}
\begin{subproof}
First, we show that $\tilde Q_\eta(\mu, t, s, a)$ is monotonically decreasing in $\eta$ for $\eta > 0$, i.e. $\frac{\partial}{\partial \eta} \tilde Q_\eta(\mu, t, s, a) \leq 0$ for all $t \in \mathcal T, s \in \mathcal S, a \in \mathcal A$.  This is the case for $t = T-1$ and arbitrary $s \in \mathcal S, a \in \mathcal A$, since $\tilde Q_\eta(\mu, T-1, s, a)$ is constant. Assume this holds for $t+1$, then for $t$ and arbitrary $s \in \mathcal S, a \in \mathcal A$ we have
\begin{align*}
    &\frac{\partial}{\partial \eta} \tilde Q_\eta(\mu, t, s, a) 
    = \sum_{s' \in \mathcal S} p(s'\mid s, a, \mu_t) \log \sum_{a' \in \mathcal A} q_{t+1}(a' \mid s') \exp \left( \frac{\tilde Q_\eta(\mu, t+1, s', a')}{\eta} \right) \\
    &\quad + \sum_{s' \in \mathcal S} p(s'\mid s, a, \mu_t) \eta \frac{ \sum_{a' \in \mathcal A} q_{t+1}(a' \mid s') \exp \left( \frac{\tilde Q_\eta(\mu, t+1, s', a')}{\eta} \right) \left( - \frac{\tilde Q_\eta(\mu, t+1, s', a')}{\eta^2} + \frac{1}{\eta} \frac{\partial}{\partial \eta} \tilde Q_\eta(\mu, t+1, s', a') \right) }{ \sum_{a' \in \mathcal A} q_{t+1}(a' \mid s') \exp \left( \frac{\tilde Q_\eta(\mu, t+1, s', a')}{\eta} \right) } \\
    &\leq \max_{s' \in \mathcal S} \left( \log \sum_{a' \in \mathcal A} q_{t+1}(a' \mid s') \exp \left( \frac{\tilde Q_\eta(\mu, t+1, s', a')}{\eta} \right) - \frac{ \sum_{a' \in \mathcal A} q_{t+1}(a' \mid s') \exp \left( \frac{\tilde Q_\eta(\mu, t+1, s', a')}{\eta} \right) \frac{\tilde Q_\eta(\mu, t+1, s', a')}{\eta} }{ \sum_{a' \in \mathcal A} q_{t+1}(a' \mid s') \exp \left( \frac{\tilde Q_\eta(\mu, t+1, s', a')}{\eta} \right) } \right) 
\end{align*}
by induction hypothesis. Let $\xi_{a'} \equiv \frac{\tilde Q_\eta(\mu, t+1, s', a')}{\eta} \in \mathbb R$ and $s' \in \mathcal S$ arbitrary, then by Jensen's inequality applied to the convex function $\phi(x) = x \log x$ we have 
\begin{align*}
    &\sum_{a' \in \mathcal A} q_{t+1}(a' \mid s') \phi( \exp \xi_{a'} ) \geq \phi \left( \sum_{a' \in \mathcal A} q_{t+1}(a' \mid s') \exp \xi_{a'}  \right) \\
    \iff &\sum_{a' \in \mathcal A} q_{t+1}(a' \mid s') \xi_{a'} \exp \xi_{a'} \geq \left( \sum_{a' \in \mathcal A} q_{t+1}(a' \mid s') \exp \xi_{a'} \right) \log \left( \sum_{a' \in \mathcal A} q_{t+1}(a' \mid s') \exp \xi_{a'} \right) \\
    \iff &\log \left( \sum_{a' \in \mathcal A} q_{t+1}(a' \mid s') \exp \xi_{a'} \right) - \frac{\sum_{a' \in \mathcal A} q_{t+1}(a' \mid s') \xi_{a'} \exp \xi_{a'}}{\left( \sum_{a' \in \mathcal A} q_{t+1}(a' \mid s') \exp \xi_{a'} \right)} \leq 0 \, ,
\end{align*}
such that $\tilde Q_\eta(\mu, t, s, a)$ is monotonically decreasing for all $t \in \mathcal T, s \in \mathcal S, a \in \mathcal A$ by induction. 

Furthermore, $\mathcal M$ is compact and both $\tilde Q_\eta$ and $Q$ are compositions, sums, products and finite maxima of continuous functions in $\mu$ and therefore continuous in $\mu$ by the standing assumptions. Since $(\mu \mapsto \tilde Q_{\eta_n}(\mu, t, s, a))_{n \in \mathbb N}$ with $\eta_n \to 0^+$ converges pointwise to $\mu \mapsto Q^*(\mu, t, s, a)$ for all $t \in \mathcal T, s \in \mathcal S, a \in \mathcal A$ by Lemma~\ref{lem:tildeqpointwise}, by Dini's theorem the convergence is uniform.
\end{subproof}

Now that $\tilde Q_\eta$ converges uniformly against $Q$, we can show that RelEnt MFE have vanishing exploitability by replicating the proof for Boltzmann MFE.

\begin{lemma} \label{lem:relent-pointwise}
Any sequence of functions $(\mu \mapsto Q^{\tilde \Phi_{\eta_n}(\mu)}(\mu, t, s, a))_{n \in \mathbb N}$ with $\eta_n \to 0^+$ converges pointwise to $\mu \mapsto Q^*(\mu, t, s, a)$ for all $t \in \mathcal T, s \in \mathcal S, a \in \mathcal A$.
\end{lemma}
\begin{subproof}
The proof is the same as in Lemma~\ref{lem:boltzmann-pointwise}. The only difference is that we additionally choose $n_2 \in \mathbb N$ in each induction step such that for all $n > n_2$ we have
\begin{align*}
    \left| \tilde Q_\eta(\mu, t, s, a) - Q^*(\mu, t, s, a) \right| \leq \frac{\Delta Q_{\mathrm{min}}^{s', \mu}}{4}
\end{align*}
for all $t \in \mathcal T, s \in \mathcal S, a \in \mathcal A$, which is possible, since by Lemma~\ref{lem:tildequniform}, $\tilde Q_\eta$ converges uniformly against $Q$. As long as we choose $n' \equiv \max(n_1, n_2, \max_{s' \in \mathcal S, a' \in \mathcal A} n_{s',a'})$, the rest of the proof will apply.
\end{subproof}

\begin{lemma} \label{lem:relent-equi}
Any sequence of functions $(\mu \mapsto Q^{\tilde \Phi_{\eta_n}(\mu)}(\mu, t, s, a))_{n \in \mathbb N}$ with $\eta_n \to 0^+$ fulfills equicontinuity for large enough $n$: For any $\varepsilon > 0$ and any $\mu \in \mathcal M$, we can choose a $\delta > 0$ and an integer $n' \in \mathbb N$ such that for all $\mu' \in \mathcal M$ with $d_{\mathcal M}(\mu, \mu') < \delta$ and for all $n > n'$ we have
\begin{align*}
    \left| Q^{\tilde \Phi_{\eta_n}(\mu)}(\mu, t, s, a) - Q^{\tilde \Phi_{\eta_n}(\mu')}(\mu', t, s, a) \right| < \varepsilon \, .
\end{align*}
\end{lemma}
\begin{subproof}
To obtain the desired property, we replicate the proof of Lemma~\ref{lem:boltzmann-equi} by setting $\mathcal F = (\mu \mapsto Q^{\tilde \Phi_{\eta_n}(\mu)}(\mu, t, s, a))_{n \in \mathbb N}$. Any bounds for $\tilde Q_\eta$ can be instantiated by the corresponding bound for $Q$ and then bounding the distance between both by uniform convergence. The only differences lie in bounding the terms
\begin{align*}
    \left| (\tilde \Phi_{\eta_n}(\mu)(a_{\mathrm{sub}} \mid s') - (\tilde \Phi_{\eta_n}(\mu')(a_{\mathrm{sub}} \mid s') \right|
\end{align*}
where the action-value function has been replaced with the soft action-value function. Since $\tilde Q_{\eta_n}$ uniformly converges to $Q$, we instantiate additional requirements $N_{t,s,a}^{s'}, \tilde N_{t,s,a}^{s'}$ to let $n > N_{t,s,a}^{s'}, n > \tilde N_{t,s,a}^{s'}$ large enough such that $\eta$ is sufficiently small enough.

The first difference is to obtain
\begin{align*}
    \left| \tilde Q_{\eta_n}(\mu', t, s, a) - \tilde Q_{\eta_n}(\mu, t, s, a) \right| < \frac{\Delta Q_{\mathrm{min}}^{s', \mu}}{4}
\end{align*}
for all $\mu' \in \mathcal M, t \in \mathcal T, s \in \mathcal S, a \in \mathcal A$ with $d_{\mathcal M}(\mu, \mu')$ sufficiently small. We choose $\hat \delta_{t,s,a}^3$ slightly stronger than in the original proof, such that if $d_{\mathcal M}(\mu, \mu') < \hat \delta_{t,s,a}^3$, we have
\begin{align*}
    \left| Q^*(\mu', t, s, a) - Q^*(\mu, t, s, a) \right| < \frac{\Delta Q_{\mathrm{\mathrm{min}}}^{s', \mu}}{12} \, .
\end{align*}
We must then additionally choose $N_{t,s,a}^{s'} \in \mathbb N$ for each induction step via uniform convergence from Lemma~\ref{lem:tildequniform} such that as long as $n > N_{t,s,a}^{s'}$, we have 
\begin{align*}
    \left| \tilde Q_{\eta_n}(\mu, t, s, a) - Q^*(\mu, t, s, a) \right| < \frac{\Delta Q_{\mathrm{min}}^{s', \mu}}{12} \, .
\end{align*}
This implies the required inequality
\begin{align*}
    \left| \tilde Q_{\eta_n}(\mu', t, s, a) - \tilde Q_{\eta_n}(\mu, t, s, a) \right|
    &\leq \left| \tilde Q_{\eta_n}(\mu', t, s, a) - Q^*(\mu', t, s, a) \right| + \left| Q^*(\mu', t, s, a) - Q^*(\mu, t, s, a) \right| \\
    &\quad + \left| Q^*(\mu, t, s, a) - \tilde Q_{\eta_n}(\mu, t, s, a) \right|
    < \frac{\Delta Q_{\mathrm{min}}^{s', \mu}}{4}
\end{align*}
and we can proceed as in the original proof. 

The second difference lies in choosing $\delta_{t,s,a}^{4,s'}$. Note that $\tilde Q_{\eta_n}$ is still bounded by $M_Q$, see Lemma~\ref{lem:qbound}. However, since $\tilde Q_{\eta_n}$ might no longer be Lipschitz with the same constant as $Q^*$, we choose an additional integer $\tilde N_{t,s,a}^{s'} \in \mathbb N$ for each induction step by Lemma~\ref{lem:tildequniform}, such that as long as $n > \tilde N_{t,s,a}^{s'}$, we have
\begin{align*}
    \left| \tilde Q_{\eta_n}(\mu, t, s, a) - Q^*(\mu, t, s, a) \right| \leq \Delta_Q^{s'} \equiv \frac{\frac{\varepsilon_{t,s,a}}{16M_Q |\mathcal A|}}{4 R_q^{\mathrm{max}} |\mathcal A| \cdot \frac{1}{\eta_{\mathrm{min}}^{s'}} \exp \left( \frac{2 M_Q}{\eta_{\mathrm{min}}^{s'}} \right)}
\end{align*}
for any $\mu' \in \mathcal M, t \in \mathcal T, s \in \mathcal S, a \in \mathcal A$. The required bound then follows immediately from
\begin{align*}
    &\left| (\Phi_{\eta_n}(\mu)(a_{\mathrm{sub}} \mid s') - (\Phi_{\eta_n}(\mu')(a_{\mathrm{sub}} \mid s') \right| \\
    &\leq R_q^{\mathrm{max}} \sum_{a' \neq a_{\mathrm{sub}}} \left| \exp \left( \frac{\tilde Q_{\eta_n}(\mu', t, s', a') - \tilde Q_{\eta_n}(\mu', t, s', a_{\mathrm{sub}})}{\eta} \right) - \exp \left( \frac{\tilde Q_{\eta_n}(\mu, t, s', a') - \tilde Q_{\eta_n}(\mu, t, s', a_{\mathrm{sub}})}{\eta} \right) \right| \\
    &\leq R_q^{\mathrm{max}} \sum_{a' \neq a_{\mathrm{sub}}} \left| \frac{1}{\eta} \exp \left( \frac{\xi_{a'}}{\eta} \right) \right| \left| (\tilde Q_{\eta_n}(\mu', t, s', a') - \tilde Q_{\eta_n}(\mu', t, s', a_{\mathrm{sub}})) - (\tilde Q_{\eta_n}(\mu, t, s', a') - \tilde Q_{\eta_n}(\mu, t, s', a_{\mathrm{sub}})) \right| \\
    &\leq R_q^{\mathrm{max}} |\mathcal A| \cdot \frac{1}{\eta_{\mathrm{min}}^{s'}} \exp \left( \frac{2 M_Q}{\eta_{\mathrm{min}}^{s'}} \right) \left( \left| \tilde Q_{\eta_n}(\mu', t, s', a') - \tilde Q_{\eta_n}(\mu, t, s', a') \right| + \left| \tilde Q_{\eta_n}(\mu, t, s', a_{\mathrm{sub}}) - \tilde Q_{\eta_n}(\mu', t, s', a_{\mathrm{sub}}) \right| \right) \\
    &\leq R_q^{\mathrm{max}} |\mathcal A| \cdot \frac{1}{\eta_{\mathrm{min}}^{s'}} \exp \left( \frac{2 M_Q}{\eta_{\mathrm{min}}^{s'}} \right) \cdot \left( 2 K_Q d_{\mathcal M}(\mu, \mu') + 4 \Delta_Q^{s'} \right) \\
    &\leq R_q^{\mathrm{max}} |\mathcal A| \cdot \frac{1}{\eta_{\mathrm{min}}^{s'}} \exp \left( \frac{2 M_Q}{\eta_{\mathrm{min}}^{s'}} \right) \cdot \left( 2 K_Q d_{\mathcal M}(\mu, \mu') \right) + \frac{\varepsilon_{t,s,a}}{16M_Q |\mathcal A|} 
    < \frac{\varepsilon_{t,s,a}}{8M_Q |\mathcal A|} 
\end{align*}
as in the original proof by letting $d_{\mathcal M}(\mu, \mu') < \delta_{t,s,a}^{4,s'}$ and choosing
\begin{align*}
    \delta_{t,s,a}^{4,s'} = \frac{\varepsilon_{t,s,a} \eta_{\mathrm{min}}^{s'}}{16M_Q |\mathcal A|^2 R_q^{\mathrm{max}} \cdot \exp \left( \frac{2 M_Q}{\eta_{\mathrm{min}}^{s'}} \right) \cdot 2 K_Q} \, .
\end{align*}

The rest of the proof is analogous. We obtain the additional requirement $n > N_{t,s,a}^{s'}$, $n > \tilde N_{t,s,a}^{s'}$ for some integers $N_{t,s,a}^{s'}, \tilde N_{t,s,a}^{s'}$ and each $t \in \mathcal T, s \in \mathcal S, s' \in \mathcal S, a \in \mathcal A$. By choosing $n' \equiv \max_{t \in \mathcal T, s \in \mathcal S, s' \in \mathcal S, a \in \mathcal A} \max(N_{t,s,a}^{s'}, \tilde N_{t,s,a}^{s'})$, the desired result holds as long as $n > n'$.
\end{subproof}

From this property, we again obtain the desired uniform convergence via compactness of $\mathcal M$.

\begin{lemma} \label{lem:uniformsoftq}
Any sequence of functions $(\mu \mapsto Q^{\tilde \Phi_{\eta_n}(\mu)}(\mu, t, s, a))_{n \in \mathbb N}$ with $\eta_n \to 0^+$ converges uniformly to $\mu \mapsto Q^*(\mu, t, s, a)$ for all $t \in \mathcal T, s \in \mathcal S, a \in \mathcal A$.
\end{lemma}
\begin{subproof}
Fix $\varepsilon > 0, t \in \mathcal T, s \in \mathcal S, a \in \mathcal A$. Then, there exists by Lemma~\ref{lem:relent-equi} for any point $\mu \in \mathcal M$ both $\delta(\mu)$ and $n'$ such that for all $\mu' \in \mathcal M$ with $d_{\mathcal M}(\mu, \mu') < \delta(\mu)$ for all $n > n'$ we have 
\begin{align*}
    \left| Q^{\tilde \Phi_{\eta_n}(\mu)}(\mu, t, s, a) - Q^{\tilde \Phi_{\eta_n}(\mu')}(\mu', t, s, a) \right| < \frac{\varepsilon}{3}
\end{align*} 
which via pointwise convergence from Lemma~\ref{lem:relent-pointwise} implies 
\begin{align*}
    \left| Q^*(\mu, t, s, a) - Q^*(\mu', t, s, a) \right| \leq \frac{\varepsilon}{3} \, .
\end{align*} 
Since $\mathcal M$ is compact, it is separable, i.e. there exists a countable dense subset $(\mu_j)_{j \in \mathbb N}$ of $\mathcal M$. Let $\delta(\mu)$ be as defined above and cover $\mathcal M$ by the open balls $(B_{\delta(\mu_j)}(\mu_j))_{j \in \mathbb N}$. By the compactness of $\mathcal M$, finitely many of these balls $B_{\delta(\mu_{n_1})}(\mu_{n_1}), \ldots, B_{\delta(\mu_{n_k})}(\mu_{n_k})$ cover $\mathcal M$. By pointwise convergence from Lemma~\ref{lem:relent-pointwise}, for any $i = 1, \ldots, k$ we can find integers $m_i$ such that for all $n > m_i$ we have 
\begin{align*}
    \left| Q^{\tilde \Phi_{\eta_n}(\mu_{n_i})}(\mu_{n_i}, t, s, a) - Q^*(\mu_{n_i}, t, s, a) \right| < \frac{\varepsilon}{3} \, .
\end{align*}
Taken together, we find that for $n > \max(n', \max_{i = 1, \ldots, k} m_i)$ and arbitrary $\mu \in \mathcal M$, we have
\begin{align*}
    \left| Q^{\tilde \Phi_{\eta_n}(\mu)}(\mu, t, s, a) - Q^*(\mu, t, s, a) \right| &< \left| Q^{\tilde \Phi_{\eta_n}(\mu)}(\mu, t, s, a) - Q^{\tilde \Phi_{\eta_n}(\mu_{n_i})}(\mu_{n_i}, t, s, a) \right| \\
    &\quad + \left| Q^{\tilde \Phi_{\eta_n}(\mu_{n_i})}(\mu_{n_i}, t, s, a) - Q^*(\mu_{n_i}, t, s, a) \right| \\
    &\quad + \left| Q^*(\mu_{n_i}, t, s, a) - Q^*(\mu, t, s, a) \right| \\
    &< \frac{\varepsilon}{3} + \frac{\varepsilon}{3} + \frac{\varepsilon}{3} = \varepsilon
\end{align*} 
for some center point $\mu_{n_i}$ of a ball containing $\mu$ from the finite cover.
\end{subproof}

As a result, a sequence of RelEnt MFE with $\eta \to 0^+$ is approximately optimal in the MFG.

\begin{lemma} \label{lem:relent-expl}
For any sequence $(\pi^*_{n}, \mu^*_{n})_{n \in \mathbb N}$ of ${\eta_n}$-RelEnt MFE with $\eta_n \to 0^+$ and for any $\varepsilon > 0$ there exists integer $n' \in \mathbb N$ such that for all integers $n > n'$ we have
\begin{align*}
    J^{\mu^*_{n}}(\pi^*_{n}) \geq \max_\pi J^{\mu^*_{n}}(\pi) - \varepsilon \ .
\end{align*}
\end{lemma}
\begin{subproof}
By Lemma~\ref{lem:uniformsoftq}, we have $\left| Q^{\tilde \Phi_{\eta_n}(\mu)}(\mu, t, s, a) - Q^*(\mu, t, s, a) \right| \to 0$ uniformly. Therefore, for any $\varepsilon > 0$, there exists by uniform convergence an integer $n'$ such that for all integers $n > n'$ we have
\begin{align*}
    Q^{\pi^*_{n}}(\mu^*_{n}, t,s,a) \geq Q^*(\mu^*_{n}, t,s,a) - \varepsilon = \max_{\pi \in \Pi} Q^{\pi}(\mu^*_{n}, t,s,a) - \varepsilon \, ,
\end{align*}
and since by Lemma~\ref{lem:policyoptimality}, we have
\begin{align*}
    J^{\mu^*_{n}}(\pi^*_{n}) = \sum_{s \in \mathcal S} \mu_0(s) \cdot \sum_{a \in \mathcal A} Q^{\pi^*_{n}}(\mu^*_{n}, t,s,a) \geq \sum_{s \in \mathcal S} \mu_0(s) \cdot \max_{\pi \in \Pi} \sum_{a \in \mathcal A} Q^{\pi}(\mu^*_{n}, t,s,a) - \varepsilon = \max_{\pi \in \Pi} J^{\mu^*_{n}}(\pi) - \varepsilon \, ,
\end{align*}
the desired result follows immediately.
\end{subproof}

By repeating the previous argumentation for Boltzmann MFE with Lemma~\ref{lem:obj-conv} and replacing Lemma~\ref{lem:boltzmann-expl} with Lemma~\ref{lem:relent-expl}, we obtain the desired result for RelEnt MFE.
\end{proof}

\section{Relative entropy mean field games} \label{app:relent}
We show that the necessary conditions for optimality hold for the candidate solution. (For further insight, see also \citet{neu2017unified}, \citet{haarnoja2017reinforcement} and references therein.) Fix a mean field $\mu \in \mathcal M$ and formulate the induced problem as an optimization problem, with $\rho_t(s)$ as the probability of our representative agent visiting state $s \in \mathcal S$ at time $t \in \mathcal T$, to obtain 
\begin{subequations}
    \begin{alignat*}{2}
    &\! \max_{\rho, \pi} &\qquad& \sum_{t=0}^{T-1} \sum_{s \in \mathcal S} \rho_t(s) \sum_{a \in \mathcal A} \pi_t(a \mid s) r(s, a, \mu_t) \\
    &\text{subject to} & & \rho_{t+1}(s') = \sum_{s \in \mathcal S} \rho_t(s) \sum_{a \in \mathcal A} \pi_t(a \mid s) p(s'\mid s, a, \mu_t) \qquad \forall s' \in \mathcal S, t \in \{ 0,\ldots,T-2 \},  \\
    &                  & & 1 = \sum_{s \in \mathcal S} \rho_t(s) \qquad \forall t \in \{ 0,\ldots,T-1 \}, \\
    &                  & & 1 = \sum_{a \in \mathcal A} \pi_t(a \mid s) \qquad \forall s \in \mathcal S, t \in \{ 0,\ldots,T-1 \}, \\
    &                  & & 0 \leq \rho_t(s), 0 \leq \pi_t(a \mid s) \qquad \forall s \in \mathcal S, a \in \mathcal A, t \in \{ 0,\ldots,T-1 \}, \\
    &                  & & \mu_0(s) = \rho_0(s) \qquad \forall s \in \mathcal S.  
    \end{alignat*}
\end{subequations}
Note that if the agent follows the mean field policy of the other agents, we have $\rho_t = \mu_t$. The optimized objective is just the expectation $\mathbb E \left[ \sum_{t=0}^{T-1} r(S_{t}, A_{t}) \right]$. As in \citet{belousov2019entropic}, we change this objective to include a KL-divergence penalty weighted by the state-visitation distribution $\rho_t(\cdot)$ by introducing the temperature $\eta > 0$ and prior policy $q \in \Pi$ to obtain
\begin{subequations}
    \begin{alignat*}{2}
    &\! \max_{\rho_t, \pi_t} &\qquad& \sum_{t=0}^{T-1} \sum_{s \in \mathcal S} \rho_t(s) \sum_{a \in \mathcal A} \pi_t(a \mid s) r(s, a, \mu_t) - \eta \sum_{t=0}^{T-1} \sum_{s \in \mathcal S} \rho_t(s) \KL{\pi_t(\cdot \mid s)}{q_t(\cdot \mid s)} \\
    &\text{subject to} & & \rho_{t+1}(s') = \sum_{s \in \mathcal S} \rho_t(s) \sum_{a \in \mathcal A} \pi_t(a \mid s) p(s'\mid s, a, \mu_t) \qquad \forall s' \in \mathcal S, t \in \{ 0,\ldots,T-2 \},  \\
    &                  & & 1 = \sum_{s \in \mathcal S} \rho_t(s) \qquad \forall t \in \{ 0,\ldots,T-1 \}, \\
    &                  & & 1 = \sum_{a \in \mathcal A} \pi_t(a \mid s) \qquad \forall s \in \mathcal S, t \in \{ 0,\ldots,T-1 \}, \\
    &                  & & 0 \leq \rho_t(s), 0 \leq \pi_t(a \mid s) \qquad \forall s \in \mathcal S, a \in \mathcal A, t \in \{ 0,\ldots,T-1 \}, \\
    &                  & & \mu_0(s) = \rho_0(s) \qquad \forall s \in \mathcal S. 
    \end{alignat*}
\end{subequations}

We ignore the constraints $0 \leq \pi_t(a \mid s)$ and $0 \leq \rho_t(s)$ and see later that they will hold automatically. This results in the simplified optimization problem
\begin{subequations}
    \begin{alignat*}{2}
    &\! \max_{\rho_t, \pi_t} &\qquad& \sum_{t=0}^{T-1} \sum_{s \in \mathcal S} \rho_t(s) \sum_{a \in \mathcal A} \pi_t(a \mid s) r(s, a, \mu_t) - \eta \sum_{t=0}^{T-1} \sum_{s \in \mathcal S} \rho_t(s) \KL{\pi_t(\cdot \mid s)}{q_t(\cdot \mid s)} \\
    &\text{subject to} & & \rho_{t+1}(s') = \sum_{s \in \mathcal S} \rho_t(s) \sum_{a \in \mathcal A} \pi_t(a \mid s) p(s'\mid s, a, \mu_t) \qquad \forall s' \in \mathcal S, t \in \{ 0,\ldots,T-2 \},  \\
    &                  & & 1 = \sum_{s \in \mathcal S} \rho_t(s) \qquad \forall t \in \{ 0,\ldots,T-1 \}, \\
    &                  & & 1 = \sum_{a \in \mathcal A} \pi_t(a \mid s) \qquad \forall s \in \mathcal S, t \in \{ 0,\ldots,T-1 \}, \\
    &                  & & \mu_0(s) = \rho_0(s) \qquad \forall s \in \mathcal S,
    \end{alignat*}
\end{subequations}
for which we introduce Lagrange multipliers $\lambda_1(t,s)$, $\lambda_2(t)$, $\lambda_3(t,s)$, $\lambda_4(s)$ and the Lagrangian 
\begin{align*}
    L(\rho, \pi, \lambda_1, \lambda_2, \lambda_3, \lambda_4) &= \sum_{t=0}^{T-1} \sum_{s \in \mathcal S} \rho_t(s) \sum_{a \in \mathcal A} \pi_t(a \mid s) \left( r(s, a, \mu_t) - \eta \log \frac{\pi_t(a \mid s)}{q_t(a \mid s)} \right) \\
    &- \sum_{t=0}^{T-1} \sum_{s' \in \mathcal S} \lambda_1(t,s') \left( \rho_{t+1}(s') - \sum_{s \in \mathcal S} \rho_t(s) \sum_{a \in \mathcal A} \pi_t(a \mid s) p(s'\mid s, a, \mu_t) \right) \\
    &- \sum_{t=0}^{T-1} \lambda_2(t) \left( 1 - \sum_{s \in \mathcal S} \rho_t(s) \right) \\
    &- \sum_{t=0}^{T-1} \sum_{s \in \mathcal S} \lambda_3(t,s) \left( \sum_{a \in \mathcal A} \pi_t(a \mid s) - 1 \right) \\
    &- \sum_{s \in \mathcal S} \lambda_4(s) \left( \mu_0(s) - \rho_0(s) \right)
\end{align*}
with the artificial constraint $\lambda_1(T-1,s) \equiv 0$, which allows us to formulate the following necessary conditions for optimality. For $\nabla_{\pi_t(a \mid s)}L$ and all $s \in \mathcal S, a \in \mathcal A, t \in \{ 0,\ldots,T-1 \}$, we obtain
\begin{align*}
    &\nabla_{\pi_t(a \mid s)} L = \rho_t(s) \left( r(s,a,\mu_t) - \eta \log \frac{\pi_t(a \mid s)}{q_t(a \mid s)} - \eta + \sum_{s' \in \mathcal S} \lambda_1(t,s') p(s'\mid s, a, \mu_t) \right) - \lambda_3(t,s) \stackrel{!}{=} 0 \\
    \implies &\pi_t^*(a \mid s) = q_t(a \mid s) \exp \left( \frac{r(s,a,\mu_t) - \eta + \sum_{s' \in \mathcal S} \lambda_1(t,s') p(s'\mid s, a, \mu_t) - \frac{\lambda_3(t,s)}{\rho_t(s)}}{\eta} \right) \, .
\end{align*}
For $\nabla_{\lambda_3}L$ and all $s \in \mathcal S, t \in \{ 0,\ldots,T-1 \}$, by inserting $\pi_t^*$ we obtain
\begin{align*}
    &\nabla_{\lambda_3(t,s)} L = 1 - \sum_{a \in \mathcal A} \pi_t(a \mid s) \stackrel{!}{=} 0 \\
    \iff &1 = \sum_{a \in \mathcal A} q_t(a \mid s) \exp \left( \frac{r(s,a,\mu_t) - \eta + \sum_{s' \in \mathcal S} \lambda_1(t,s') p(s'\mid s, a, \mu_t) - \frac{\lambda_3(t,s)}{\rho_t(s)}}{\eta} \right)
\end{align*} 
which is fulfilled by choosing
\begin{align*}
    \lambda_3^*(t,s) = \eta \rho_t(s) \log \sum_{a \in \mathcal A} q_t(a \mid s) \exp \left( \frac{r(s,a,\mu_t) - \eta + \sum_{s' \in \mathcal S} \lambda_1(t,s') p(s'\mid s, a, \mu_t)}{\eta} \right)
\end{align*}
since it fulfills the required equation
\begin{align*}
    &\quad \sum_{a \in \mathcal A} q_t(a \mid s) \exp \left( \frac{r(s,a,\mu_t) - \eta + \sum_{s' \in \mathcal S} \lambda_1(t,s') p(s'\mid s, a, \mu_t) - \frac{\lambda_3^*(t,s)}{\rho_t(s)}}{\eta} \right) \\
    &= \sum_{a \in \mathcal A} q_t(a \mid s) \exp \left( \frac{r(s,a,\mu_t) - \eta + \sum_{s' \in \mathcal S} \lambda_1(t,s') p(s'\mid s, a, \mu_t)}{\eta} \right) \\
    &\quad \cdot \left( \sum_{a \in \mathcal A} q_t(a \mid s) \exp \left( \frac{r(s,a,\mu_t) - \eta + \sum_{s' \in \mathcal S} \lambda_1(t,s') p(s'\mid s, a, \mu_t)}{\eta} \right) \right)^{-1} = 1 \ .
\end{align*}
Finally, inserting $\lambda_3^*$ and $\pi^*$, for $\nabla_{\rho_t(s)}L$ we obtain 
\begin{align*}
    \nabla_{\rho_t(s)} L 
    &= \sum_{a \in \mathcal A} \pi_t(a \mid s) \left( r(s, a, \mu_t) - \eta \log \frac{\pi_t(a \mid s)}{q_t(a \mid s)} + \sum_{s' \in \mathcal S} \lambda_1(t,s') p(s'\mid s, a, \mu_t) + \lambda_2(t) \right) - \lambda_1(t-1,s) \\
    &= \sum_{a \in \mathcal A} \pi_t(a \mid s) \left( \eta + \lambda_2(t) + \frac{\lambda_3(t,s)}{\rho_t(s)} \right) - \lambda_1(t-1,s) \stackrel{!}{=} 0
\end{align*}
which implies
\begin{align*}
    \lambda_1^*(t-1,s) &= \eta + \lambda_2(t) + \eta \log \sum_{a \in \mathcal A} q_t(a \mid s) \exp \left( \frac{r(s,a,\mu_t) - \eta + \sum_{s' \in \mathcal S} \lambda_1(t,s') p(s'\mid s, a, \mu_t)}{\eta} \right) \, .
\end{align*}
We can subtract $\lambda_2(t)$ and shift the time index to obtain the soft value function $\tilde V_\eta(\mu,t,s)$ defined via terminal condition $\tilde V_\eta(\mu,T,s) \equiv 0$ and the recursion
\begin{align*} 
    \tilde V_\eta(\mu,t,s) &= \eta \log \sum_{a \in \mathcal A} q_t(a \mid s) \exp \left( \frac{r(s,a,\mu_t) + \sum_{s' \in \mathcal S} \tilde V_\eta(\mu,t+1,s') p(s'\mid s, a, \mu_t)}{\eta} \right)
\end{align*}
since then, by normalization the optimal policy for all $s \in \mathcal S, a \in \mathcal A, t \in \{ 0,\ldots,T-1 \}$ is equivalent to
\begin{align*}
    \pi_t^*(a \mid s) &= \frac{ q_t(a \mid s) \exp \left( \frac{r(s,a,\mu_t) + \sum_{s' \in \mathcal S} \lambda_1(t,s') p(s'\mid s, a, \mu_t)}{\eta} \right) }{ \sum_{a' \in \mathcal A} q_t(a' \mid s) \exp \left( \frac{r(s,a',\mu_t) + \sum_{s' \in \mathcal S} \lambda_1(t,s') p(s'\mid s, a', \mu_t)}{\eta} \right) } \\
    &= \frac{ q_t(a \mid s) \exp \left( \frac{r(s,a,\mu_t) + \sum_{s' \in \mathcal S} \tilde V_\eta(\mu,t+1,s') p(s'\mid s, a, \mu_t)}{\eta} \right) }{ \sum_{a' \in \mathcal A} q_t(a' \mid s) \exp \left( \frac{r(s,a',\mu_t) + \sum_{s' \in \mathcal S} \tilde V_\eta(\mu,t+1,s') p(s'\mid s, a', \mu_t)}{\eta} \right) } \, .
\end{align*}
To obtain a recursion in $\tilde Q_\eta$, define
\begin{align*}
    \tilde Q_\eta(\mu, t, s, a) &\equiv r(s, a, \mu_t) + \sum_{s' \in \mathcal S} p(s'\mid s, a, \mu_t) \eta \log \sum_{a' \in \mathcal A} q_{t+1}(a' \mid s') \exp \left( \frac{\tilde Q_\eta(\mu, t+1, s', a')}{\eta} \right)
\end{align*} 
with terminal condition $\tilde Q_\eta(\mu, T, s, a) \equiv 0$ to obtain
\begin{align*}
    \pi_t^*(a \mid s) &= \frac{ q_t(a \mid s) \exp \left( \frac{\tilde Q_\eta(\mu, t, s, a)}{\eta} \right) }{ \sum_{a' \in \mathcal A} q_t(a' \mid s) \exp \left( \frac{\tilde Q_\eta(\mu, t, s, a')}{\eta} \right) }
\end{align*}
which is the desired result as $\pi^*$ fulfills all constraints and determines $\rho$ uniquely. For the uniform prior $q_t(a \mid s) = 1/|\mathcal A|$, we obtain the maximum entropy solution.

\end{document}